\documentclass[aps,pra,twocolumn,superscriptaddress,floatfix,nofootinbib,showpacs,longbibliography,groupedaddress]{revtex4-2}
\usepackage{times}
\usepackage[table]{xcolor}
\usepackage[scaled=0.86]{berasans}  
\usepackage[colorlinks=true, citecolor=blue, urlcolor=blue]{hyperref}
\usepackage{graphicx} 
\usepackage[babel]{microtype}  
\usepackage{amsmath,amssymb,amsthm,bm,amsfonts,mathrsfs,bbm} 

\usepackage{xspace}  
\usepackage{pgf,tikz}
\usepackage{xcolor}
\usepackage{multirow}
\usepackage{array}
\usepackage{bigstrut}
\usepackage{braket}
\usepackage{color}
\usepackage{natbib}
\usepackage{multirow}
\usepackage{mathtools}
\usepackage{float}
\usepackage[caption = false]{subfig}
\usepackage{xcolor,colortbl}
\usepackage{color}
\usepackage{centernot}
\usepackage{subfig}

\newcommand{\setword}[2]{%
	\phantomsection
	#1\def\@currentlabel{\unexpanded{#1}}\label{#2}%
}
\newcommand{\Tr}{\operatorname{Tr}}

\newcommand{\be}{\begin{equation}}
	\newcommand{\ee}{\end{equation}}
\newcommand{\ba}{\begin{eqnarray}}
	\newcommand{\ea}{\end{eqnarray}}

\newcommand{\tr}{\operatorname{Tr}}
\newcommand{\proj}[1]{\ket{#1}\bra{#1}}

\newtheorem{theorem}{Theorem}

\newtheorem{definition}{Definition}
\newtheorem{proposition}{Proposition}

\newtheorem{remark}{Remark}

\def\>{\rangle}
\def\<{\langle}

\begin{document}

\title{Experimental Verification of Many-Body Entanglement Using Thermodynamic Quantities}

\author{Jitendra Joshi}
\thanks{These authors contribute equally in this project and share equal credit of the first authorship.}
\affiliation{Department of Physics and NMR Research Center,\\ Indian Institute of Science Education and Research, Pune 411008, India}

\author{Mir Alimuddin}
\thanks{These authors contribute equally in this project and share equal credit of the first authorship.}
\affiliation{Department of Physics of Complex Systems, S.N. Bose National Center for Basic Sciences, Block JD, Sector III, Salt Lake, Kolkata 700106, India.}

\author{T S Mahesh}
\affiliation{Department of Physics and NMR Research Center,\\ Indian Institute of Science Education and Research, Pune 411008, India}

\author{Manik Banik}
\affiliation{Department of Physics of Complex Systems, S.N. Bose National Center for Basic Sciences, Block JD, Sector III, Salt Lake, Kolkata 700106, India.}

\begin{abstract}
The phenomenon of quantum entanglement underlies several important protocols that enable emerging quantum technologies. Entangled states, however, are extremely delicate and often get perturbed by tiny fluctuations in their external environment. Certification of entanglement is therefore immensely crucial for the successful implementation of protocols involving this resource. In this work, we propose a set of entanglement criteria for multi-qubit systems that can be easily verified by measuring certain thermodynamic quantities. In particular, the criteria depend on the difference in optimal global and local works extractable from an isolated quantum system under global and local interactions, respectively. As a proof of principle, we demonstrate the proposed scheme on nuclear spin registers of up to 10 qubits using the Nuclear Magnetic Resonance architecture. We prepare noisy Bell-diagonal state and noisy Greenberger–Horne–Zeilinger class of states in star-topology systems and certify their entanglement through our thermodynamic criteria. Along the same line, we also propose an entanglement certification scheme in many-body systems when only partial or even no knowledge about the state is available.
\end{abstract}


\maketitle
{\it Introduction.--} Quantum entanglement, identified as a puzzling feature of multipartite quantum systems \cite{Einstein1935,Bohr1935,Schrdinger1936,Bell1966}, plays the pivotal role in a number of important quantum information protocols \cite{Ekert1991,Bennett1992,Bennett1993,Deutsch1985,Yurke1986,Giovannetti2004,Xia2023} (see also \cite{Horodecki2009}). In quantum systems involving more than two parts entanglement appears in different inequivalent and exotic forms \cite{Dur2000,Verstraete2002}, that have been proved to be useful in several distributed protocols \cite{Cirac1997,Raussendorf2001,Arrighi2006,Kimble2008,Wehner2018,Agrawal2019,Saha2020,Rout2021,Banik2021,Bhattacharya2021}. However, entangled states are fragile and easily lost by external perturbations. Successful implementation of the protocols involving entanglement, therefore, demands faithful certification of entanglement. Although the generic {\it separability problem} is known to be extremely hard even for bipartite systems \cite{Gurvits2003}, negative-partial-transposition (NPT) criterion \cite{Peres1996,Horodecki1996} and sometimes measurement of entanglement witness operator \cite{Dirkse2020} become useful for certifying  entanglement. On the other hand, there exist entropic quantities that also serve the purpose of entanglement certification \cite{Bennett1996,Devetak2005}. However, these entropic quantities are not directly measurable in experiments, and calculating the value of witness operator \& evaluating NPT-ness of a state demands complete tomographic knowledge which is practically impossible when large number of subsystems are involved. 

During the recent past, in a completely different approach, researchers are trying to identify operationally motivated thermodynamic quantities that can capture the signature of entanglement in multipartite quantum systems \cite{Perarnau2015,Mukherjee16,Alimuddin2019,Alimuddin2020(1),Alimuddin2020(2),Puliyil2022,Yang'23}. In this work we show that suitably defined functions of such a thermodynamic quantity, namely the ergotropic work, can serve as bona-fide entanglement certifiers for generic $N$-qubit systems. The optimal amount of work extractable from an isolated quantum system by keeping its entropy unchanged is known as ergotropic work \cite{Allahverdyan2004}. Depending upon whether a many-body quantum system is addressed globally or its parts are addressed separately, different kinds of ergotropic works can be extracted. Interestingly, entanglement of the initially prepared multipartite state keeps it footprints in the difference of these global and local ergotropic works. Furthermore, while extracting work one might infer the spectral of the state in question. Depending on the available information about the spectral of the global state and its marginals we propose several entanglement certifiers. As proof of principle, we implement the proposed thermodynamic entanglement criterion on nuclear spin registers of up to $10$ qubits via  Nuclear Magnetic Resonance (NMR) architecture.  In particular, the star-topology systems allow preparation of Greenberger–Horne–Zeilinger (GHZ) class of states in large registers \cite{Jones2009,mahesh2021star,shukla2014noon}. We prepare two-qubit Bell diagonal state and noisy states comprising of singlet/GHZ state and white noise, and certify their entanglement through our proposed criteria. 

{\it Theory and Framework.--} State of an $N$-qubit system is described by a density operator $\rho_{A_1\cdots A_N}\in\mathcal{D}\left((\mathbb{C}^2)^{\otimes N}\right)$; where $\mathcal{D}(\mathbb{H})$ denote the set of positive trace-one operators acting on the Hilbert space $\mathbb{H}$. A state is called fully separable if it is a probabilistic mixture of fully product state, {\it i.e.}, $\rho_{A_1\cdots A_N}=\sum_ip_i\left(\bigotimes_{j=1}^N\ket{\psi^i_{A_j}}\bra{\psi^i_{A_j}}\right)$, with $\ket{\psi^i_{A_j}}\in\mathbb{C}^2_{A_j}\equiv\mathbb{C}^2$. States lying outside the set of fully separable states are entangled. However, different kinds of entanglement are possible in multi-qubit systems. Let $\mathcal{S}[X|X^c]$ denotes the set of states separable across $X$-vs-$X^c$ bipartite cut, where $X$ contain $\kappa$ parties together and $X^c$ contains the remaining parties; $\kappa\in\{1,\cdots N-1\}$. States lying outside $\mathcal{S}[X|X^c]$ contains entanglement across $X$-vs-$X^c$ bipartition. 

When an isolated such system evolves from an initial state $\rho$ to a lower energy state $\sigma$, the difference in energies can be extracted as work. Study of this topic dates back to late seventies \cite{Pusz1978,Lenard1978} and it gains renewed interest in the recent past \cite{Aberg2013,Skrzypczyk2014,Skrzypczyk2015,PhysRevA.106.042601}. Consider an $N$-qubit system governed by non-interacting Hamiltonian $H=\sum_{l=1}^{N}\tilde{H}_{l}$, where $\tilde{H}_{l}:=\mathbf{I}_1\otimes\cdots\otimes\mathbf{I}_{l-1}\otimes H_l\otimes\mathbf{I}_{l+1}\otimes\cdots\otimes\mathbf{I}_N$, with  $H_l=\sum^{1}_{i=0}(E_l+i\alpha_l)|i\rangle\langle i|$ and $\ket{0}$ \& $\ket{1}$ being the energy eigenkets with respective eigenvalues $E_l$ and $E_l+\alpha_l$; $\mathbf{I}_j$ be the identity operator on $j^{th}$ qubit. Evolution from the initial state to final state is governed through a cyclic unitary $U(\tau)$ generated by switching on a time dependent interaction. 

The optimally extractable work, called ergotropy, amounts to $W(\rho)=\tr\left[\rho H\right]-\min_{U(\tau)} \tr\left[U(\tau)\rho U^{\dagger}(\tau)H\right]$, where optimization is considered over all unitaries. As it turns out during optimal work extraction the system evolves to the passive state $\rho^P$, and accordingly we have, $W(\rho):=E(\rho)-E(\rho^P)=\tr\left[\rho H\right]-\tr[\rho^PH]$ \cite{Pusz1978,Lenard1978}. Passive state is the lowest energetic state with spectral identical to the initial state. Moreover, it is diagonal in energy basis where higher energy states are lessly populated. In multipartite scenario different parts of the system can be probed separately leading to several inequivalent configuration for work extraction. For instance, in the $X$-vs-$X^c$ configuration, with $X$ containing $\kappa$ parties together, the optimal extractable work from $X$ subsystem is given by, $W_{[\kappa]}(\rho_X):=\tr[\rho_X H_X]-\tr[\rho^P_X H_X]$, where $\rho_X:=\tr_{X^c}(\rho)\in\mathcal{D}((\mathbb{C}^2)^{\otimes \kappa}),~\rho\in\mathcal{D}((\mathbb{C}^2)^{\otimes N}),~H_X$ is the Hamiltonian of the subsystem $X$, $\rho_X^P$ is the passive state corresponding to $\rho_X$, with $W_{[N]}(\rho)$ simply denoted as $W(\rho)$. 

We will denote the spectral for a generic $N$-qubit state $\rho$ as $\vec{t}_\rho\equiv\{t_j\}_{j=0}^{2^N-1}$, arranged in decreasing order. System's Hamiltonian $H$ can be re-expressed as $H=\sum_{j=0}^{2^N-1}(E_g+n_j)|e_j\rangle\langle e_j|$, where $\ket{e_0}=\ket{0}^{\otimes N}$ is the ground state with energy value $E_g=\sum_{l=1}^NE_l$, and the energy eigenvalues are arranged in increasing order, {\it i.e.} $n_{j+1}\ge n_j,~\forall~j$, with $n_0=0$. The highest exited state $\ket{e_{2^N-1}}=\ket{1}^{\otimes N}$ has energy value $E_g+\sum_{l=1}^N\alpha_l$. Spectral of the subsystem $X$ will be denoted as $\vec{x}_{\rho_X}\equiv\{x_j\}_{j=0}^{2^\kappa-1}$,with its Hamiltonian re-expressed as $H_X=\sum_{j=0}^{2^\kappa-1}(E^X_g+m_j)|f_j\rangle\langle f_j|$. While extracting work in $X$-vs-$X^c$ configuration, we can evaluate the thermodynamic quantity
\begin{align}
\Delta_{X|X^c}:=W(\rho)-W_{[\kappa]}(\rho_X)-E(\rho_{X^c})+E^{X^c}_g.\label{delta}   
\end{align}
Here the first three terms are state dependent and their values can be evaluated through experiment; the last term designates the ground state energy of the Hamiltonian of the $X^c$ part. We are now in a position to provide our thermodynamic entanglement criteria (proof defer to Supplemental section \cite{Supple}).  
\begin{theorem}\label{theo1}
An $N$-qubit state separable across $X$-vs-$X^c$ bipartition satisfies
\footnotesize
\begin{subequations}
\begin{align}
\Delta_{X|X^c}&\le \sum_{i=1}^{2^\kappa-1}(m_i-m_1)x_i+\sum_{i=1}^{2^N-1}(m_1-n_i)t_i:=\delta^{GL}_{X|X^c},\label{ecs}\\
\Delta_{X|X^c}&\le\sum_{i=1}^{2^\kappa-2}(m_i-n_i)t_i+\sum_{i=2^\kappa-1}^{2^N-1}(m_{2^\kappa-1}-n_i)t_i:=\delta^{G}_{X|X^c},\label{ecw}
\end{align}
\end{subequations}
\normalsize
\mbox{where}, $m_{i+1}\ge m_i$ \mbox{for} $i\in\{0,2^\kappa-1\}$, $n_{i+1}\ge n_i$ \mbox{for} $i\in\{0,2^N-1\}$, \mbox{and} $m_{2^\kappa-1}=\sum_{i=1}^\kappa\alpha_i$. 
\end{theorem}
Violation of any of the conditions in Theorem \ref{theo1} certifies  entanglement across $X-$vs$-X^c$ bipartition. Please note that, examining condition (\ref{ecs}) requires knowledge of the global and local spectrals $\vec{t}_\rho$ and $\vec{x}_{\rho_X}$. Quite interestingly this thermodynamic criterion turns out to be a special case of Nielsen-Kempe entanglement criterion \cite{Nielsen2001} (see Remark {\bf 1} in \cite{Supple}). The separability bound $\delta^{G}_{X|X^c}$ in (\ref{ecw}) depends only on the global spectral of the given state (and hence the superscript `G') and generally turns out to be a weaker than (\ref{ecs}). One can also come up with a thermodynamic entanglement criterion that does not involve the knowledge about the state, albeit it will be weaker than the state dependent criteria  (see Remark {\bf 2} in \cite{Supple}). The state independent separability bound will be denoted as $\delta^I_{X|X^c}$. For instance, for the identical $3$-qubit noisy GHZ states $\rho_\lambda[3]:=(1-\lambda)\mathbb{I}/8+\lambda\ket{\psi_3}\bra{\psi_3}$, where $\ket{\psi_3}:=(\ket{000}+\ket{111})/\sqrt{2}$ and all the subsystems are governed through identical Hamiltonian, criterion (\ref{ecw}) can detect entanglement for $\lambda>3/7$, whereas a state independent criterion  detects entanglement for $\lambda>0.66$. In remaining part of the manuscript we investigate the aforesaid thermodynamic entanglement criteria by preparing specific classes of entangled states in NMR setup.
\begin{figure}[t!]
\centering
\includegraphics[width=8.5cm]{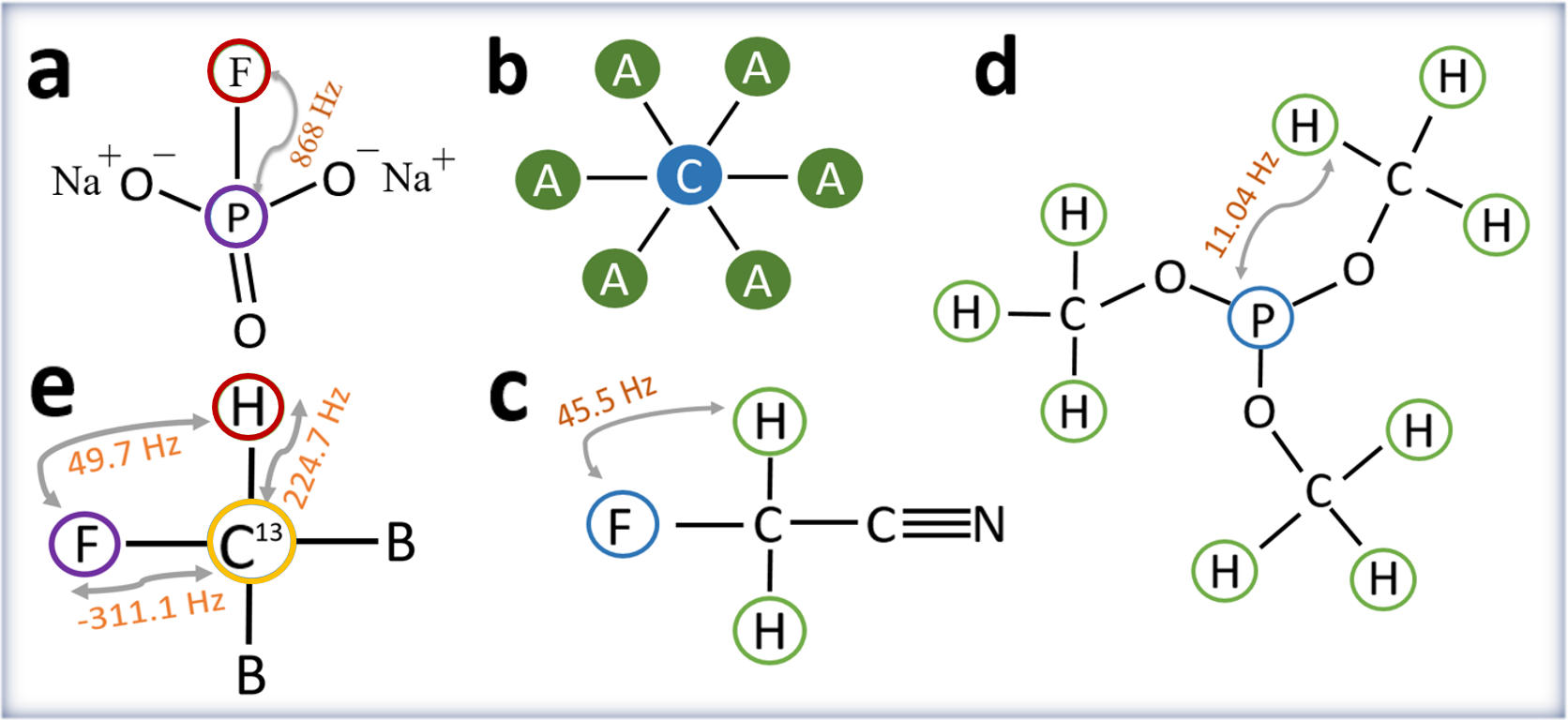}
\caption{(Color Online) (a) $2$-qubit system sodium fluoro-phosphate (NAFP, with $^{19}$F and $^{31}$P being first and second qubits) used in Experiment I. (b) The star topology configuration, wherein each ancillary spin A interacts with the central spin $^{13}$C. (c,d)  $3$-qubit star-system fluoroacetonitril (FAN) and $10$-qubit star-system  trimethyl-phosphate (TMP) used in  Experiment II. (e) Di-bromo fluoromethane (DBFM, with $^{1}$H, $^{13}$C, and $^{19}$F being first, second, and third qubits) used in Experiment III.}\label{fig1}
\vspace{-.5cm}
\end{figure}

{\it Experiment I: (Two-qubit Bell diagonal states).--} In our first experiment we deal with two-qubit Bell diagonal states prepared using the system NAFP (Fig. \ref{fig1} (a)) dissolved in D$_2$O. All the experiments were carried out on a $500$ MHz Bruker NMR spectrometer at an ambient temperature of $300$ K. Hamiltonian of the system consisting of the internal part and the RF drive reads as $H_{12} = H_{12}^{\mathrm{int}} +H_{12}^{\mathrm{RF}}$, where $H_{12}^{\mathrm{int}} =  - \omega_F I_{1z} - \omega_P I_{2z} + 2\pi J I_{1z}I_{2z} $ and $H_{12}^{\mathrm{RF}}= \Omega_F(t) I_{1x}+\Omega_P(t) I_{2x}$, with $I_{ix}=\sigma_{ix}/2,I_{iy}=\sigma_{iy}/2,I_{iz}=\sigma_{iz}/2$. Here $(\omega_F$, $\omega_P)$ and ($\Omega_F$, $\Omega_P$) respectively denote the Larmor frequencies and RF amplitudes of ($^{19}$F, $^{31}$P), $\hbar = 1$, and $J$ is the scalar coupling constant. We prepare Bell diagonal states with two independent controllable parameters $\beta$ and $\gamma$.  We start with the thermal state, which under high-field, high-temperature approximation 
reads as $\rho_{th} = \mathbf{I}/4 + \epsilon_P (\frac{\gamma_F}{\gamma_P}I_{1z} + I_{2z}$), where  $\epsilon_P = \gamma_PB_0/ 4k_BT$. The method of spatial averaging yields us psuedo-pure state (PPS) $\ket{11}\bra{11}^{pps} = (1-\epsilon)\mathrm{I}/4 + \epsilon\ket{11}\bra{11}$, with $\epsilon << 1$ (Fig.\ref{fig2} (a)). Within the paradigm of PPS we set $\epsilon = 1$ and realize the effective state $\ket{11}\bra{11}$ \cite{cory1997ensemble}. Subsequently we prepare the two-parameter Bell diagonal state
\begin{align}
&~~~~~~~~~~~~\rho_{12}= \sum_{i,j=0}^{1}p_{ij}\ket{\mathcal{B}_{ij}}\bra{\mathcal{B}_{ij}},\\
\ket{\mathcal{B}_{0j}}&:=(\ket{00}+(-1)^j\ket{11})/\sqrt{2},~p_{00}:=[\mathrm{S}(\beta/2)\mathrm{S}(\gamma/2)]^2,\nonumber\\
\ket{\mathcal{B}_{1j}}&:=(\ket{01}+(-1)^j\ket{10})/\sqrt{2},~p_{01}:= [\mathrm{S}(\beta/2)\mathrm{C}(\gamma/2)]^2,\nonumber\\
p_{10}&:=[\mathrm{C}(\beta/2)\mathrm{S}(\gamma/2)]^2,~~~~p_{11}:= [\mathrm{C}(\beta/2)\mathrm{C}(\gamma/2)]^2,\nonumber
\end{align} 
where $\mathrm{S}(\star):=\sin(\star)$ and $\mathrm{C}(\star):=\cos(\star)$. 
Further details on the preparation circuit are provided in \cite{Supple}.
To evaluate the quantity in Eq.(\ref{delta}), we evolve the state $\rho_{12}$ into its passive state, the lowest energetic state. While this requires an optimization over all possible unitary operations, in Supplemental material we argue that for generic two-qubit Bell-diagonal state this can be achieved by considering only $24$ permutation operations ($U_\mathrm{per}$). This way we obtain the value of the quantity $\Delta^{Expt}_{1|2}$ experimentally. Notably as the number of qubits increases, the optimization over the set of possible unitary operations expands significantly, and hence the scalability issues persists. Nonetheless, for two-qubit case, arranging $\{p_{ij}\}$ in descending order and denoting the resulting vector as $\vec{t}\equiv\{t_k\}_{k=0}^3$, theoretically we have 
\begin{align}
\left\{\!\begin{aligned}
\Delta_{1|2}= (1.162 - 2.324t_2-3.324t_3-t_1)\omega_P,\\
\delta^G_{1|2}= (1.324t_1-t_3)\omega_P,~~~     \delta^I_{1|2}= 0.662\omega_P
\end{aligned}\right\}.\label{exp1}
\end{align}
Note that, evaluation of the quantity $\delta^G_{1|2}$ in Eq.(\ref{exp1}) demands knowledge of the global spectral, whereas $\delta^I_{1|2}$ is state independent. Entanglement is certified whenever $\Delta_{1|2}$ is strictly greater than any one of these quantities. Varying the parameter $\beta$ and $\gamma$ we show the entanglement certification in Fig.\ref{fig2}. through `gradient-color-plot'. As expected and also evident from the plot, the state independent certification scheme turns out to be weaker than the state dependent scheme. For instance, the specific values of $\beta=2\pi/5$ and $\gamma=3\pi/10$ yield $\Delta_{1|2}=0.338\omega_P$, which is strictly less than $\delta^I_{1|2}=0.662\omega_P$, but greater than $\delta^G_{1|2}=0.292\omega_P$. This is not visible in Fig.\ref{fig2} (b)-(c) due to limited pixel resolution, but can be seen in Fig.\ref{fig2} (d). Important to note that, our entanglement certification scheme does not require tomographic knowledge of the state, rather it is obtained by evaluating expected energies of the given state and the unitarily evolved state. More specifically, our thermodynamic entanglement criteria can certify entanglement in the given state without requiring the information of the population frequencies for different energy eigenstates.         
\begin{figure}
\centering
\includegraphics[width=8.5cm]{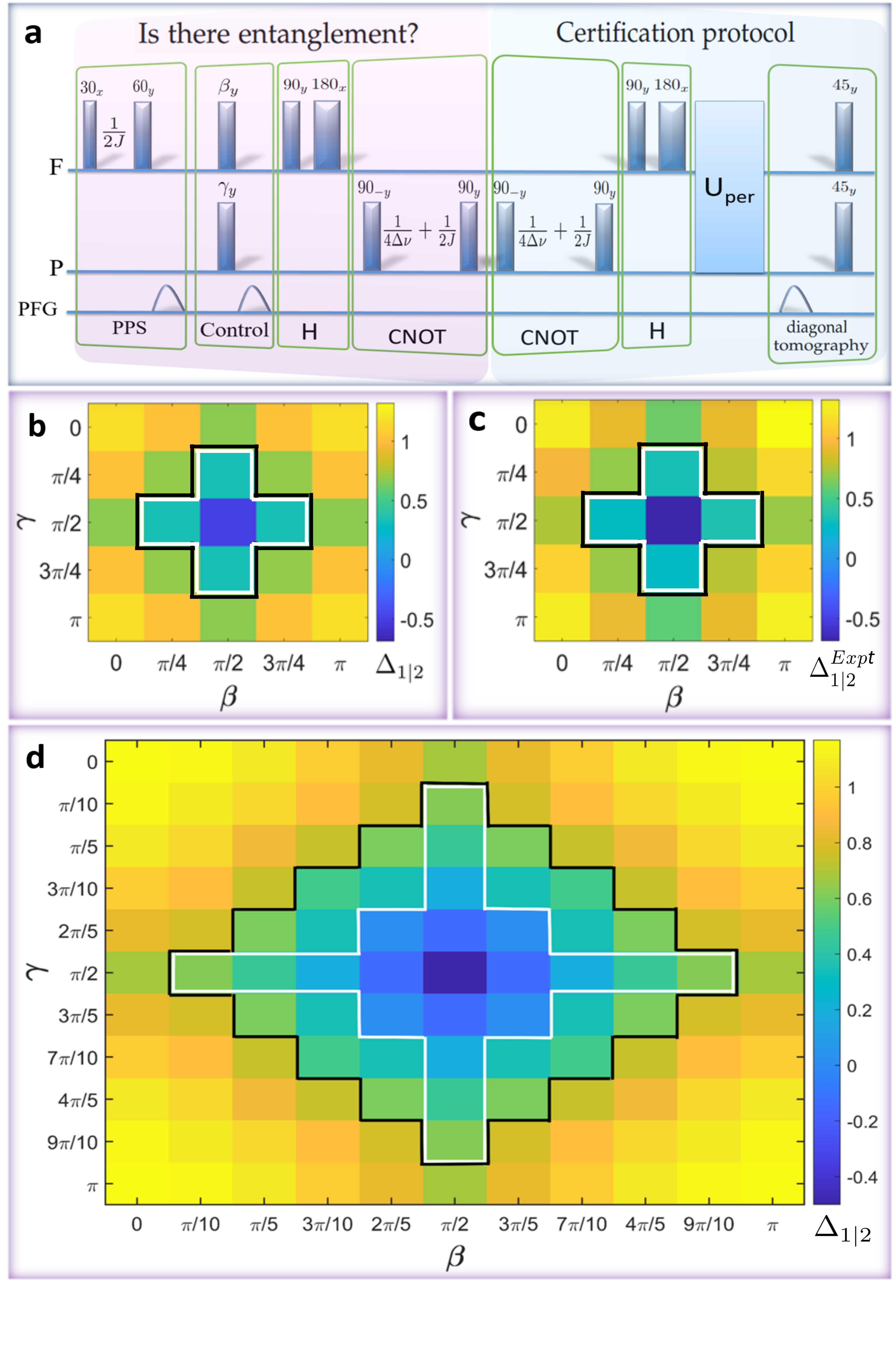}
\caption{(Color Online) (a) The NMR pulse sequence to prepare Bell diagonal state with control parameter $\beta$ and $\gamma$. Here PFG is the Pulsed-Field Gradient and $\Delta \nu$ is the resonance offset of both $^{19}$F and $^{31}$P.  
One of the operations [$U_\mathrm{per}\circ H\circ \mathrm{CNOT}$] takes the Bell diagonal state to its passive. Detailed explanation of this sequence is given in the supplemental. (b) Gradient-color-plot for theoretical values of $\Delta_{1|2}$ (in units of $\omega_P$) vs the control parameters $\beta$ and $\gamma$. States outside the white line (inner-perimeter) are entangled as $\Delta_{1|2} > \delta^G_{1|2}$. For the states outside the black line (outer-perimeter) $\Delta_{1|2} > \delta^I_{1|2}$ and hence they are also entangled. (c) Gradient-color-plot for the experimental values of $\Delta^{Exp}_{1|2}$ with estimated  errors of $\pm 0.1 \omega_P$. Here error originates both from the spin system as well as the NMR hardware, accordingly we have estimated the random error from the experimental NMR spectrum corresponding to the least signal-to-noise ratio providing a useful upper bound for errors. (d) $11\times11$ pixel theoretical gradient-color-plot of $\Delta_{1|2}$: evidently state independent certification scheme is weaker than the state dependent scheme.}\label{fig2}
\vspace{-.5cm}
\end{figure}
\begin{figure}
\centering
\includegraphics[trim={8cm 99cm 0cm 40cm},clip,width=9cm]{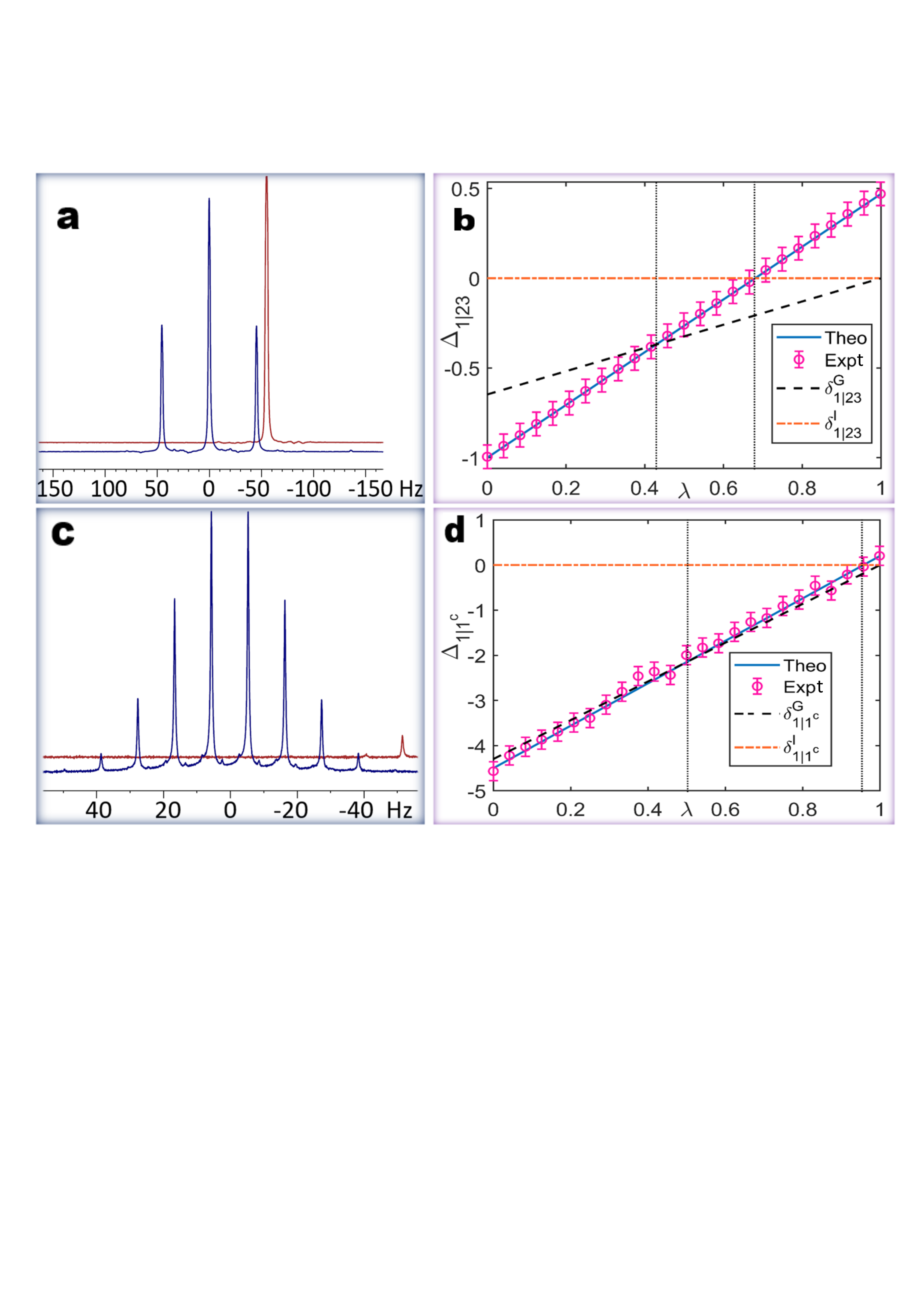}
\caption{(Color Online) (a) $^{19}$F spectra of FAN corresponding to one-pulse experiment on thermal state (front) and to the three-qubit GHZ state (back). (b) Plot $\Delta_{1|23}$ (in the unit of $ \omega_H$) vs purity $\lambda$ for $3$-qubit noisy GHZ states. Comparing the values of $\Delta_{1|23}$ and $\delta^G_{1|23}$, $\delta^I_{1|23}$, we identify the threshold values marked by the dotted lines: $\lambda = 3/7$ and $\lambda = 0.68$, respectively. Above these thresholds, the state exhibits entanglement. (c) $^{31}$P spectra of TMP corresponding to one pulse experiment (front), and to GHZ (back) . (d) $\Delta_{1|1^c}$ vs purity $\lambda$ for the $10$-qubit GHZ class.  Here $\lambda = 0.499$ and $\lambda = 0.957$ for $\delta^G_{1|1^c}$ and $\delta^I_{1|1^c}$ marks the entanglement threshold boundary. The errorbar represents the estimated random error due to noise.} \label{fig3}
\vspace{-.5cm}
\end{figure}

{\it Experiment II: (Multi-qubit systems).--} Multi-qubit entangled states within NMR architecture can be prepared in star topology register (STR) \cite{mahesh2021star}. STR involves a central qubit C (1st qubit) uniformly interacting with a set of $N-1$ identical satellite qubits A (see Fig. \ref{fig1}(b)). Central qubit can be selectively addressed as it is realized by a different nuclear isotope. The ancillary qubits being indistinguishable can be addressed globally only. STR allows efficient preparation of entangled GHZ state \cite{Jones2009}. The STR Hamiltonian along with the pulse sequence dynamics is described in Supplemental part.

We carry out experiments on the following two systems: (i) $3$-qubit STR using FAN, wherein $^{19}$F spin is the central qubit and two $^1$H spins are the satellite qubits, with $J_{CA} = 45.5$ Hz; (ii) $10$-qubit STR using TMP, wherein $^{31}$P spin is the central qubit and nine $^1$H spins are the satellite qubits, with $J_{CA} = 11.04$ Hz (see Fig. \ref{fig1}). After preparing the noisy state $\rho_\lambda[N]=(1-\lambda)\mathbf{I}/2^N+\lambda\ket{\psi_N}\bra{\psi_N}$, with $\ket{\psi_N}:=(\ket{0}^{\otimes N}+\ket{1}^{\otimes N})/\sqrt{2}$, we test entanglement across C-vs-A bipartition considering $N = 3$ and $10$ respectively for FAN and TMP. In doing that, we experimentally determine the thermodynamic quantity $\Delta_{1|1^c}$ along with global spectral dependent separability bound $\delta^G_{1|1^c}$ and state independent bound $\delta^I_{1|1^c}$. Subsequently, we find out the ranges of $\lambda$ for which $\Delta_{1|1^c}$ exceeds these bounds, and accordingly entanglement across C-vs-A gets certified. Experimental results are shown in Fig. \ref{fig3}, and the detailed analysis is deferred to the Supplemental section. Importantly, separability bounds generally rely on the Hamiltonian of the individual systems. This is noteworthy for the state-independent bound in particular. As the central qubit has a greater energy gap between the ground and excited states, the range of entanglement proportionally expands (see Remark {\bf 3} in Supplemental).

{\it Experiment III: (Global vs Global-Local separability bounds).--}
Here we experimentally establish the difference between global-local separability bound of condition (\ref{ecs}) and global separability bound of condition (\ref{ecw}) as stated in Theorem \ref{theo1} using DBFM (see Fig. \ref{fig1}(e)), dissolved in Acetone-d6 as the three-qubit spin system.  Here 1H, 13C are together treated as $X$, whereas 19F is treated as $X^c$. The procedure to obtain global and local spectral is discussed in the supplemental materials \cite{Supple}. In Fig. \ref{fig4} we plot $\Delta_{12|3}$ (experimental as well as theoretical), $\delta^{GL}_{12|3}$, $\delta^{G}_{12|3}$, and $\delta^{I}_{12|3}$ against the noise parameter $\lambda$.
\begin{figure}[t!]
\centering
\includegraphics[width=8cm, height=5cm]{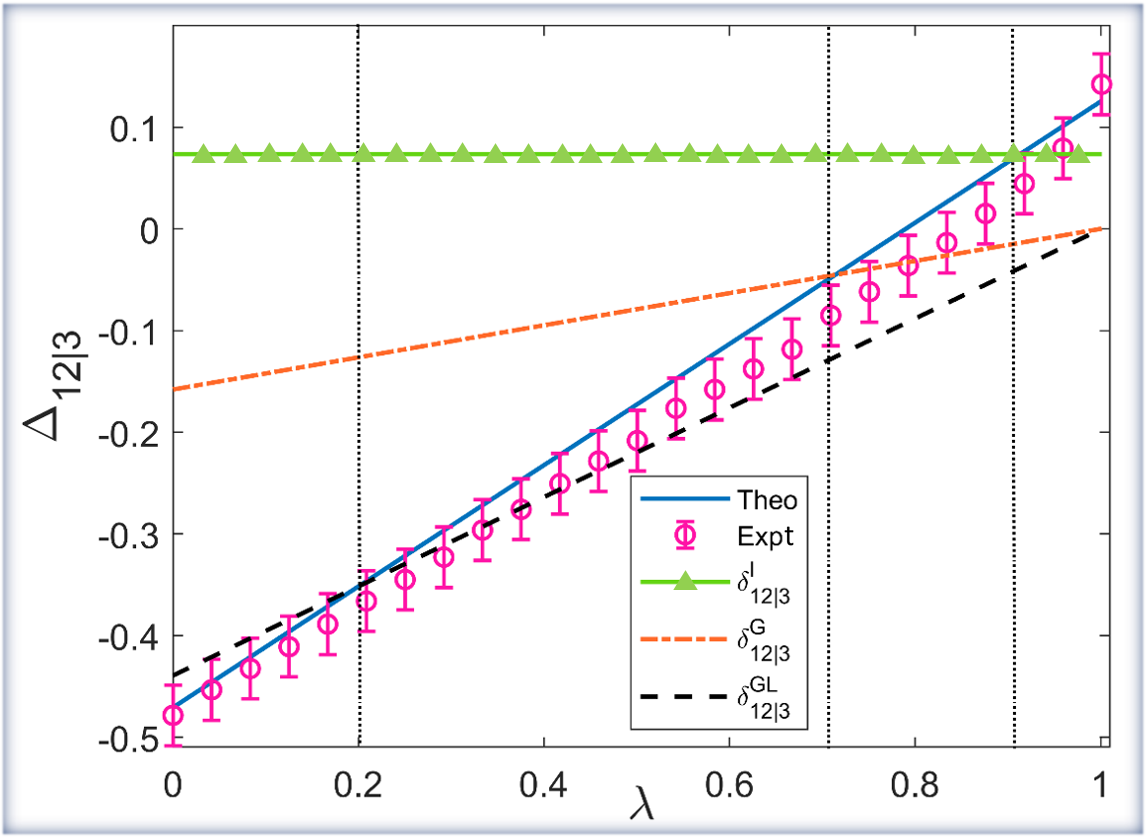}
\caption{(Color Online) Plot of $\Delta_{12|3}$ (experimental as well as theoretical), $\delta^{GL}_{12|3}$, $\delta^{G}_{12|3}$, and $\delta^{I}_{12|3}$ against the noise parameter $\lambda$. The state independent, the global, and the global-local separability conditions certify entanglement for the parameter ranges $\lambda> 0.91$, $\lambda> 0.71$, and $\lambda>1/5$ respectively, establishing hierarchy among these conditions.}.\label{fig4}
\vspace{-.5cm}
\end{figure}

{\it Discussion.--} Manipulating entanglement efficiently in multipartite system is essential for emerging quantum technologies that involve distributed quantum information protocols. Continuous research effort is going on to this aim with different quantum architectures \cite{Raimond2001,Kok2007,Wendin2017,Lu2015,Xiang2013}. Certifying entanglement is a crucial step for successful implementation of many quantum protocols. Along with device independent certification scheme through Bell tests \cite{Aspect1981,Aspect1982(1),Aspect1982(2),Bouwmeester1999,Pan2000,Zhao2003}, there exists device dependent witness based method of entanglement certification \cite{Chen2016,Loredo2016}. However, implementing those methods are quite challenging in practice when the quantum systems are composed of many subsystems.  

In that respect our proposed thermodynamic criteria are less demanding. It provides a way to certify entanglement by measuring global and local ergotropic works. We experimentally validate the proposed thermodynamic entanglement criterion in NMR architecture by considering particular classes  of $2$-qubit, $3$-qubit, and $10$-qubit noisy entangled states. In comparison with \cite{Singh_2018} that relies on full density state tomography of pure states and \cite{xue2022experimental} that is limited to bipartite system, our method achieves certification for multi-qubit mixed states with three different bounds based on the thermodynamic quantifiers of the system. Our thermodynamic approach opens up an easy avenue to certify entanglement even when the knowledge about the state in question is not available. While we have invoked the pseudopure paradigm for our ensemble architecture, similar protocols can be easily setup for other architectures with access to different degrees of state purity.  For instance, entanglement enhanced quantum sensing by optical probes \cite{Xia2023} or NV centers \cite{Xie2021} may be benefited from prior certification of entanglement. At this point, we would like to point out that study of ergotropy is constantly advancing with different quantum architectures, such as optical mode \cite{Andolina2019} and bosonic Gaussian models \cite{Francica2020,Tirone2021}. It will be therefore interesting to test entanglement in those physical systems using our proposed criteria. The recent work of \cite{Safranek2023} wherein coarse-grained measurement scheme is proposed is worth-mentioning at this point. Our study welcomes a number of other questions for future research. For instance, generalizing our criteria for systems with arbitrary local dimensions and generalizing to capture more exotic kinds of entanglement, such as genuine multi-partite entanglement, would be quite important. While the local passivity in our case is studied under local unitary operations, more general notion of strong local passivity is introduced by considering more general local quantum operation \cite{Frey2015,Alhambra2019,Katiyar2023}. Obtaining entanglement certification criteria under this generic consideration could also be quite interesting. 

{\bf Acknowledgments}: JJ acknowledges support from CSIR (Council of Scientific and Industrial Research) fellowship. MA and MB acknowledge funding from the National Mission in Interdisciplinary Cyber-Physical systems from the Department of Science and Technology through the I-HUB Quantum Technology Foundation (I-HUB QTF) (Grant no: I-HUB/PDF/2021-22/008). TSM acknowledges funding  from DST/ICPS/QuST/2019/Q67 and I-HUB QTF. MB acknowledges support through the research grant of INSPIRE Faculty fellowship from the Department of Science and Technology, Government of India and the start-up research grant from SERB, Department of Science and Technology (Grant no: SRG/2021/000267).   

\onecolumngrid
\hypersetup{linkcolor=black}
\tableofcontents
\section{Majorization based entanglement criteria}\label{app-a}
In this section we briefly review the majorization based entanglement criteria as relevant to the present work. The concept of majorization has been extensively studied in mathematics \cite{Marshall1979}, and its applications span across various domains, including quantum information theory. For instance, majorization plays a crucial role in detecting bipartite entanglement (through Nielsen-Kempe criteria \cite{Nielsen2001}), quantum state transformation \cite{Nielsen1999, Horodecki2003, Winter2016}, quantum thermodynamics \cite{Horodecki2013}, and more. Here, we shortly recall the concept of majorization. We will denote a probability distribution $\{p_i\}_{i=0}^{n-1}$ as a vector $\Vec{p}\equiv\{p_i\}_{i=0}^{n-1}\in \mathbb{R}^n$ with the elements arranged in decreasing order, {\it i.e.}, $p_{i+1}\ge p_i,~\forall~i\in\{0,\cdots,n-1\}$.
\begin{definition}
	A probability distribution $\Vec{p}$ majorizes another probability distribution $\Vec{q}$, denoted as $\Vec{p}\succ\Vec{q}$, {\it if and only if}  
	\begin{align}
		\sum^k_{i=0}p_i\geq\sum^k_{i=0}q_i, ~ \forall~k \in \{0,\cdots n-2\};~\&~~
		\sum^{n-1}_{i=0}p_i =\sum^{n-1}_{i=0}q_i.
	\end{align}
\end{definition}
If dimensions of the two vectors are not same extra zeros need to be appended to check their majorization relation. 
\begin{definition}
	A quantum state $\rho\in\mathcal{D}(\mathbb{X})$ majorizes another quantum state $\sigma\in\mathcal{D}(\mathbb{Y})$, denoted as $\rho\succ \sigma$, {\it if and only if} spectral vector $\vec{p}_\rho$ of the state $\rho$ majorizes the spectral vector $\vec{p}_\sigma$ of the state $\sigma$, {\it i.e.}, $\vec{p}_\rho\succ\vec{p}_\sigma$.
\end{definition}
Interestingly, Nielsen and Kempe provided a useful separability criteria based on majorization \cite{Nielsen2001}.\\
{\bf Nielsen-Kempe criteria}: Any bipartite separable state $\rho_{AB}\in\mathcal{D}(\mathbb{X}_A\otimes\mathbb{Y}_B)$ satisfies 
\begin{align}
	\vec{p}_{\rho_{A}}\succ\vec{p}_{\rho_{AB}},~~\&~~\vec{p}_{\rho_{B}}\succ\vec{p}_{\rho_{AB}}.  
\end{align}
\begin{figure}[t!]
	\centering
	\includegraphics[trim={0cm 0cm 0cm 0cm},clip=,width=8cm]{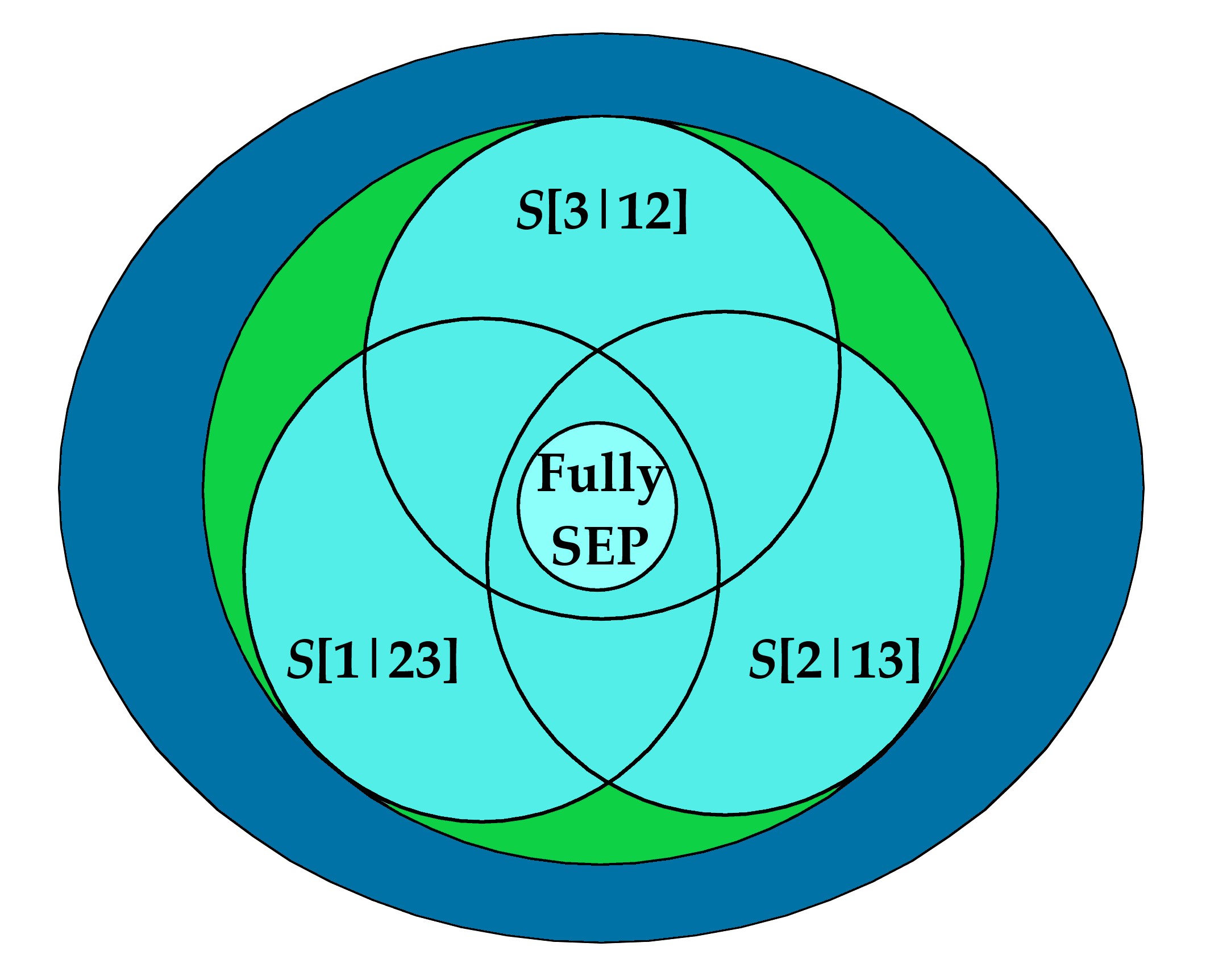}
	\caption{Entanglement in three qubit noisy GHZ states: The diagram illustrates various types of entanglement within three qubit systems. The large blue disc represents the complete state space of the three qubit system. The subset $\mathcal{S}[X|X^c]$, where $X \in {1,2,3}$, denotes a convex set that encompasses states separable across the partition of $X$ versus $X^c$. A state is considered bi-separable if it resides within $\mathrm{CovHul}\{\mathcal{S}[1|1^c],\mathcal{S}[2|2^c],\mathcal{S}[3|3^c]\}$. On the other hand, if a state can be expressed as a convex combination of tripartite product states it is called fully separable. It is worth noting that the intersecting region $\cap^3_{X=1}\mathcal{S}[X|X^c]$ is separable across all partitions. Interestingly, it is known that the set of fully separable states is a strict subset of $\cap^3_{X=1}\mathcal{S}[X|X^c]$ \cite{Bennett1999}. In this analysis, we examine the entanglement region of three qubit noisy GHZ states. These states are fully separable if and only if $0\leq p \leq \frac{1}{5}$. For $\frac{1}{5} < p \leq \frac{3}{7}$, the states are bi-separable but only reside in the green region, which re-establishes the fact that the set  $\mathrm{CovHul}\{\mathcal{S}[1|1^c],\mathcal{S}[2|2^c],\mathcal{S}[3|3^c]\}$ strictly contains $\cup_{X=1}^3\mathcal{S}[X|X^c]$. Conversely, states are genuinely entangled if and only if $p > \frac{3}{7}$.}\label{figs1}
\end{figure}
Here $\rho_A:=\Tr_B[\rho_{AB}]$ and $\rho_B:=\Tr_A[\rho_{AB}]$. Violation of any one of these criteria implies the state $\rho_{AB}$ is entangled. This condition can also be applied to certify entanglement across  any bipartition of a multipartite state $\rho_{A_1\cdots A_N}\in\mathcal{D}\left(\otimes_{i}\mathbb{C}^{d_i}\right)$. Recall that, $\mathcal{S}[X|X^c]$ denotes the set to states that are separable across $X$-vs-$X^c$ cut, with $X$ containing $\kappa$ parties and $X^c$ containing the remaining parties; and Nielsen-Kempe criteria can be applied by considering the global state $\rho_{A_1\cdots A_N}$ and the parts $\rho_X$ \& $\rho_{X^c}$. A state is genuinely entangled if it lies outside the convex hull ($\mathrm{CovHul}$) of all biseparable states. For instance, consider the $N$-qubit noisy GHZ states
\begin{equation}
	\rho_{(\lambda,N)}= \lambda \proj{\psi}_{\small{GHZ}}+(1-\lambda)\frac{I}{2^N}\label{noisyGHZ}, 
\end{equation}
where $|\psi\rangle_{GHZ}=\frac{1}{\sqrt{2}}|0\rangle^{\otimes N}+\frac{1}{\sqrt{2}}|1\rangle^{\otimes N}$ and $\lambda\in[0,1]$. The states in this class are fully separable {\it if and only if} $0\leq p \leq \frac{1}{1+2^{N-1}} $ \cite{Dur00}, and genuinely entangled {\it if and only if} $\frac{1-2^{1-N}}{2-2^{1-N}}<p\leq 1$ \cite{Guhne2010}. Using Nielsen-Kempe criteria it can be concluded that the state exhibits entanglement across the $(N-1)$-vs-$1$ partition when the value of $p>\frac{1}{1+2^{N-1}}$. In the intermediate region, specifically when $\frac{1}{1+2^{N-1}}<p\leq \frac{1-2^{1-N}}{2-2^{1-N}}$, the state is observed to be bi-separable across the $(N-1)$-vs-$1$ partition. Nevertheless, due to the symmetrical nature of this class, all the $(N-1)$-vs-$1$ partitions exhibit entanglement but not genuine entanglement (refer to Fig. \ref{figs1}).

\section{Thermodynamic criteria of Entanglement}\label{app-b}
In this section, we will present the proof of Theorem \textcolor{red}{1} and explore the connection between spectral-dependent thermodynamic criteria, as presented in Theorem \textcolor{red}{1}, and the Nielsen-Kempe criteria. Additionally, we will examine the potential utility of certain thermodynamic quantities in detecting entanglement in multipartite systems.  We will also derive some state independent thermodynamic entanglement criteria by considering an explicit form of the system's Hamiltonian for the systems on which we perform the experiments.

\subsection{Spectral-dependent thermodynamic criteria of entanglement}
\subsubsection{Proof of Theorem 1}
\begin{proof}
	The quantity $\Delta_{X|X^c}$ in Eq.(\textcolor{red}{1}) reads as
	\begin{align}
		\Delta_{X|X^c}
		& =\{\Tr[\rho H]-\Tr[\rho^p H]\}-\{\Tr[\rho_X H_X]-\Tr[\rho^p_X H_X]\}\nonumber\\
		&~~~~~~~~~~~-E(\rho_{X^c})+E^{X^c}_g \nonumber\\
		& = \sum^{2^{\kappa}-1}_{i=0}m_jx_j - \sum^{2^{N}-1}_{i=0}n_jt_j + E^{X^c}_g + E^{X}_g - E_g \nonumber\\
		\Delta_{X|X^c}& = \sum^{2^{\kappa}-1}_{i=1}m_jx_j - \sum^{2^{N}-1}_{i=1}n_jt_j\label{thermquant}.
	\end{align}
	According to Nielsen-Kempe separability criteria a state separable across $X$-vs-$X^c$ cut satisfies  $\rho_X\succ\rho$, {\it i.e.}, 
	\begin{align}
		x_0\geq t_0, ~~
		\Rightarrow \sum^{2^{\kappa}-1}_{i=1}x_i \leq\sum^{2^{N}-1}_{i=1}t_i~.\label{xmajo}
	\end{align}
	Substituting Eq.(\ref{xmajo}) in Eq.(\ref{thermquant}) we obtain
	\begin{align*}
		\Delta_{X|X^c}\le \sum_{i=1}^{2^\kappa-1}(m_i-m_1)x_i+\sum_{i=1}^{2^N-1}(m_1-n_i)t_i := \delta^{GL}_{X|X^c}.\label{GLSDB}
	\end{align*}
	This is the entanglement criterion (\textcolor{red}{2a}) of Theorem \textcolor{red}{1}. Evaluation of this criterion requires information about the spectral of the global state as well as spectral of the X-marginal. One can, however, achieve a entanglement criterion depending on the global spectral only. For that, rewrite Eq.(\ref{thermquant}) as
	\begin{align*}
		\Delta_{X|X^c}&=  \sum^{2^{\kappa}-1}_{i=1}m_jx_j + \sum^{2^\kappa-1}_{j=1}(m_j-m_{j-1})\sum^{2^\kappa-1}_{i=j}x_i-\sum^{2^\kappa-1}_{j=1}(m_j-m_{j-1})\sum^{2^\kappa-1}_{i=j}x_i- \sum^{2^{N}-1}_{i=1}n_jt_j~.
	\end{align*}
	Substituting the separability condition $\sum^{2^\kappa-1}_{i=j}x_i\leq \sum^{2^N-1}_{i=j}t_i$ we obtain,
	\begin{align}
		\Delta_{X|X^c}&\leq \sum^{2^{\kappa}-1}_{i=1}m_jx_j + \sum^{2^\kappa-1}_{j=1}(m_j-m_{j-1})\sum^{2^N-1}_{i=j}t_i-\sum^{2^\kappa-1}_{j=1}(m_j-m_{j-1})\sum^{2^\kappa-1}_{i=j}x_i- \sum^{2^{N}-1}_{i=1}n_jt_j\nonumber\\
		&=\sum^{2^\kappa-1}_{j=1}(m_j-m_{j-1})\sum^{2^N-1}_{i=j}t_i-\sum^{2^{N}-1}_{i=1}n_jt_j\nonumber\\
		& = \sum_{i=1}^{2^\kappa-2}(m_i-n_i)t_i+\sum_{i=2^\kappa-1}^{2^N-1}(m_{2^\kappa-1}-n_i)t_i\nonumber\\
		&:= \delta^G_{X|X^c}.\nonumber
	\end{align}
	This is the condition (\textcolor{red}{2b}) of Theorem \textcolor{red}{1}; and this completes the proof.
\end{proof}
\par
Note that, the criterion (\textcolor{red}{2a}) and the criterion (\textcolor{red}{2b}) help us to detect entanglement at $X$-vs-$X^c$ cut in the following way. For a given state $\rho$, if the values of thermodynamic quantity $\Delta_{X|X^c}$ of Eq.(\textcolor{red}{1}) exceeds the separability bound $\delta^{GL}_{X|X^c}$ then the state is entangled, and the thermodynamic criteria turns out to be equivalent to some conditions of Nielsen-Kempe criteria:
\begin{align}
	\sum^{2^{\kappa}-1}_{i=1}m_jx_j - \sum^{2^{N}-1}_{i=1}n_jt_j &>  \sum_{i=1}^{2^\kappa-1}(m_i-m_1)x_i+\sum_{i=1}^{2^N-1}(m_1-n_i)t_i \nonumber\\
	\Rightarrow\sum^{2^\kappa-1}_{i=1}x_i & > \sum^{2^N-1}_{i=1}t_i~~\Rightarrow~~x_0<t_0.
\end{align} 

\begin{remark}
	For the $N$-qubit noisy GHZ state of Eq. (\ref{noisyGHZ}) entanglement across $X$-vs-$X^c$ bipartition can be certified by comparing the values of $\Delta_{X|X^c}$ and $\delta^{GL}_{X|X^c}$. As it turns out according to this test the state is entangled for $\lambda>\lambda^{GL}_\kappa := \frac{2^{N-\kappa}-1}{2^{N-1}+2^{N-\kappa}-1}$, where the $X$ subsystem contains $\kappa$ qubits.  Important to note that the bound  $\lambda^{GL}_{\kappa}$ depends on the value of $\kappa$. As more subsystems are considered in the $X$ part, the criteria will encompass a broader range of entanglement, as illustrated in Fig. \ref{figs2}. In the present example, our criteria will encompass the entire range of entanglement, just like the Nielsen-Kempe criteria, when we choose $\kappa=N-1$.      
\end{remark}
\begin{figure}[t!]
	\centering
	\includegraphics[trim={0cm 0cm 0cm 0cm},clip=,width=8cm]{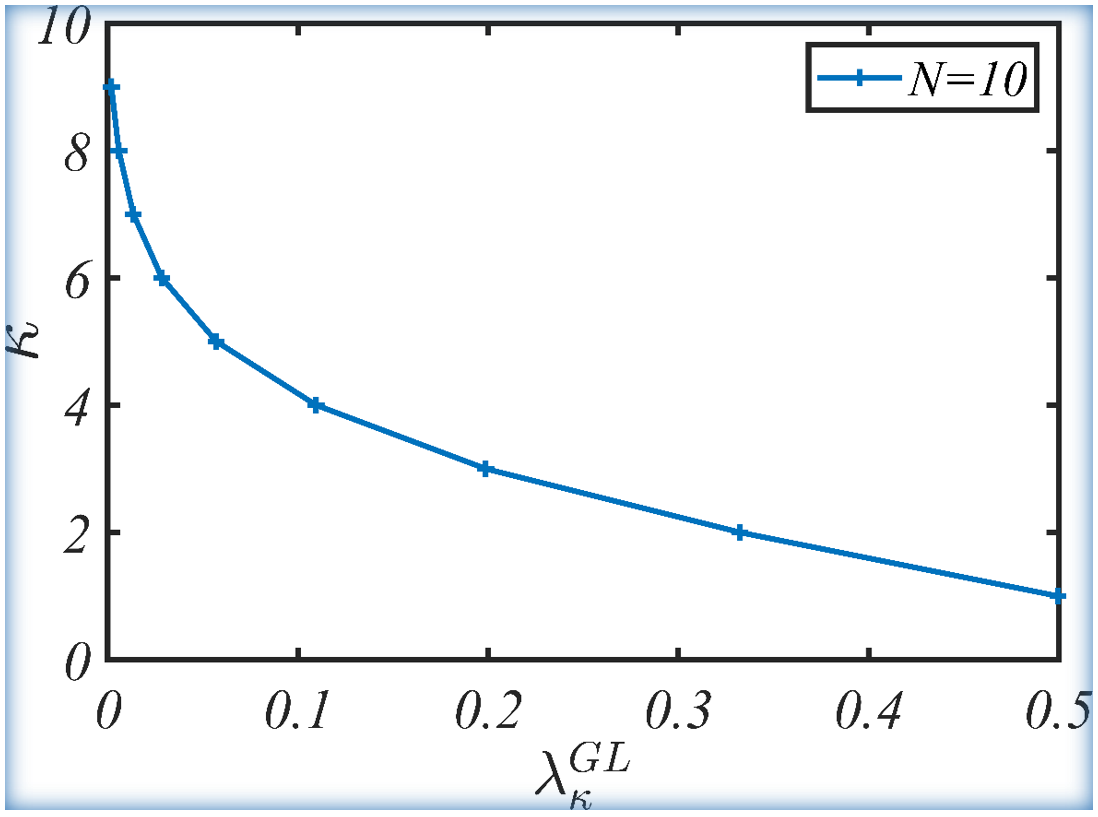}
	\caption{A plot of global and local spectral dependent entanglement threshold ($\lambda^{GL}_\kappa$) across $X$-vs-$X^c$ bi-partition  versus the number of subsystems ($\kappa$) in the $X$ part.}\label{figs2}
\end{figure}

Important to note that likewise  Nielsen-Kempe criteria criterion, (\textcolor{red}{2a}) of Theorem \textcolor{red}{1} is independent of the Hamiltonian of the given system. However, as evident from Eq.(\textcolor{red}{5}) and the separability bound (\textcolor{red}{2b}), entanglement detection condition explicitly depends on the Hamiltonian (except for $\delta^G_{1|1^c}$, since it becomes equal to $\delta^{GL}_{1|1^c}$) as shown below,
\begin{align}
	\sum^{2^{\kappa}-1}_{i=1}m_ix_i - \sum^{2^{N}-1}_{i=1}n_it_i &> \sum_{i=1}^{2^\kappa-2}(m_i-n_i)t_i+\sum_{i=2^\kappa-1}^{2^N-1}(m_{2^\kappa-1}-n_i)t_i \nonumber\\
	\Rightarrow\sum^{2^\kappa-1}_{i=1}m_ix_i &> \sum^{2^{\kappa}-2}_{i=1}m_it_i+m_{2^\kappa-1}\sum^{2^N-1}_{i=2^\kappa-1}t_i.
\end{align}

\subsection{State-independent thermodynamics criteria of entanglement}

Thermodynamic criteria provided in the Theorem \textcolor{red}{1} are spectral dependent and therefore entanglement certification through this criteria demands knowledge about the state in question. Next we will show that one can in-fact obtain spectral independent entanglement criterion, albeit weaker than spectral dependent criterion. But, the advantage is that one can invoke this state independent criterion to test entanglement of an unknown state. Such criteria generally depends on the Hamiltonian as well as the number of parties involved in the $X$-vs-$X^c$ partition. Since it is difficult to analyze the most generic case in one go, in the following we analyze different systems one after another.
\subsubsection{Two-qubit system}
Consider a two-qubit system governed by the Hamiltonian 
\begin{align}
	H=H_1\otimes\mathbf{I}+\mathbf{I}\otimes H_2,  
\end{align}
\begin{proposition}
	Any separable state of a two-qubit system governed by the Hamiltonian $H=H_1\otimes\mathbf{I}+\mathbf{I}\otimes H_2$ satisfies the condition 
	\begin{align*}
		\Delta_{1|2}\leq  \max\{(\alpha_1-\alpha_2)/2,0\}:=\delta^I_{1|2}.
	\end{align*}
	Here, $H_l:=\sum_{i=0}^1(E_l+i\alpha_l)\proj{i}$, for $l=1,2$.
\end{proposition}
\begin{proof}
	In this case criterion (\textcolor{red}{2b}) of Theorem \textcolor{red}{1} boils down to 
	\begin{equation}
		\Delta_{1|2}\leq \sum^3_{i=1}(m_1-n_i)t_i.\label{bi}
	\end{equation}
	Depending upon the values of $\alpha_1$ and $\alpha_2$ several cases are possible which we analyze below. 
	\begin{itemize}
		\item[(C-I)] $\alpha_1>\alpha_2>0$: In this case we have, $m_1=\alpha_1$, $n_1=\alpha_2, n_2=\alpha_1, n_3=\alpha_1+\alpha_2$, and accordingly condition (\ref{bi}) becomes 
		\begin{equation}
			\Delta_{1|2}\leq (\alpha_1-\alpha_2)t_1-\alpha_2t_3.
		\end{equation}
		As the global spectral $\vec{t}\equiv\{t_i\}_{i=0}^3$ are arranged in decreasing order, in the above inequality maximization occurs at $\Vec{t}\equiv \{\frac{1}{2},\frac{1}{2},0,0\}$; and therefore we have the spectral independent criterion
		\begin{equation}
			\Delta_{1|2}\leq \frac{\alpha_1-\alpha_2}{2} = \delta^I_{1|2}. \label{case1}
		\end{equation}
		\item[(C-II)] $\alpha_1=\alpha_2=\alpha>0$: Here we have $m_1=\alpha, n_1=\alpha, n_2=\alpha, n_3=2\alpha$, and accordingly condition (\ref{bi}) becomes 
		\begin{align}
			\Delta_{1|2}\leq -\alpha t_3.   
		\end{align}
		In this case maximization occurs at $\Vec{t}\equiv \{t_0,t_1,t_2,0\}$ and thus we have 
		\begin{equation}
			\Delta_{1|2}\leq 0 = \delta^I_{1|2} .
			\label{case2}
		\end{equation}
		\item[(C-III)] $0<\alpha_1<\alpha_2$: $m_1=\alpha_1, n_1=\alpha_1, n_2=\alpha_2, n_3=\alpha_1+\alpha_2$; and condition (\ref{bi}) reads as
		\begin{align}
			\Delta_{1|2}\leq -(\alpha_2-\alpha_1)t_2-\alpha_2 t_3.   
		\end{align}
		As maximization occurs at $\Vec{t}\equiv \{t_0,t_1,0,0\}$, we have
		\begin{equation}
			\Delta_{1|2}\leq 0 = \delta^I_{1|2} .\label{case3}
		\end{equation}
	\end{itemize}
	Combining (\ref{case1}), (\ref{case2}), and (\ref{case3}), we therefore have 
	\begin{equation}
		\Delta_{1|2}\leq \delta^I_{1|2} =  \max\left\{\frac{\alpha_1-\alpha_2}{2},0\right\}.
	\end{equation}
	This completes the proof. 
\end{proof}

\subsubsection{Three-qubit system}
Consider a $3$-qubit system with Hamiltonian
\begin{align}
	H=H_1\otimes\mathbf{I}\otimes\mathbf{I}+\mathbf{I}\otimes H_2\otimes\mathbf{I}+\mathbf{I}\otimes\mathbf{I}\otimes H_3, \label{H3}  
\end{align}
with $H_i=\sum^1_{j=0}(E_i+j\alpha_i)\proj{j}$ for $i\in\{1,2,3\}$, and $\alpha_i>0~\forall~i$. Depending upon the values of $\{\alpha_i\}_{i=1}^3$ and depending on the bipartitions considered several cases are possible. In the following, we will analyze the cases relevant to our experiment.   
\begin{proposition}
	Consider a $3$-qubit system governed by the Hamiltonian (\ref{H3}). Any state of this system separable across $1$-vs-$23$ bipartition satisfies the condition    
	\begin{subequations}
		\begin{align}
			\Delta_{1|23} &\leq 0 := \delta^I_{1|23}, ~~~~~~~~~~~~ \mbox{when} ~~\alpha_1<\alpha_2=\alpha_3:=\alpha \\
			\Delta_{1|23} &\leq \frac{\alpha-\alpha_3}{2} := \delta^I_{1|23}, ~~~~\mbox{when} ~~~ \alpha_1=\alpha_2=\alpha > \alpha_3.
		\end{align}
	\end{subequations}
\end{proposition}
\begin{proof}
	In this case, across $1$-vs-$23$ cut the criterion (\textcolor{red}{2b}) of Theorem \textcolor{red}{1} reads as \begin{equation}
		\Delta_{1|23} \leq \sum^7_{i=1}(m_1-n_i)t_i~.\label{3qu1}
	\end{equation}
	\begin{itemize}
		\item[(C-I)] When $\alpha_1<\alpha_2=\alpha_3=\alpha$, we have $m_1=\alpha_1, n_1=\alpha_1,n_2=n_3=\alpha,n_4=n_5=\alpha+\alpha_1,n_6=2\alpha,n_7=2\alpha+\alpha_1$, and accordingly above becomes
		\begin{equation}
			\Delta_{1|23}\leq -[(\alpha-\alpha_1)(t_2+t_3)+\alpha(t_4+t_5)+(2\alpha-\alpha_1)t_6+2\alpha t_7]. \label{delta123}
		\end{equation}
		Since the function $f:=(\alpha-\alpha_1)(t_2+t_3)+\alpha(t_4+t_5)+(2\alpha-\alpha_1)t_6+2\alpha t_7$ is a linear in $\Vec{t}\equiv \{t_i\}^7_{i=0}$, and since spectral are arranged in decreasing order, it is evident that $f$ will take minimum value $0$ at $\vec{t}\equiv (1,0,0,0,0,0,0,0)$. This proves the claim 
		\begin{align}
			\Delta_{1|23} \leq 0 = \delta^I_{1|23}.
			\label{SIB3}
		\end{align}
		\item[(C-II)] Similarly, for $\alpha_1=\alpha_2=\alpha>\alpha_3$ we have $m_1=\alpha, n_1=\alpha_3, n_2 = n_3 = \alpha, n_4 = n_5 =\alpha+\alpha_3, n_6=2\alpha, n_7=2\alpha+\alpha_3$, and accordingly
		condition (\ref{3qu1}) becomes
		\begin{equation}
			\Delta_{1|23} \leq (\alpha-\alpha_3)t_1-\alpha_3(t_4+t_5)-\alpha t_6 -(\alpha+\alpha_3)t_7.
		\end{equation}
		In this case the function $g=(\alpha-\alpha_3)t_1-\alpha_3(t_4+t_5)-\alpha t_6 -(\alpha+\alpha_3)t_7$ takes maximum value at $\Vec{t}\equiv (1/2,1/2,0,0,0,0,0,0)$, which further imply
		\begin{align*}
			\Delta_{1|23}\leq \frac{\alpha-\alpha_3}{2}=\delta^I_{1|23}.   
		\end{align*}
	\end{itemize}
	This completes the proof.
\end{proof}
\begin{proposition}
	Consider a $3$-qubit system governed by the Hamiltonian (\ref{H3}). Any state of this system separable across $12$-vs-$3$ bipartition satisfies the condition    
	\begin{subequations}
		\begin{align}
			\Delta_{12|3} &\leq \frac{\alpha_1}{4} := \delta^I_{12|3}, ~~~~~~~~~~~~~~~~\mbox{when} ~~0<\alpha_1<\alpha_2=\alpha_3:=\alpha; \\
			\Delta_{12|3} &\leq \frac{\alpha-\alpha_3}{4}+\frac{\alpha}{4} := \delta^I_{12|3}, ~~~ \mbox{when} ~~~ \alpha_1=\alpha_2=\alpha > \alpha_3>0~~ and~~ \alpha_3 \geq \frac{2}{3} \alpha;\label{threeI}\\
			\Delta_{12|3} &\leq \frac{\alpha-\alpha_3}{2}+\frac{\alpha}{6} := \delta^I_{12|3}, ~~~ \mbox{when} ~~~ \alpha_1=\alpha_2=\alpha > \alpha_3>0~~ and~~ \alpha_3 \leq \frac{2}{3} \alpha.\label{threeII}
		\end{align}
	\end{subequations}
\end{proposition}
\begin{proof}
	Criterion (2b) of Theorem 1 boils down to 
	\begin{equation}
		\Delta_{12|3} \leq \sum^3_{i=1}(m_i-n_i)t_i+\sum^7_{i=4}(m_3-n_i)t_i.\label{3quibt2}
	\end{equation}
	\begin{itemize}
		\item [(C-I)] For $0<\alpha_1<\alpha_2=\alpha_3=\alpha$, we have $m_1=\alpha_1,m_2=\alpha,m_3=\alpha+\alpha_1,n_1=\alpha_1,n_2=n_3=\alpha,n_4=n_5=\alpha+\alpha_1,n_6=2\alpha,n_7=2\alpha+\alpha_1,$ and accordingly condition (\ref{3quibt2}) becomes,
		\begin{equation}
			\Delta_{12|3}\leq \alpha_1t_3.
		\end{equation}
		As maximization occurs at $\vec{t}\equiv\{1/4,1/4,1/4,1/4,0,0,0,0\}$, we have 
		\begin{equation}
			\Delta_{12|3}\leq \frac{\alpha_1}{4}=\delta^I_{12|3}.
		\end{equation}
		\item [(C-II)] For $0<\alpha_1<\alpha_2=\alpha_3=\alpha$, we have $m_1=\alpha,m_2=\alpha,m_3=2\alpha,n_1=\alpha_3,n_2=n_3=\alpha,n_4=n_5=\alpha+\alpha_3,n_6=2\alpha,n_7=2\alpha+\alpha_3,$ and accordingly condition (\ref{3quibt2}) becomes,
		\begin{equation}
			\Delta_{12|3}\leq (\alpha-\alpha_3)(t_1+t_4+t_5)+\alpha t_3.
		\end{equation}
		If $\alpha_3 \geq \frac{2}{3}\alpha$, maximization occurs at $\vec{t}\equiv\{1/4,1/4,1/4,1/4,0,0,0,0\}$ and we have 
		\begin{equation}
			\Delta_{12|3}\leq \frac{\alpha-\alpha_3}{4}+\frac{\alpha}{4} =\delta^I_{12|3}.
		\end{equation}
		If $\alpha_3 \leq \frac{2}{3}\alpha$, maximization occurs at $\vec{t}\equiv\{1/6,1/6,1/6,1/6,1/6,1/6,0,0\}$ and we have 
		\begin{equation}
			\Delta_{12|3}\leq \frac{\alpha-\alpha_3}{2}+\frac{\alpha}{6} =\delta^I_{12|3}.
		\end{equation}
	\end{itemize}
	This completes proof of the claim. 
\end{proof}

\begin{remark}
	It is instructive to see an explicit example how the entanglement certification criteria get weakened with lesser amount of information about the state in question. For that, consider a three-qubit noisy GHZ state described in Eq. (\ref{noisyGHZ}), where $N=3$, and the Hamiltonian has the following specifications: $\alpha_1=\alpha_2=\alpha > \alpha_3=2\alpha/3$. Let's set $\kappa=2$ and focus on the thermodynamic quantity $\Delta_{12|3}$(Eq. \textcolor{red}{1}). By comparing the values of $\Delta_{12|3}$ with $\delta^{GL}_{12|3}$(\textcolor{red}{2a}), $\delta^G_{12|3}$,(\textcolor{red}{2b}) and $\delta^I_{12|3}$(\ref{threeI} \& \ref{threeII}), we obtain the entanglement threshold values: $\lambda^{GL}_{12|3}=1/5 < \lambda^{G}_{12|3}=1/2 < \lambda^{I}_{12|3}=4/5$, respectively. Beyond these threshold values, the state is entangled, and it is evident that the range of entanglement expands as our separability bound $\delta_{12|3}$ incorporates more information about the state.    
\end{remark}

\subsubsection{Ten-qubit system} 
Consider a $10$-qubit system with Hamiltonian 
\begin{align}
	H=\sum_{l=1}^{10}\tilde{H}_{l};~\mbox{where}~ \tilde{H}_{l}:=\mathbf{I}_1\otimes\cdots\otimes\mathbf{I}_{l-1}\otimes H_l\otimes\mathbf{I}_{l+1}\otimes\cdots\otimes\mathbf{I}_{10},   
\end{align}
with  $H_l=\sum^{1}_{i=0}(E_l+i\alpha_l)|i\rangle\langle i|$. Our experiment considers a central qubit whose energy gap between excited and ground state is denoted by $\alpha_c$ and the rest qubits are identical where energy gap takes the value $\alpha>\alpha_c$.
\begin{proposition}
	Any state of this system separable across $1$-vs-$1^c$ bipartition satisfies the condition
	\begin{subequations}
		\begin{align}
			\Delta_{1|1^c} &\leq 0 := \delta^I_{1|1^c}, ~~~~~~~~~~~~ \mbox{when} ~~\alpha_1=\alpha_c<\alpha_j=\alpha ~~~ \forall~ j \neq 1  \\
			\Delta_{1|1^c} &\leq \frac{\alpha-\alpha_c}{2} := \delta^I_{1|1^c}, ~~~~\mbox{when} ~~~ \alpha_1=\alpha.
		\end{align}
	\end{subequations}  
\end{proposition}
\begin{proof}
	In this case criterion (2b) of Theorem 1 boils down to 
	\begin{equation}
		\Delta_{1|1^c}\leq \sum^{2^{10}-1}_{i=1}(m_1-n_i)t_i.\label{ten}.
	\end{equation}
	Depending on the central qubit's arrangement in bi-partition two cases are possible.
	\begin{itemize}
		\item [(C-I)] $\alpha>\alpha_1=\alpha_c>0: m_1=\alpha_c,n_1=\alpha_c,n_j\geq \alpha~~ \forall j \neq 1$, and accordingly condition (\ref{ten}) becomes,
		\begin{equation}
			\Delta_{1|1^c}\leq -\sum^{2^{10}-1}_{i=2}(n_i-\alpha_c)t_i.
			\label{delta10}
		\end{equation}
		Note that maximization occurs at $\vec{t}\equiv\{t_0,t_1,t_3,t_k=0\}~~\forall k\in\{4,\cdots,2^{10}-1\}$ and we have 
		\begin{equation}
			\Delta_{1|1^c}\leq 0 := \delta^I_{1|1^c}.
			\label{SIB10}
		\end{equation}
		\item [(C-II)] $0<\alpha_1=\alpha : m_1=\alpha,n_1=\alpha_c,n_j\geq \alpha ~~ \forall j \neq 1$ and accordingly condition (\ref{ten}) becomes,
		\begin{equation}
			\Delta_{1|1^c}\leq (\alpha-\alpha_c)t_1-\sum^{2^{10}-1}_{j=2}(n_j-\alpha_c) t_j.
		\end{equation}
		In this case maximization occurs at $\vec{t}\equiv\{1/2,1/2,t_k\} ~~\forall k \in \{3,\cdots,2^{10}-1\}$ and we have 
		\begin{equation}
			\Delta_{1|1^c}\leq \frac{\alpha-\alpha_c}{2} := \delta^I_{1|1^c}.
		\end{equation}
	\end{itemize}
	This completes the proof.
\end{proof}
\begin{figure}[t!]
	\centering
	\includegraphics[trim={0cm 0cm 0cm 0cm},clip=,width=8cm]{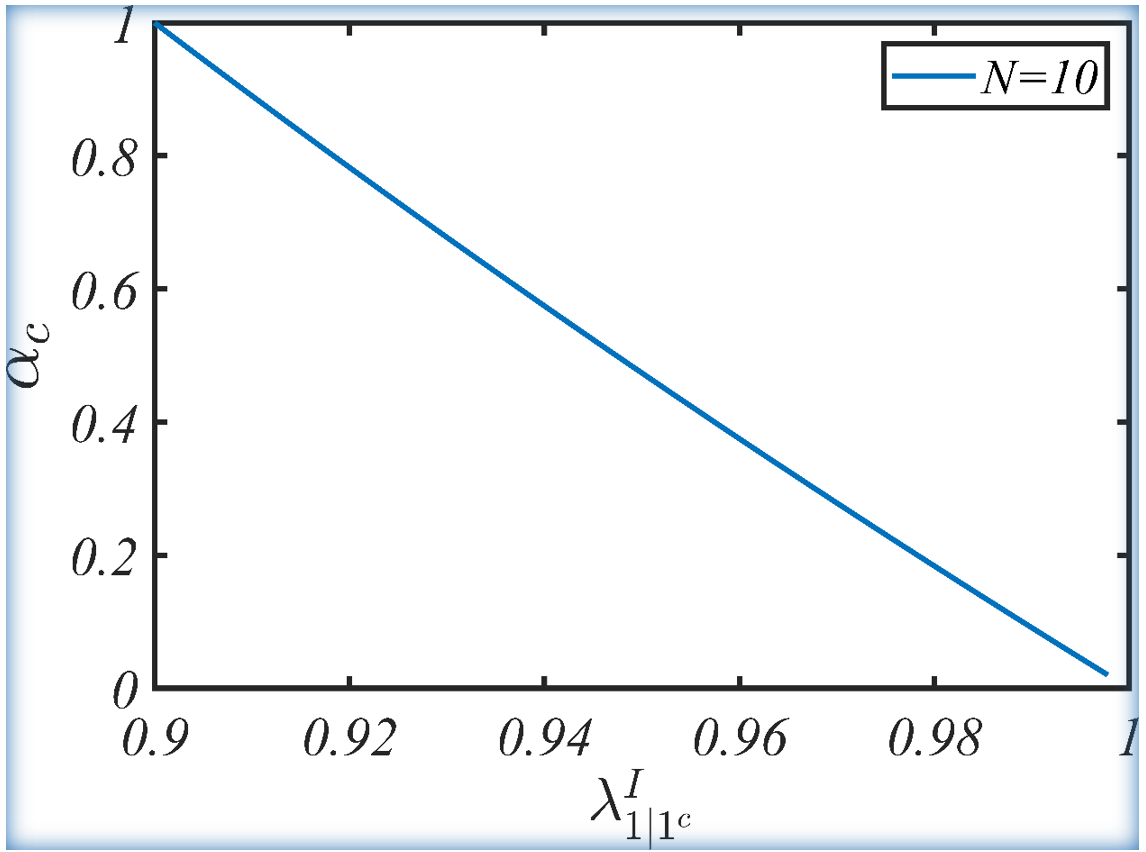}
	\caption{A plot of state independent entanglement threshold ($\lambda^{GL}_\kappa$) across $1$-vs-$1^c$ bi-partition versus the energy difference in central qubit ($\alpha_c$) for a ten-qubit system.}\label{figs3}
\end{figure}
\begin{remark}
	It turns out that the state-independent value of purity ensuring entanglement varies with the energy difference between the ground and excited states of the central qubit ($\alpha_c$). Let's consider an $N$-qubit noisy GHZ system (\ref{noisyGHZ}) where $N-1$ qubits are identical with energy parameter $\alpha$, while the central qubit has $\alpha_c<\alpha$. By comparing the values of $\Delta_{1|1^c}$ with the separability bound $\delta^I_{1|1^c}$, we can determine that the entanglement threshold value is given by $\lambda^{I}_{1|1^c}=\frac{(N-1)\alpha}{\alpha_c+(N-1)\alpha}$. If the purity parameter $\lambda$ exceeds this bound, we classify the state as entangled. Note that as $\alpha_c$ approaches to $\alpha$, state independent thresh-hold value of purity $\lambda^I_{1|1^c}$ decreases and consequently more region of entanglement can be detected (see Fig.\ref{figs3}).
\end{remark}

\section{Analysis of the experimental procedure and Results}\label{app-c}
\subsection{Experiment I-a: Two-qubit Bell-diagonal state}
In this experiment we addresses entanglement certification in two-qubit Bell diagonal state. For that, we use a two-qubit spin system of Sodium Fluoro-Phosphate (NaFP) dissolved in heavy water (D$_2$O) (see Fig. \textcolor{red}{2}(b) in manuscript). The Hamiltonian of the system in terms of the internal part and the RF drive reads as 
\begin{subequations}
	\begin{align}
		H_{12}&= H_{12}^{\mathrm{int}} + 
		H_{12}^{\mathrm{RF}},~\mbox{where}\\ 
		H_{12}^{\mathrm{int}}&=  -  \omega_F I_{1z} -  \omega_P I_{2z} + 2\pi J I_{1z}I_{2z} ~\mbox{and}\nonumber\\
		H_{12}^{\mathrm{RF}}&=  \Omega_F(t) I_{1x}+ \Omega_P(t) I_{2x},
	\end{align}
\end{subequations}
with $\{I_{ix}=\sigma_{ix}/2,I_{iy}=\sigma_{iy}/2,I_{iz}=\sigma_{iz}/2\}$ denoting the spin operators in NMR language \cite{cavanagh1996protein}.\\\\
{\bf Preparation step}\\
To prepare the Bell diagonal state with two independent parameter $\beta$ and $\gamma$ we follow the quantum circuit as shown in Fig.\ref{figs4}(a) \cite{riedel2021bell} where we feed pure state $\ket{11}_{12}$ at the input. To create this input state within NMR architecture, we start with thermal state, which under high-field and high-temperature approximation ($ \hbar \omega_H << k_BT$) reads as $\rho_{th} = \mathbf{I}/4 + \epsilon_P (\frac{\gamma_F}{\gamma_P}I_{1z} + I_{2z}$), where  $\epsilon_P =   \gamma_PB_0/ 4k_BT$. Following the method of spatial averaging \cite{cory1997ensemble} which results in PPS (Psuedo-Pure State) pulse sequence (see Fig. 2(a) in main manuscript) we have $\ket{11}\bra{11}^{pps} = (1-\epsilon)\mathrm{I}/4 + \epsilon\ket{11}\bra{11}$. While working in the regime of PPS we take $\epsilon = 1$ and accordingly we have $\ket{11}\bra{11}^{pps} \equiv \ket{11}\bra{11}$. Now, following the quantum circuit of Fig. \ref{figs4}(b), whose experimental implementation is shown in Fig. 2(a) of the main manuscript, we prepare
\begin{align}
	\rho_{12}&= \sum_{i,j=0}^{1}p_{ij}\ket{\mathcal{B}_{ij}}\bra{\mathcal{B}_{ij}},~\mbox{where}\\
	\ket{\mathcal{B}_{0j}}&:=\frac{\ket{00}+(-1)^j\ket{11}}{\sqrt{2}},~\ket{\mathcal{B}_{1j}}:=\frac{\ket{01}+(-1)^j\ket{10}}{\sqrt{2}},\nonumber\\
	p_{00}&:=\sin^2(\beta/2)\sin^2(\gamma/2),~~~~p_{01}:= \sin^2(\beta/2)\cos^2(\gamma/2),\nonumber\\
	p_{10}&:=\cos^2(\beta/2)\sin^2(\gamma/2),~~~~p_{11}:= \cos^2(\beta/2)\cos^2(\gamma/2).\nonumber
\end{align}
\begin{figure}[t!]
	\centering
	\includegraphics[trim={0cm 0cm 0cm 0cm},clip=,width=10cm]{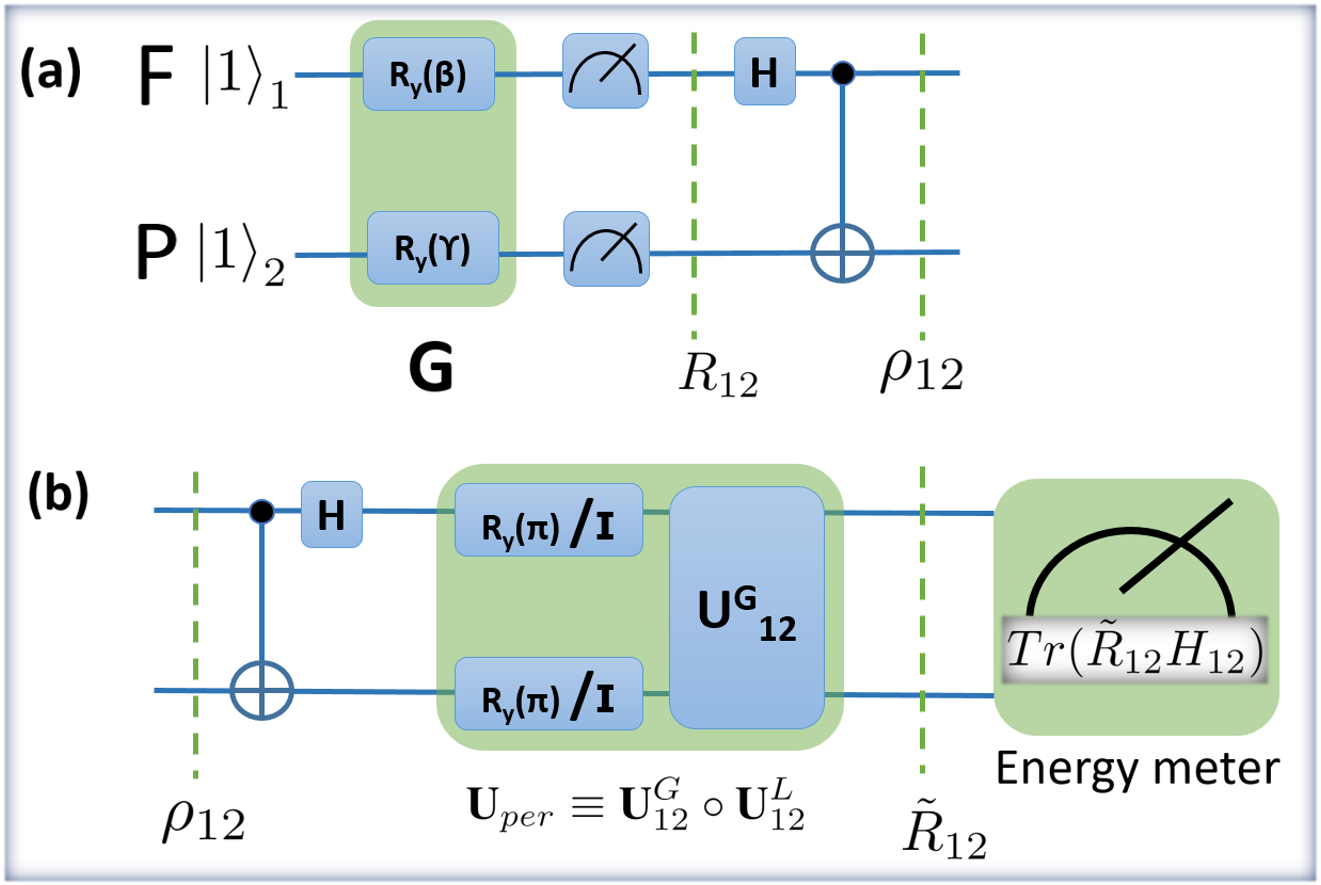}
	\caption{(a) Quantum circuit to prepare Bell-diagonal state $\rho_{12}$. Here $R_y(\alpha)$ is the rotation along y-axis with angle $\alpha$ and the unread measurement after G is equivalent to a crusher gradient that destroys all the coherence in x-y plane. (b) Simplified circuit with two independent parameters along with passive state $\rho^p_{12}$ conversion.}\label{figs4}
\end{figure}
~\vspace{-.5cm}\\
{\bf Entanglement certification step}\\
Our aim is to certify entanglement of the prepared state through the thermodynamics entanglement criteria obtained in Theorem 1. For that we need to evaluate the following thermodynamic quantity from experiment:
\begin{align}
	\Delta_{1|2}:=W(\rho_{12})-W_{[\kappa]}(\rho_1)-E(\rho_{2})+E^{2}_g.\label{expbd}   
\end{align}
To evaluate $W(\rho_{12})$, we need to apply a proper unitary so that the state $\rho_{12}$ evolves to the lowest energy state ({\it i.e.} the passive state). In case of an unknown state $\rho_{12}$ finding the optimal unitary is quite tedious. However, for the Bell-diagonal state in the following we provide a systematic method to do the same. First we apply CNOT and Hadamard [see Fig.\ref{figs4}(b)] to evolve $\rho_{12}$ to 
\begin{align}
	R_{12} &= \sum_{i,j=0}^{1}p_{ij}\ket{ij}\bra{ij},    
\end{align}
which is diagonal in computational basis $\{\ket{ij}\}_{i,j=0}^1$. In the next step, we have to apply unitary operation that permutes the computational basis that accordingly leads us to the lowest energy state. As it turns out, in this case $24$ different permutations $\textbf{U}_{per}$ are possible that can be represented as $\textbf{U}_{per}\equiv \textbf{U}_{12}^{G}\circ \textbf{U}_{12}^L$, where
\begin{subequations}
	\begin{align}
		\textbf{U}_{12}^{G}\in&\left\{\textbf{I}_{12}, \textbf{CNOT}_1, \textbf{CNOT}_2, \textbf{CNOT}_1 \circ \textbf{CNOT}_2, \textbf{CNOT}_2 \circ  \textbf{CNOT}_1,\textbf{CNOT}_1 \circ  \textbf{CNOT}_2 \circ  \textbf{CNOT}_1\right\},\\  
		\textbf{U}_{12}^{L}\in&\left\{\textbf{I}_1 \otimes \textbf{I}_2,\textbf{I}_1 \otimes \textbf{Y}_2,\textbf{Y}_1 \otimes \textbf{I}_2, \textbf{Y}_1 \otimes \textbf{Y}_2\right\},~\mbox{with}
	\end{align}    
\end{subequations}
\begin{align}
	\textbf{Y}_s:\left\{\!\begin{aligned}
		\ket{0}_s\mapsto\ket{1}_s\\
		\ket{1}_s\mapsto\ket{0}_s
	\end{aligned}\right\};~~\textbf{CNOT}_1:\left\{\!\begin{aligned}
		\ket{00}_{12}\mapsto\ket{00}_{12}\\
		\ket{01}_{12}\mapsto\ket{01}_{12}\\
		\ket{10}_{12}\mapsto\ket{11}_{12}\\
		\ket{11}_{12}\mapsto\ket{10}_{12}\\
	\end{aligned}\right\};~~\textbf{CNOT}_2:\left\{\!\begin{aligned}
		\ket{00}_{12}\mapsto\ket{00}_{12}\\
		\ket{01}_{12}\mapsto\ket{11}_{12}\\
		\ket{10}_{12}\mapsto\ket{10}_{12}\\
		\ket{11}_{12}\mapsto\ket{01}_{12}\\
	\end{aligned}\right\};
\end{align}
and $\textbf{I}$ be the identity operation. The task now boils down to evaluating the energy of the state $\tilde{R}_{12}=\textbf{U}_{per}R_{12}\textbf{U}^\dagger_{per}$ for obtaining the lowest energy state by varying these $24$ permutations. This completes the process of evaluating $W(\rho_{12})$. On the other hand, $W_{[\kappa]}(\rho_1)$ in this case becomes zero as $\rho_{12}$ being the Bell-diagonal state we have $\rho_1=\textbf{I}/2$. Since the system's Hamiltonian is known the other two terms in left hand side of Eq.(\ref{expbd}) can be evaluated immediately. Please note that, during the process we do not need to know the spectral of the given state since the expected energy value suffice the purpose. This experimentally obtained thermodynamic quantity $\Delta^{Expt}_{1|2}$ is plotted in Fig.2(c) of the main manuscript. Now to certify entanglement, we need to evaluate the separability bounds $\delta^G_{1|2}$ and $\delta^I_{1|2}$ for the given state $\rho_{12}$. We can also calculate $\Delta_{1|2}$ analytically to tally with our experiment. For the system in consideration  with $H_{12}^{\mathrm{int}} \approx  -  \omega_F I_{1z} -  \omega_P I_{2z}$ where $J << (\omega_F,\omega_P)$ we have 
\begin{align}
	\left\{\!\begin{aligned}
		x_0= 1/2,~~ x_1 = 1/2,~~ m_0 = 0,~~ m_1 =  \omega_F=\alpha_1,~~ n_0 = 0, \\
		n_1 =  \omega_P=\alpha_2,~~ n_2 =  \omega_F,~~n_3 =  (\omega_F+\omega_P) 
	\end{aligned}\right\};
\end{align}
and accordingly we get
\begin{subequations}
	\begin{align}
		\Delta_{1|2}&= m_1x_1-\left(\sum^3_{i=1}n_it_i\right)= \left[\omega_F/2 - \left(\omega_P t_1+ \omega_F t_2+(\omega_P+\omega_F)t_3\right)\right],\\
		\delta^G_{1|2}&= \sum^3_{i=1}\left(m_1-n_i\right)t_i= \left(\omega_F t_1 -\omega_P(t_1+t_3)\right), \\
		\delta^I_{1|2}&= \frac{\alpha_1-\alpha_2}{2}= (\omega_F -\omega_P)/2,   
	\end{align}
\end{subequations}
where vector $\vec{t}\equiv\{t_k\}_{k=0}^3$ is obtained by arranging $\{p_{ij}\}$ in descending order.Now, in $500$MHz spectrometer we have, $\omega_F = 470.385$MHz and $\omega_P = 202.404$MHz, which thus yield
\begin{subequations}
	\begin{align}
		\Delta_{1|2}&= (1.162 - 2.324t_2-3.324t_3-t_1) \omega_P, \\
		\delta^G_{1|2}&= (1.324t_1 - t_3) \omega_P,\\ 
		\delta^I_{1|2}&= 0.662 \omega_P.
	\end{align}
	\label{supeq:uper}
\end{subequations}
Once again, note that quantity $\delta^I_{1|2}$ is state independent. Whenever, $\Delta^{Expt}_{1|2}>\delta^I_{1|2}$ the state in question is entangled. In Fig.2(c), entanglement of the states with parameters $\beta,\gamma$ lying outside the ``black boundary line" is certified with this criterion. To tally this experimental result we have also plotted the analytically obtained quantity $\Delta_{1|2}$ in Fig.2(b). On the other hand, the quantity $\delta^G_{1|2}$ depends on state parameter and the criterion $\Delta^{Expt}_{1|2}>\delta^G_{1|2}$ becomes stronger than the earlier one as depicted in Fig.2(b),(c),(d).\\\\
\textbf{NMR pulse sequence} \\
\begin{figure}[h]
	\centering
	\includegraphics[trim={0cm 0cm 0cm 0cm},clip=,width=12.5cm,height=5.5cm]{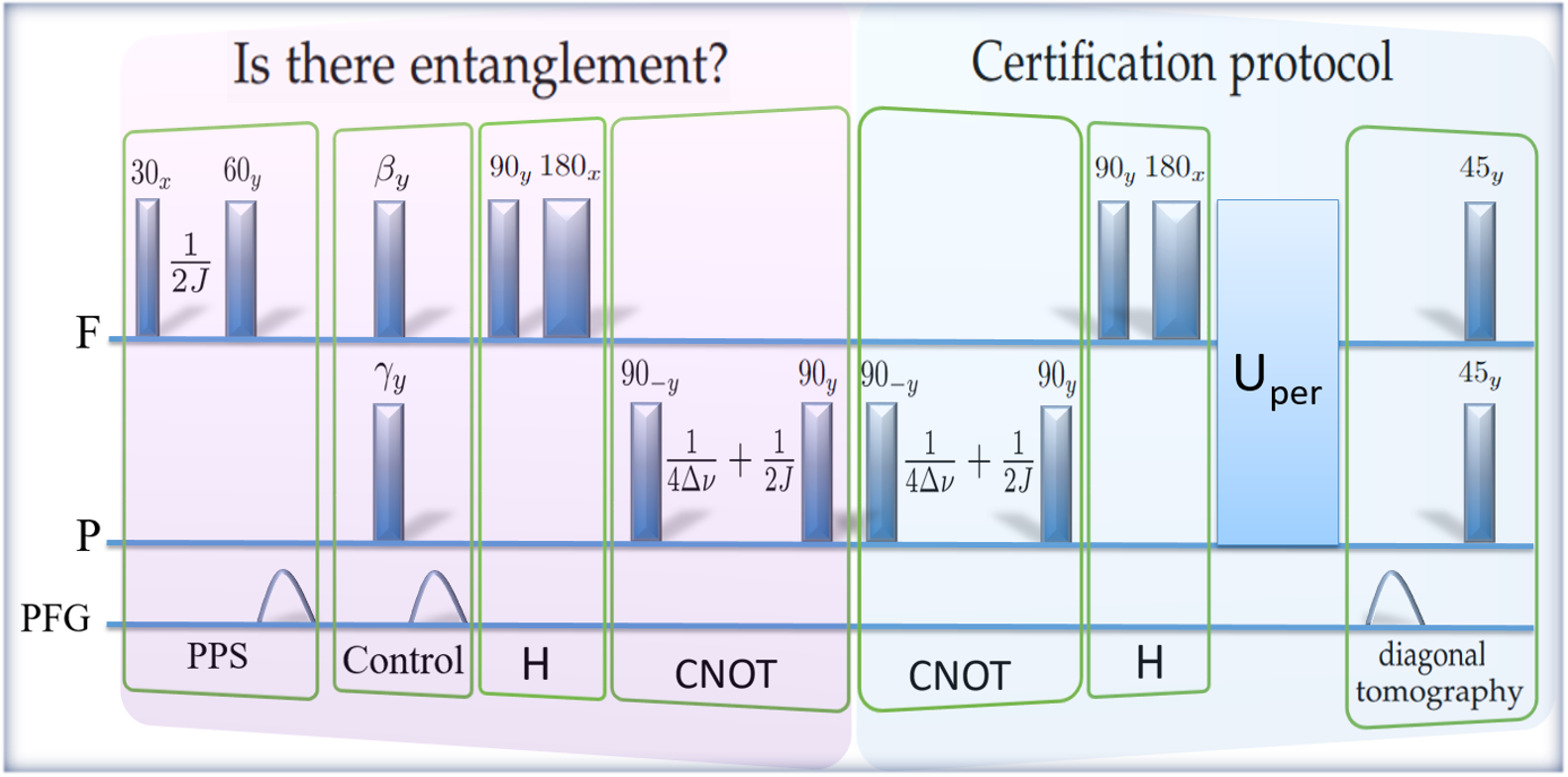}
	\caption{The NMR pulse sequence (same as in Fig. 2a of the main text) to prepare Bell diagonal state with control parameter $\beta$ and $\gamma$. 
		Here, the rectangles represent RF rotations, delays represent free-evolutions, PFG is the Pulsed-Field Gradient, and $\Delta \nu$ is the resonance offset of both $^{19}$F and $^{31}$P. 
		One of the operations [$U_\mathrm{per}\circ H\circ \mathrm{CNOT}$] takes the Bell diagonal state to its passive. }
	\label{sup:fig2}
\end{figure}
The NMR pulse sequence for the two-qubit Bell-diagonal state is shown in Fig. \ref{sup:fig2}.
~\\
\noindent \textbf{Preparation of PPS}\\
We start with the thermal state, which under high-field, high-temperature approximation reads as $\rho_{th} = \mathbf{I}/4 + \epsilon_P (\frac{\gamma_F}{\gamma_P}I_{1z} + I_{2z}$), where  $\epsilon_P = \gamma_PB_0/ 4k_BT$ and $\frac{\gamma_F}{\gamma_P}\approx 4/\sqrt{3}$.  The identity part is invariant under the unitary transformations, neither contributes to the NMR signal, and is ignored henceforth.  The PPS preparation (upto the scaling factor $\epsilon_P$ can be described as follows:
\begin{align}
	\frac{1}{4}\mathbf{I} + \epsilon_P \left[\frac{4}{\sqrt{3}}I_{1z} + I_{2z}\right]
	\xrightarrow[]{30_x^F}&
	\frac{1}{4}\mathbf{I} + \epsilon_P \left[\frac{4}{\sqrt{3}}\frac{\sqrt{3}}{2}I_{1z} - \frac{4}{\sqrt{3}}\frac{1}{2}I_{1y}+ I_{2z} \right]
	\nonumber \\
	\xrightarrow[]{1/2J}&
	\frac{1}{4}\mathbf{I} + \epsilon_P \left[2I_{1z} + \frac{4}{\sqrt{3}} I_{1x}I_{2z}+ I_{2z}\right]
	\nonumber \\
	\xrightarrow[]{60_y^F, ~ \mathrm{PFG}}&
	\frac{1}{4}\mathbf{I} + \epsilon_P \left[I_{1z} - 2I_{1z}I_{2z}+ I_{2z}
	\right]
	\nonumber \\
	&\equiv (1-\epsilon) \frac{1}{4}\mathbf{I} + \epsilon \proj{11}.
\end{align}
~\\
\noindent \textbf{Control of $p_{ij}$} \\
Using notations $C_\theta = \cos \theta$ ~\&~ $S_\theta = \sin \theta$,
\begin{align}
	&\frac{1}{4}\mathbf{I} + \epsilon_P \left[I_{1z} - 2I_{1z}I_{2z}+ I_{2z}
	\right]
	\xrightarrow[]{\beta_y^F,~\gamma_y^P,~\mathrm{PFG}}
	\frac{1}{4}\mathbf{I} + \epsilon_P \left[C_\beta I_{1z} - C_\beta C_\gamma 2I_{1z}I_{2z}+ C_\gamma I_{2z}
	\right]
	\nonumber \\
	&= \frac{1}{4}\mathbf{I} + \epsilon_P
	\frac{1}{2}\left[
	\begin{array}{cccc}
		C_\beta + C_\gamma - C_\beta C_\gamma & & & \\
		& C_\beta - C_\gamma + C_\beta C_\gamma & &  \\
		&& -C_\beta + C_\gamma + C_\beta C_\gamma &  \\
		&&& -C_\beta - C_\gamma - C_\beta C_\gamma \\
	\end{array}
	\right],
	\nonumber \\
	&=
	\frac{1}{4}\mathbf{I} + \epsilon_P
	\frac{1}{2}\left[
	\begin{array}{cccc}
		1-4p_{00} & & & \\
		& 1-4p_{01} & &  \\
		&& 1-4p_{10} &  \\
		&&& 1-4p_{11} \\
	\end{array}
	\right]
	\equiv
	\sum_{i,j \in \{0,1\}} p_{ij} \proj{ij} 
\end{align}
up to the uniform background population within the PPS paradigm. \\

\noindent \textbf{Hadamard operation}\\
The Hadamard operation on $F$ is realized by
\begin{align}
	U_{180_x}U_{90_y} = \exp(-i \pi I_x) \exp(-i \pi I_y/2) = H,~\mbox{up to a global phase}.
\end{align} \\

\noindent \textbf{Implementing CNOT} 
\begin{align}
	U_\mathrm{CNOT} &= 
	U_y U_{zz} U_{z12}^\dagger U_y^\dagger
	\nonumber \\
	&= \exp(-i \pi I_{2y}/2)
	\exp(-i \pi I_{1z}I_{2z})
	\exp(i \pi(I_{1z}+I_{2z})/2)
	\exp(i \pi I_{2y}/2).
\end{align}
In our experiment, local z-rotations $U_{z12}^\dagger$ was realized by introducing a temporary resonance offset of $\Delta \nu = 10000$ Hz for a duration $1/(4\Delta\nu) = 25\mu$s.
Similarly, the bilinear rotation $U_{zz}$ was realized by using the $J$-Hamiltonian $H_J = 2\pi J I_{1z}I_{2z}$ evolving for time $1/(2J)$ as indicated in the pulse sequence. \\

\noindent \textbf{Implementing $U_\mathrm{per}$}\\
The $U_\mathrm{per}$ operator has a sequence of nonlocal CNOT gates and local $90_y$ rotations as described in Eq. \ref{supeq:uper}.  These are implemented in the same way as described in the previous steps.\\

\noindent \textbf{Detection:} \\
The final 45 degree y-pulse allows the measurement of not only linear $I_{iz}$ terms, but also bilinear $I_{1z}I_{2z}$ term:
\begin{align}
	I_{iz} &\xrightarrow{45_y} I_{ix},
	\nonumber \\
	I_{iz}I_{iz} &\xrightarrow{45_y} 
	\frac{1}{\sqrt{2}} \left(I_{ix}I_{iz} + 
	I_{iz}I_{ix}
	\right),
\end{align}
both the cases leading to single-quantum observable NMR signals.  For $90_y$ pulse, the bilinear term becomes $I_{1x}I_{2x}$ which is a combination of zero- and two-quantum coherence and hence not directly observable.

\subsection{Experiment I-b: Two-qubit Werner state}
As of special interest, in this experiment we addresses entanglement certification in two-qubit Werner class of states. For that, we use a two-qubit homo-nuclear spin system of $5$-Bromothiophene-$2$-Carbaldehyde (BRTP) dissolved in Dimethyl Sulphoxide (DMSO) [see Fig. \ref{figs5}(a)]. All the experiments are carried out on a $500$ MHz Bruker NMR spectrometer at an ambient temperature of $300$ K. The Hamiltonian of the system can be written in terms of the internal part and the RF drive as 
\begin{subequations}
	\begin{align}
		H_{12}&= H_{12}^{\mathrm{int}} + 
		H_{12}^{\mathrm{RF}},~~ \mbox{where}\\ 
		H_{12}^{\mathrm{int}}&=  -  \left(\omega_H + \pi\Delta \nu \right) I_{1z} -  \left(\omega_H - \pi\Delta \nu \right) I_{2z} + 2\pi J I_{1z}I_{2z}\nonumber\\ 
		&\approx - \omega_H (I_{1z} + I_{2z}),~~ \mbox{for}~~ (J,\Delta \nu) << \omega_H,~~\mbox{and}\\
		H_{12}^{\mathrm{RF}}&=  \Omega_H(t) \left(I_{1x}+I_{2x}\right).
	\end{align}
\end{subequations}
Here $\omega_H$ denotes the Larmor frequency of proton, $\Delta \nu$ is the chemical shift difference between the two spins, $J$ denotes the scalar coupling constant, and $\Omega_H$ denotes the RF amplitude. In BRTP, $\omega_H = 500.2$ MHz, $\Delta\nu = 192$ Hz and $J = 4.01$ Hz.  Using appropriate RF pulses, we can realize various local rotations and an entangling gate can be implemented through scalar coupling [see Fig. \ref{figs5}(b)].  When applied on a pure superposition state, the entangling gate can produce a singlet pair. Intrusion of white noise via decoherence lead to a noisy singlet state belonging in Werner class 
\begin{align}
	\rho_\lambda:=\lambda\ket{\psi^-}\bra{\psi^-}+(1-\lambda)\mathbf{I}/4,~\lambda\in[0,1], ~~\mbox{where}~~ \ket{\psi^-} = \frac{1}{\sqrt{2}}\left(\ket{01}-\ket{10}\right). 
\end{align}
Now criterion (\textcolor{red}{2b} in manuscript) can certify entanglement for the parameter values $\lambda>1/3$, whereas a state independent thermodynamic criterion can detect entanglement for the parameter ranges $\lambda>\frac{  \omega_H}{2  \omega_H}= 1/2$.
\begin{figure}
	\centering
	\includegraphics[trim={0cm 0cm 0cm 0cm},clip,width=13cm]{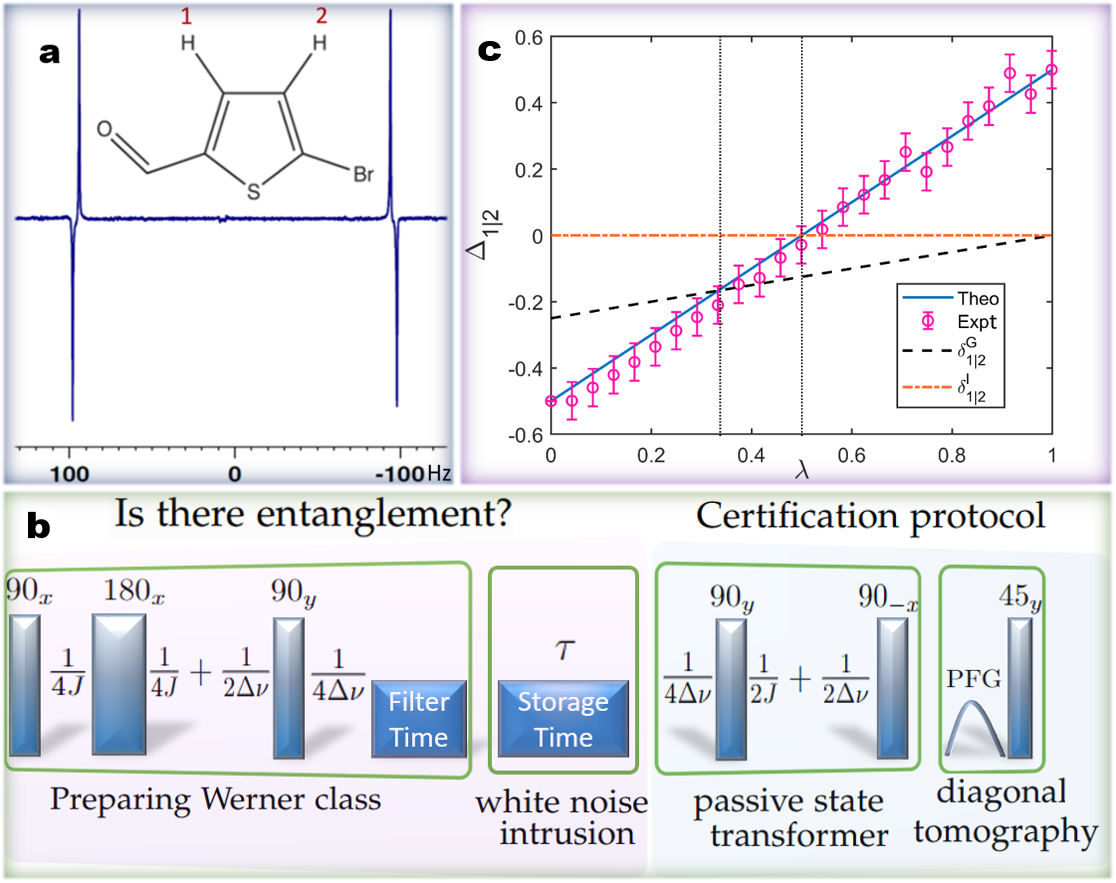}
	\caption{(a) The $^1$H spectrum corresponding to the singlet or Werner state of pair of hydrogen spins BRTP. (b) The NMR pulse sequence to produce Werner class states with controlled purity and then to certify the presence or absence of entanglement.  Here rectangles with $\Theta_{\beta} = e^{-i\Theta (I_{1\beta}+I_{2\beta})}$ represent RF rotations, delays represent free-evolutions, and PFG is the Pulsed-Field Gradient.  Here a short filter time of a few seconds is used to allow the decay of triplet state thereby retaining high-fidelity singlet state \cite{roy2010density}.  The final 45 degree y-pulse allows the measurement of not only linear $I_{iz}$ terms, but also bilinear $I_{1z}I_{2z}$ term.  (c) Plot of $\Delta_{1|2}$ in units of $  \omega_H$ versus purity of the Werner class. The vertical dotted lines indicate the purity threshold for entanglement: the left one marks purity $\lambda = 1/3$, above which $\Delta_{1|2}$ surpasses $\delta^G_{1|2}$ and the state becomes entangled; the right one marks $\lambda = 1/2$, corresponding to the state-independent bound $\delta^I_{1|2}$. The errorbar represents the estimated random error due to noise.}\label{figs5}
	\vspace{-.5cm}
\end{figure}
~\\
{\bf Preparation Step}\\
To prepare the Werner class of state in BRTP system, we start with thermal state, which under high-field, high-temperature approximation ($ \hbar \omega_H << k_BT$) reads as $\rho_{th} = \mathbf{I}/4 + \epsilon_H(I_{1z} + I_{2z})$, where $\epsilon_H =   \omega_H/ 4k_BT$. We then prepare the long-lived singlet state (LLS) by applying suitable pulse sequence as shown in the preparation part of Fig. \ref{figs5}(b),
\begin{align}
	\rho_{LLS} &= \frac{\mathbf{I}}{4} -\epsilon_H (I_{1x}I_{2x}+I_{1y}I_{2y}+I_{1z}I_{2z})\nonumber \\ 
	&= (1-\epsilon_H)\frac{\mathbf{I}}{4} +\epsilon_H\proj{\psi_-}.
	\label{eq:lls}
\end{align}
As the name suggests, LLS outlives any other nonequilibrium state \cite{carravetta2004long}. In BRTP, LLS decay constant $T_{LLS}$ is 12.9 s, which is considerably longer than the individual $T_1$ values (longitudinal relaxation time constants) of $6.1$ s and $8.2$ s.  Thus, after the filter time shown in Fig. \ref{figs5}(b), we obtain a clean state in the above form of Eq. (\ref{eq:lls}) \cite{roy2010density}.
Since LLS is isomorphic to a pure singlet, we invoke the paradigm of pseudo-pure states \cite{cory1997ensemble} and set $\epsilon_H=1$. Thus  our reference state at time $t=0$ is $\rho_{\lambda(0)} \equiv \proj{\psi^-}$.  During the storage time, white noise intrusion takes place and the singlet state gradually becomes more and more mixed Werner state, i.e., $\rho_{\lambda(\tau)} =  (1-\lambda(\tau))\mathbf{I}/4+\lambda(\tau)\ket{\psi^-}\bra{\psi^-}$, where $\lambda(\tau) =  e^{-\tau/T_{LLS}}$ is the purity factor of Werner class.\\\\
{\bf Entanglement certification step}\\
To evaluate global ergotropic work $W[\rho_{\lambda(\tau)}]$ we unitarily drive the system to its passive state by applying appropriate pulse sequence as shown in the $2^{nd}$ part of Fig.\ref{figs5}(b)
\begin{align}   &\rho_{\lambda(\tau)}=\eta(\tau)\mathbf{I}+\lambda(\tau)[\mathbf{I}/4-(I_{1x}I_{2x}+I_{1y}I_{2y}+I_{1z}I_{2z})]\nonumber\\&\xrightarrow[]{\frac{1}{4\Delta \nu}}\eta(\tau)\mathbf{I}+\lambda(\tau)[\mathbf{I}/4-(-I_{1x}I_{2y}+I_{1y}I_{2x}+I_{1z}I_{2z})]\nonumber\\
	&\xrightarrow[]{90_y^{(1,2)}}\eta(\tau)\mathbf{I}+\lambda(\tau)[\mathbf{I}/4-(I_{1z}I_{2y}-I_{1y}I_{2z}+I_{1z}I_{2z})]\nonumber\\
	&\xrightarrow[]{\frac{1}{2J}}\eta(\tau)\mathbf{I}+\lambda(\tau)[\mathbf{I}/4-(-I_{1x}+I_{2x}+2I_{1x}I_{2x})/2]\nonumber\\
	&\xrightarrow[]{\frac{1}{2\Delta \nu}}\eta(\tau)\mathbf{I}+\lambda(\tau)[\mathbf{I}/4-(I_{1y}+I_{2y}-2I_{1y}I_{2y})/2]\nonumber\\
	&\xrightarrow[]{90_{-x}^{(1,2)}}\eta(\tau)\mathbf{I}+\lambda(\tau)[\mathbf{I}/4-(-I_{1z}-I_{2z}-2I_{1z}I_{2z})/2]\nonumber\\
	&=\mbox{diag}\left(1-3\eta(\tau),\eta(\tau),\eta(\tau),\eta(\tau)\right)=\rho^p_\lambda(\tau),
\end{align}
where $\eta(\tau):=(1-\lambda(\tau))/4$, $90_{x/y}$ represents a $\pi/2$ rotation with rotation axis along $x/y$, ($1/2J$) and ($1/4J$) delays represent the free evolution under the scalar coupling, ($1/2\Delta\nu$) $\equiv$ $90^1_z90^2_{-z}$ and ($1/4\Delta\nu$) $\equiv$ $45^1_z45^2_{-z}$ represent the evolution under the chemical shift and PFG is the crusher gradient used in dephasing the coherence. For a finite storage time the system remains in the passive state having energy $E(\rho_\lambda^p(\tau))=\Tr(H_{AB}^{\mathrm{int}} \rho_\lambda^p(\tau)) \simeq - \omega_H \lambda(\tau)$, where $H_{AB}^{\mathrm{int}} \approx - \omega_H (I_{1z} + I_{2z}),~~ \mbox{for}~~ (J,\Delta \nu) << \omega_H$ and it can be estimated simply by measuring the diagonal elements of $\rho_\lambda^p(\tau)$ following the procedure shown in the last part of Fig. \ref{figs5}(b).

For singlet state, the marginal states are maximally mixed and no work can be extracted under a cyclic unitary operation. Therefore, using Eq. (\ref{thermquant}), we have $\Delta_{1|2}= \omega_H\lambda(\tau)- \omega_H/2$. Furthermore, since $m_1 =  \omega_H$ , $t_1 = t_2 = t_3 = (1-\lambda(\tau))/4$, we therefore have $\delta^G_{1|2}= (\lambda(\tau)-1)  \omega_H/4$.
In this case, the spectral independent bound turns out to be $\delta^I_{1|2} = 0$. In Fig.\ref{figs5} (c) theoretically estimated (solid line) as well as experimentally estimated (circles) values of  $\Delta_{1|2}$ are plotted against the purity factor $\lambda(\tau) = e^{-\tau/T_{LLS}}$. The vertical dashed and dotted-dashed lines correspond to $\delta^G_{1|2}$ and  $\delta^I_{1|2}$, respectively. For the region where $\Delta_{1|2}$ exceeds those bounds the state is entangled.

\subsection{Experiment II: Multi-qubit systems}
The STR Hamiltonian is given by
\begin{subequations}
	\begin{align}
		H_{STR} &= H_{STR}^{\mathrm{int}} + H_{STR}^{RF} \nonumber,~~\mbox{where,}\\
		H_{STR}^{\mathrm{int}} & =
		-  \omega_C I_{1z} - \omega_A \sum_{i=2}^N I_{iz} + 2\pi J I_{1z} \sum_{i=2}^N I_{iz}
		\\
		H_{STR}^{RF} & = \Omega_C(t) I_{1x} + \Omega_A(t) \sum_{i=2}^N I_{ix}.
	\end{align}
\end{subequations}
Here $\omega_C$ and $\omega_A$ are the Larmor frequencies of C and A spins respectively, $J$ is the scalar coupling constant between C and A spins (see Fig. \textcolor{red}{1}(b) in the manuscript), and $\Omega_C(t)$, $\Omega_A(t)$ are the time-dependent RF drives on C and A respectively.\\

{\bf Preparation step}\\
The initial thermal state of an STR is of the form
\begin{align*}
	\rho_{th} = \frac{\mathbf{I}}{2^N} + \epsilon_A
	\left(
	\frac{\gamma_C}{\gamma_A} 
	I_{1z} + \sum_{i=2}^N I_{iz}
	\right),
\end{align*}
where $\epsilon_A =   \gamma_A B_0 /(2^Nk_B T)$.  Starting from the thermal state a $\theta_y$ pulse on all the qubits followed by PFG results in the control over purity, i.e.,
\begin{align*}
	\rho_{th} & ~~\stackrel{\theta_y^{C,A}}{\longrightarrow}
	~~\stackrel{\mathrm{PFG}}{\longrightarrow} ~~\frac{\mathbf{I}}{2^N} + \epsilon_A \cos{\theta}
	\left(
	\frac{\gamma_C}{\gamma_A} 
	I_{1z} + \sum_{i=2}^N I_{iz}
	\right)\nonumber\\
	&\rightarrow
	\frac{\mathbf{I}}{2^N} +  \epsilon_A\cos{\theta}
	\sum_{i=2}^N I_{iz},
\end{align*}
where we have ignored the first-qubit component, since it does not lead to the GHZ state (and gets filtered away by the subsequent PFG pulses).
On further applying INEPT followed by CNOT, we obtain
\begin{align*}
	\rho_{Nq_h} &= \sum_{q_h} r_{q_h}\left[(
	1-\epsilon_A \cos{\theta})
	\frac{\mathbf{I}}{2^N} + \epsilon_A \cos{\theta} ~\rho_{q_h}
	\right],
\end{align*}
where $\sum_{q_h} r_{q_h} = 1$, with $\rho_{q_h}$ describing a set of entangled states of coherence order $q_h =
N - 2h$. Note that $q_h
\in \{N, N - 2, . . . ,-N + 2\}$ for $h \in
\{0, 1, . . . , N - 1\}$. Using a pair of PFGs, we can filter out GHZ class ($q_h=N$, $h=0$) 
by choosing lopsidedness 
\begin{align}   
	l_0 = 1+(N-1)\frac{\gamma_A}{\gamma_C}
\end{align}
\cite{mahesh2021star,shukla2014noon}. Invoking the paradigm of pseudopure states by setting $r_{N} = 1$ and $\epsilon_A=1$, we obtain
\begin{align}
	\rho_{N}(\lambda(\theta)) &= (1-\lambda(\theta))\frac{\mathbf{I}}{2^N} + \lambda(\theta) \proj{\psi_{N}},~~\mbox{where} ~~ \ket{\psi_{N}}:=\frac{1}{\sqrt{2}}(\ket{0}_C\ket{0}^{\otimes (N-1)}_A+\ket{1}_C\ket{1}^{\otimes (N-1)}_A)
	\label{nq0}
\end{align}
is the $N$-qubit GHZ state with purity control
$\lambda(\theta) = \cos{\theta}$.
This completes the preparation protocol of GHZ class of states.\\

{\bf Entanglement certification step}\\
The entanglement certification step involves transforming GHZ 
class of states to their corresponding passive states. This is accomplished by another CNOT gate followed by a $90_{y}$ pulse (which cancels with a subsequent readout pulse) as depicted in Fig.\ref{figs6}.  A pair of PFG pulses of relative ratio $-l_0$ is used to filter the GHZ state
for which
the passive state reads as
\begin{align*}
	\rho_{N}^p(\lambda(\theta))&= [1-\lambda(\theta)]\frac{\mathbf{I}}{2^N}+\lambda(\theta) \proj{0}_C \otimes \proj{0}_A^{\otimes (N-1)}.
\end{align*}
We have carried out the experiments on two systems as shown in Fig. \textcolor{red}{1} in the manuscript. This final spectrum corresponding to the GHZ class contains a single transition at frequency $-J(N-1)/2$ \cite{mahesh2021star}. Comparing the final intensity with that corresponding to $\theta=0$ we have a direct estimation of the purity factor $\lambda(\theta)$.

\subsubsection{Three-qubit system} 
\begin{figure}[t!]
	\centering
	\includegraphics[trim={0cm 0cm 0cm 0cm},clip=,width=10cm]{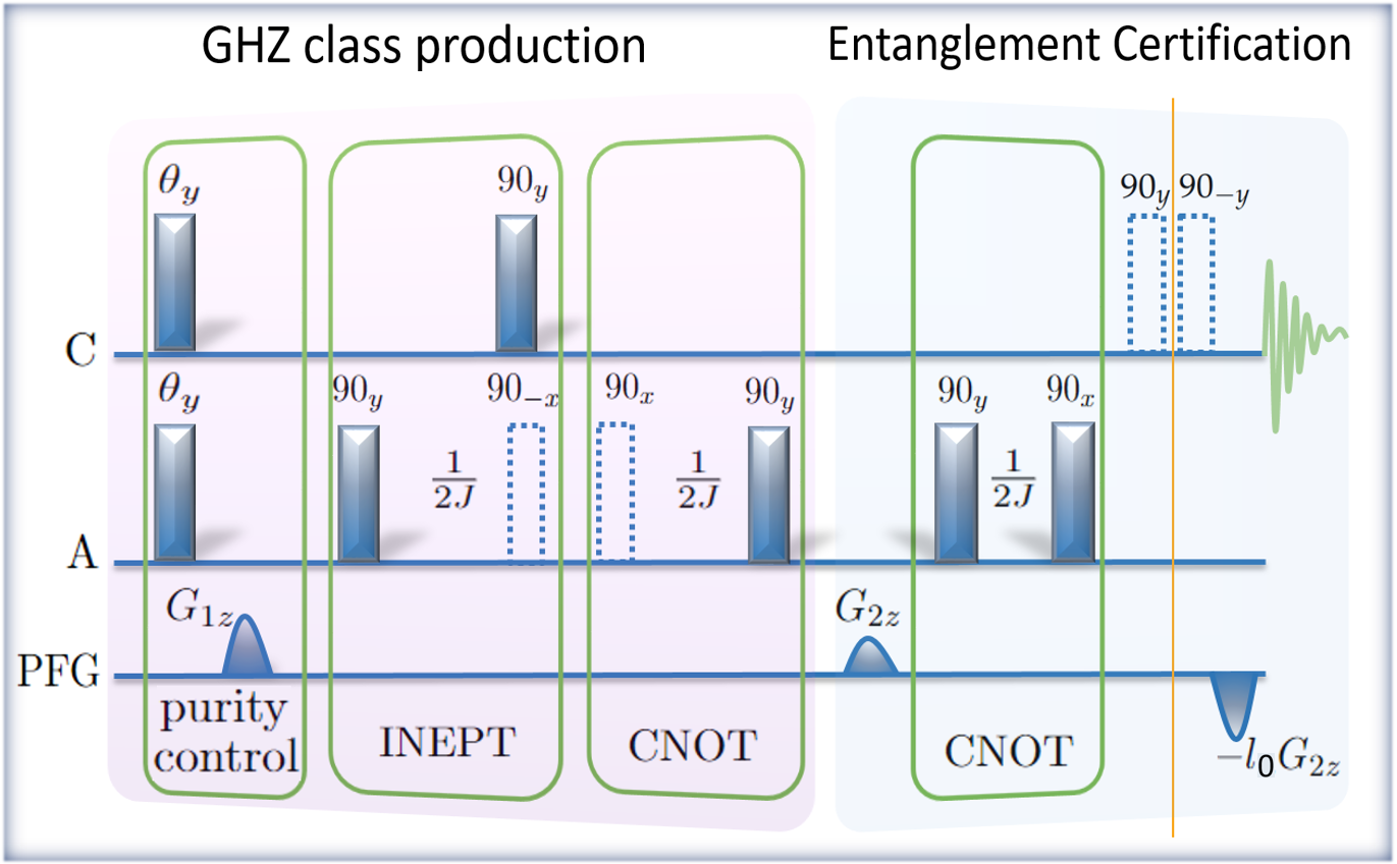}
	\caption{The NMR pulse sequence to prepare GHZ class states with controlled purity on an STR and then certify the presence or absence of entanglement.  The vertical line shows the instant when the passive state is created. The dashed pulses cancel each other.}\label{figs6}
\end{figure}

We have 
used FAN (fluoroacetonitril; Fig. \textcolor{red}{1} (d) in the manuscript), 
wherein $^{19}$F spin is the central qubit and two $^1$H spins are the satellite qubits.  Here $J_{FH} = 11.04$ Hz. Now following Eq. (\textcolor{red}{10a}) and Eq.(\textcolor{red}{11}) we have our Hamiltonian and initial global state as $H_{STR}^{\mathrm{int}} \approx -  \omega_F I_{1z} - \omega_H \sum_{i=2}^3 I_{iz}$ as $J<<(\omega_H,\omega_F)$ and $\rho_{N=3}(\lambda(\theta))$. Furthermore, using Eq.(\textcolor{red}{1}) we have 
\begin{align}
	\Delta_{1|23}&=W(\rho_{N=3}(\lambda(\theta)))-W(\rho_1)-E(\rho_{23})+E^{23}_g\nonumber \\
	&=[E(\rho_{N=3}(\lambda(\theta)))-E(\rho^p_{N=3}(\lambda(\theta)))]-[E(\rho_1)-E(\rho^p_1)]-E(\rho_{23})+E^{23}_g \nonumber \\
	&=-E(\rho^p_{N=3}(\lambda(\theta)))+E^{23}_g \nonumber \\
	&=\lambda(\theta)( \omega_F+2 \omega_H)/2-  \omega_H ,
	\label{delta1}
\end{align}
where $E(\rho^p_{N=3}(\lambda(\theta))) = tr(H_{STR}^{\mathrm{int}}\rho^p_{N=3}(\lambda(\theta))) = -\lambda(\theta)( \omega_F+2 \omega_H)/2$, $E^{23}_g = -  \omega_H$ is the ground state energy of the second and third qubit and $E(\rho^p_1) = 0$ . Now following Eq.(\ref{delta123}) and using values $\alpha_1 =  \omega_P$, $\alpha =  \omega_H$ and $t_1 = ...= t_7 = (1-\lambda(\theta))/8$ we have
\begin{align}
	\delta^G_{1|23} = (\lambda(\theta)-1)( \omega_H - 3 \omega_F/8)
	\label{deltab1}
\end{align}
and following Eq.(\ref{SIB3}) we have spectral independent bound as $\delta^I_{1|23} = 0$. Fig.\textcolor{red}{3}(b) in the manuscript shows both $\delta^G_{1|23}$ and $\delta^I_{1|23}$ marking $\Delta_{1|23}$ at $\lambda = 3/7$ and $\lambda = 0.68$ and for any value of $\lambda$ above them the state $\rho_{N=3}(\lambda(\theta))$ is said to be entangled.
\subsubsection{Ten-qubit system}
We have 
used TMP (trimethyl-phosphate; Fig. \textcolor{red}{1} (c) in the manuscript), 
wherein $^{31}$P spin is the central qubit and nine $^1$H spins are the satellite qubits.  Here $J_{PH} = 11.04$ Hz. Following Eq. (\textcolor{red}{10a}) and Eq. (\textcolor{red}{11}) we have our Hamiltonian and initial global state as $H_{STR}^{\mathrm{int}} \approx -  \omega_P I_{1z} - \omega_H \sum_{i=2}^{10} I_{iz}$ as $J<<(\omega_H,\omega_P)$ and $\rho_{N=10}(\lambda(\theta)) $ for GHZ class we have 
\begin{align}
	\Delta_{1|1^c}&=W(\rho_{N=10}(\lambda(\theta)))-W(\rho_1)-E(\rho_{1^c})+E^{1^c}_g\nonumber \\
	&=[E(\rho_{N=10}(\lambda(\theta)))-E(\rho^p_{N=10}(\lambda(\theta)))]-[E(\rho_1)-E(\rho^p_1)]-E(\rho_{1^c})+E^{1^c}_g \nonumber \\
	&=-E(\rho^p_{N=10}(\lambda(\theta)))+E^{1^c}_g \nonumber \\
	&=\lambda(\theta)( \omega_P+9 \omega_H)/2-9  \omega_H/2,     \label{delta1}
\end{align}
where $E(\rho^p_{N=10}(\lambda(\theta))) = tr(H_{STR}^{\mathrm{int}}\rho^p_{N=10}(\lambda(\theta))) = -\lambda(\theta)( \omega_P+9 \omega_H)/2$, $E^{1^c}_g = -9  \omega_H/2$ is the ground state energy of rest of the 9 qubits and $E(\rho^p_1) = 0$. Now following Eq.(\ref{delta10}) and using values $\alpha_1 =  \omega_P$, $\alpha =  \omega_H$ and $t_1 = t_2 =...= t_{2^{10}-1} = (1-\lambda(\theta))/2^{10}$ we have
\begin{align}
	\delta^G_{1|1^c} = (\lambda(\theta)-1)(9 \omega_H/2 - (2^9-1) \omega_P/2^{10})
\end{align}
and following Eq.(\ref{SIB10}) we have spectral independent bound as $\delta^I_{1|1^c} = 0$. Fig \textcolor{red}{3}(d) in the manuscript shows both $\delta^G_{1|1^c}$ and $\delta^I_{1|1^c}$ marking $\Delta_{1|1^c}$ at $\lambda = 0.499$ and $\lambda = 0.957$  and if any value of $\lambda$ is above the marked values then $\rho_{N=10}(\lambda(\theta)) $ is entangled.

\subsection{Experiment III: Global vs Global-Local Separability bounds}
In this experiment we aim to compare the efficacy of entanglement certification of three-qubit noisy GHZ state through Global-Local and Global separability bounds as stated in Theorem 1. In our earlier experiments when we were evaluation $\Delta_{X|X^c}$, we considered the $X$ part consisting only one subsystem. Accordingly the Global-Local separability bound and Global separability bound become identical. In this case we consider a three-qubit system with $X$ consisting of two qubits. In experiment we use Di-Bromo Fuloromethane (DBFM) [see Fig. \textcolor{red}{2}(e) in manuscript] as the three-qubit system. The Hamiltonian of the system can be written in terms of the internal part and the RF drive as 
\begin{subequations}
	\begin{align}
		H_{123}&= H_{123}^{\mathrm{int}} + 
		H_{123}^{\mathrm{RF}},~~\mbox{where}\\
		H_{123}^{\mathrm{int}}&=  -  \omega_H I_{1z} -  \omega_C I_{2z}-  \omega_F I_{3z} + 2\pi J_{HC} I_{1z}I_{2z}+ 2\pi J_{CF} I_{2z}I_{3z}+ 2\pi J_{HF} I_{1z}I_{3z} \nonumber\\
		H_{123}^{\mathrm{RF}}&=  \Omega_H(t) I_{1x}+ \Omega_C(t) I_{2x}+ \Omega_F(t) I_{3x},\\
		H_{12}^{\mathrm{int}}&=  -  \omega_H I_{1z} -  \omega_C I_{2z} + 2\pi J_{HC} I_{1z}I_{2z} ,~~\mbox{where}~~X\equiv12.
	\end{align}
\end{subequations}
\begin{figure}[t!]
	\centering
	\includegraphics[trim={0cm 0cm 0cm 0cm},clip=,width=10cm]{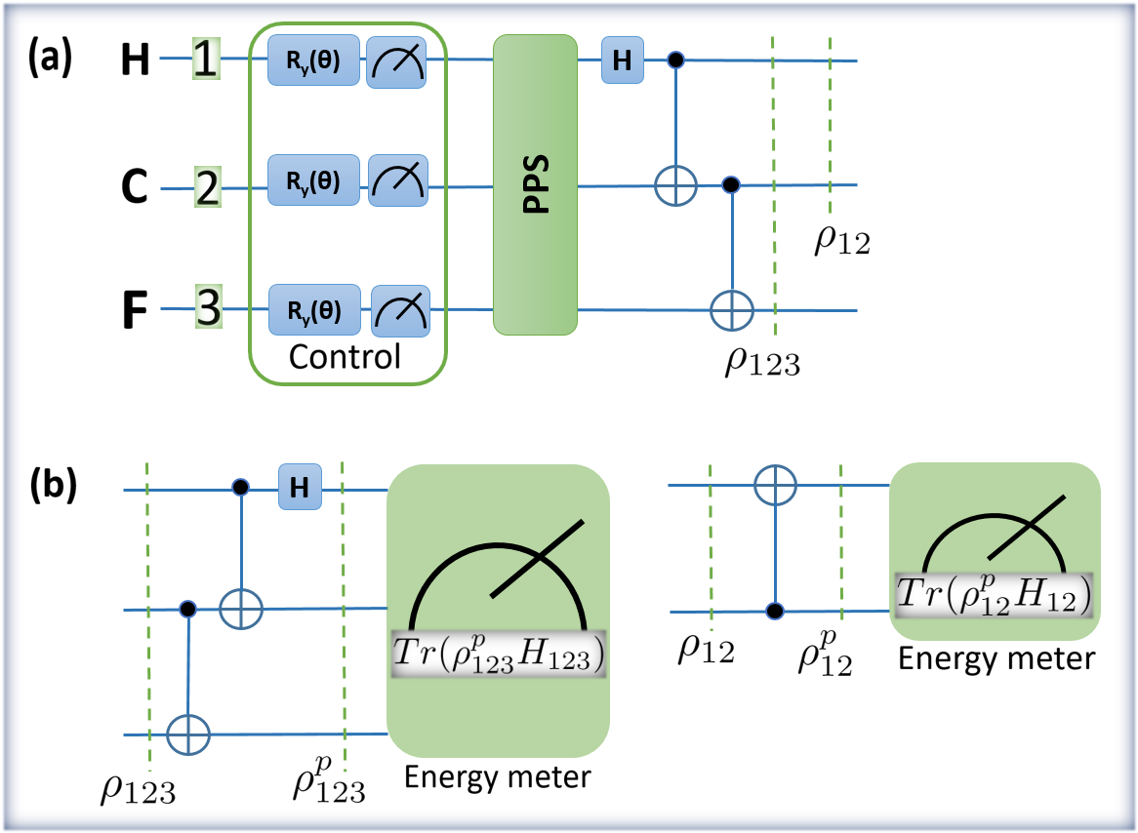}
	\caption{(a) Quantum circuit to prepare 3-qubit GHZ state $\rho_{123}$ with control $\theta$, here PPS stands for pulse sequence to prepare psuedo-pure state from thermal equilibrium.(b) Quantum circuit for global and local passive state conversion and to measure their corresponding energies.}  \label{figs7}
\end{figure}
Here, $H,C,~\& F$ correspond to first, second, and third system respectively.\\\\
{\bf Preparation step}\\
We start with the initial thermal state 
\begin{align}
	\sigma^{th}_{123} =\frac{1}{8}\mathbf{I}+ \epsilon_H \left(I_{2z}+\frac{\gamma_C}{\gamma_H}I_{2z} + \frac{\gamma_F}{\gamma_H}I_{3z}\right),~~\mbox{where}~~ \epsilon_H =   \gamma_HB_0/ 4k_BT.  
\end{align}
After application of the control sequence as shown in Fig. \ref{figs7}(a) and the PPS pulse sequence \cite{krithika2021observation}, we obtain the state 
\begin{align}
	\sigma^\epsilon_{123}(\lambda) = \frac{1-\epsilon\lambda}{8}\mathbf{I}+ \epsilon\lambda\ket{000}\bra{000},   
\end{align}
where $\lambda =\cos\theta$. Within the PPS paradigm taking $\epsilon = 1$ we have 
\begin{align}
	\sigma_{123}(\lambda)= \frac{1-\lambda}{8}\mathbf{I}+\lambda\ket{000}\bra{000}.
\end{align}
Now, consecutively applying Hadamard on $H$,  CNOT with $H$ being the control and $C$ being the target, and CNOT ($C$ control and $F$ target) [Fig. \ref{figs7}(a)], we prepare the noisy GHZ state
\begin{align}
	\rho_{123}(\lambda) &=\frac{1-\lambda}{8}\mathbf{I}+ \lambda\ket{\psi}_{GHZ}\bra{\psi},
\end{align}
where $\ket{\psi}_{GHZ}:=(\ket{000}+\ket{111})/\sqrt{2}$.\\\\
{\bf Entanglement certification step}\\
To evaluate the quantity $\Delta_{12|3}$ we will first evaluate $W(\rho_{123}(\lambda))$. For that we unitarily evolve $\rho_{123}(\lambda)$ to the corresponding lowest energy state. While for a generic state the optimization over the unitary is a tedious task, however for this class of state it can be achieved by applying a single unitary operation [see Fig. \ref{figs7}(b)]. The resulting state reads as 
\begin{align}
	\rho^p_{123}(\lambda) &=\frac{1-\lambda}{8}\mathbf{I}+ \lambda\ket{000}\bra{000},
	\label{3q_global}
\end{align}
which in turn will yield $W(\rho_{123}(\lambda))=E(\rho_{123}(\lambda))-E(\rho^p_{123}(\lambda))$. To evaluate $W_{[12]}(\rho_{12})$ we evolve the subsystem $12$ to the corresponding lowest energy state $\rho^p_{12}$ [see Fig. \ref{figs7}(c)]. We have
\begin{subequations}
	\begin{align}
		\rho_{12}&= \frac{1+\lambda}{4}\ket{00}\bra{00} +\frac{1-\lambda}{4}\ket{01}\bra{10} +\frac{1-\lambda}{4}\ket{10}\bra{01}+\frac{1+\lambda}{4}\ket{11}\bra{11}, \\ 
		\rho^p_{12}&= \frac{1+\lambda}{4}\ket{00}\bra{00} +\frac{1+\lambda}{4}\ket{01}\bra{01} +\frac{1-\lambda}{4}\ket{10}\bra{10}+\frac{1-\lambda}{4}\ket{11}\bra{11},
		\label{3q_local}
	\end{align}
\end{subequations}
where $\rho_{12}=\tr_3[\rho_{123}]$. Since the system Hamiltonian is known we can eventually calculate energy of the third subsystem $E(\rho_3)$, its ground state energy $E^3_g$, and obtain $\Delta^{Exp}_{12|3}$ experimentally. We can also calculate this quantity theoretically. For the system in consideration with $(J_{HC},J_{CF},J_{HF})<<(\omega_H,\omega_C,\omega_F)$ we have 
\begin{align}
	\left\{\!\begin{aligned}
		x_0=\frac{1+\lambda}{4},~~x_1=\frac{1+\lambda}{4},~~x_2=\frac{1-\lambda}{4},~~x_3=\frac{1-\lambda}{4},~~
		t_0=\frac{1+7\lambda}{8},\\
		~\{t_1, \cdots,t_7\} = \frac{1-\lambda}{8},~~m_0 = 0,~~m_1 =  \omega_C,~~ m_2 =~~  \omega_H,~~
		m_3 =  (\omega_H+\omega_C),\\ ~~n_0 = 0,~~ n_1 =  \omega_C, ~~
		n_2 =  \omega_F,~~ n_3 =  \omega_H,~~ n_4 =  (\omega_C+\omega_F),\\ n_5 =  (\omega_C+\omega_H),~~ n_6 =  (\omega_F+\omega_H),~~ n_7 =  (\omega_C+\omega_F+\omega_H),\\
		\omega_H = 500\mbox{MHz},~~\omega_C = 125.721\mbox{MHz},~~\omega_F= 470.385\mbox{MHz}
	\end{aligned}\right\}.   
\end{align}
These yield us
\begin{subequations}
	\begin{align}
		\Delta_{12|3} &= \sum^3_{j=1} m_jx_j-\sum^7_{j=1}n_jt_j=  (\omega_C+\omega_F)\lambda/2- \omega_F/2,\\
		\delta^{GL}_{12|3} &= \sum^3_{j=1} (m_j-m_1)x_j+\sum^7_{j=1}(m_1-n_j)t_j=  (4\omega_F -\omega_C)\lambda/8+ (\omega_C -4\omega_F)/8, \\
		\delta^G_{12|3} &= \sum^2_{j=1} (m_j-n_j)t_j+\sum^7_{j=3}(m_3-n_j)t_j=  (\omega_H+\omega_C-2\omega_F)(1-\lambda)/4,\\
		\delta^I_{12|3}&= (2\omega_H-2\omega_F +\omega_C)/5= 0.0743 \omega_H.
	\end{align}
\end{subequations}
These quantities along with $\Delta^{Exp}_{12|3}$ are plotted in Fig.4 of the main manuscript.

\twocolumngrid


\begin{thebibliography}{86}%
	\makeatletter
	\providecommand \@ifxundefined [1]{%
		\@ifx{#1\undefined}
	}%
	\providecommand \@ifnum [1]{%
		\ifnum #1\expandafter \@firstoftwo
		\else \expandafter \@secondoftwo
		\fi
	}%
	\providecommand \@ifx [1]{%
		\ifx #1\expandafter \@firstoftwo
		\else \expandafter \@secondoftwo
		\fi
	}%
	\providecommand \natexlab [1]{#1}%
	\providecommand \enquote  [1]{``#1''}%
	\providecommand \bibnamefont  [1]{#1}%
	\providecommand \bibfnamefont [1]{#1}%
	\providecommand \citenamefont [1]{#1}%
	\providecommand \href@noop [0]{\@secondoftwo}%
	\providecommand \href [0]{\begingroup \@sanitize@url \@href}%
	\providecommand \@href[1]{\@@startlink{#1}\@@href}%
	\providecommand \@@href[1]{\endgroup#1\@@endlink}%
	\providecommand \@sanitize@url [0]{\catcode `\\12\catcode `\$12\catcode
		`\&12\catcode `\#12\catcode `\^12\catcode `\_12\catcode `\%12\relax}%
	\providecommand \@@startlink[1]{}%
	\providecommand \@@endlink[0]{}%
	\providecommand \url  [0]{\begingroup\@sanitize@url \@url }%
	\providecommand \@url [1]{\endgroup\@href {#1}{\urlprefix }}%
	\providecommand \urlprefix  [0]{URL }%
	\providecommand \Eprint [0]{\href }%
	\providecommand \doibase [0]{https://doi.org/}%
	\providecommand \selectlanguage [0]{\@gobble}%
	\providecommand \bibinfo  [0]{\@secondoftwo}%
	\providecommand \bibfield  [0]{\@secondoftwo}%
	\providecommand \translation [1]{[#1]}%
	\providecommand \BibitemOpen [0]{}%
	\providecommand \bibitemStop [0]{}%
	\providecommand \bibitemNoStop [0]{.\EOS\space}%
	\providecommand \EOS [0]{\spacefactor3000\relax}%
	\providecommand \BibitemShut  [1]{\csname bibitem#1\endcsname}%
	\let\auto@bib@innerbib\@empty
	\bibitem [{\citenamefont {Einstein}\ \emph {et~al.}(1935)\citenamefont
		{Einstein}, \citenamefont {Podolsky},\ and\ \citenamefont
		{Rosen}}]{Einstein1935}%
	\BibitemOpen
	\bibfield  {author} {\bibinfo {author} {\bibfnamefont {A.}~\bibnamefont
			{Einstein}}, \bibinfo {author} {\bibfnamefont {B.}~\bibnamefont {Podolsky}},\
		and\ \bibinfo {author} {\bibfnamefont {N.}~\bibnamefont {Rosen}},\ }\bibfield
	{title} {\bibinfo {title} {Can quantum-mechanical description of physical
			reality be considered complete?},\ }\href
	{https://doi.org/10.1103/PhysRev.47.777} {\bibfield  {journal} {\bibinfo
			{journal} {Phys. Rev.}\ }\textbf {\bibinfo {volume} {47}},\ \bibinfo {pages}
		{777} (\bibinfo {year} {1935})}\BibitemShut {NoStop}%
	\bibitem [{\citenamefont {Bohr}(1935)}]{Bohr1935}%
	\BibitemOpen
	\bibfield  {author} {\bibinfo {author} {\bibfnamefont {N.}~\bibnamefont
			{Bohr}},\ }\bibfield  {title} {\bibinfo {title} {Can quantum-mechanical
			description of physical reality be considered complete?},\ }\href
	{https://doi.org/10.1103/PhysRev.48.696} {\bibfield  {journal} {\bibinfo
			{journal} {Phys. Rev.}\ }\textbf {\bibinfo {volume} {48}},\ \bibinfo {pages}
		{696} (\bibinfo {year} {1935})}\BibitemShut {NoStop}%
	\bibitem [{\citenamefont {Schr\"{o}dinger}(1936)}]{Schrdinger1936}%
	\BibitemOpen
	\bibfield  {author} {\bibinfo {author} {\bibfnamefont {E.}~\bibnamefont
			{Schr\"{o}dinger}},\ }\bibfield  {title} {\bibinfo {title} {Probability
			relations between separated systems},\ }\href
	{https://doi.org/10.1017/s0305004100019137} {\bibfield  {journal} {\bibinfo
			{journal} {Math. Proc. Cambridge Philos. Soc.}\ }\textbf {\bibinfo {volume}
			{32}},\ \bibinfo {pages} {446} (\bibinfo {year} {1936})}\BibitemShut
	{NoStop}%
	\bibitem [{\citenamefont {Bell}(1966)}]{Bell1966}%
	\BibitemOpen
	\bibfield  {author} {\bibinfo {author} {\bibfnamefont {J.~S.}\ \bibnamefont
			{Bell}},\ }\bibfield  {title} {\bibinfo {title} {On the problem of hidden
			variables in quantum mechanics},\ }\href
	{https://doi.org/10.1103/RevModPhys.38.447} {\bibfield  {journal} {\bibinfo
			{journal} {Rev. Mod. Phys.}\ }\textbf {\bibinfo {volume} {38}},\ \bibinfo
		{pages} {447} (\bibinfo {year} {1966})}\BibitemShut {NoStop}%
	\bibitem [{\citenamefont {Ekert}(1991)}]{Ekert1991}%
	\BibitemOpen
	\bibfield  {author} {\bibinfo {author} {\bibfnamefont {A.~K.}\ \bibnamefont
			{Ekert}},\ }\bibfield  {title} {\bibinfo {title} {Quantum cryptography based
			on bell's theorem},\ }\href {https://doi.org/10.1103/PhysRevLett.67.661}
	{\bibfield  {journal} {\bibinfo  {journal} {Phys. Rev. Lett.}\ }\textbf
		{\bibinfo {volume} {67}},\ \bibinfo {pages} {661} (\bibinfo {year}
		{1991})}\BibitemShut {NoStop}%
	\bibitem [{\citenamefont {Bennett}\ and\ \citenamefont
		{Wiesner}(1992)}]{Bennett1992}%
	\BibitemOpen
	\bibfield  {author} {\bibinfo {author} {\bibfnamefont {C.~H.}\ \bibnamefont
			{Bennett}}\ and\ \bibinfo {author} {\bibfnamefont {S.~J.}\ \bibnamefont
			{Wiesner}},\ }\bibfield  {title} {\bibinfo {title} {Communication via one-
			and two-particle operators on einstein-podolsky-rosen states},\ }\href
	{https://doi.org/10.1103/PhysRevLett.69.2881} {\bibfield  {journal} {\bibinfo
			{journal} {Phys. Rev. Lett.}\ }\textbf {\bibinfo {volume} {69}},\ \bibinfo
		{pages} {2881} (\bibinfo {year} {1992})}\BibitemShut {NoStop}%
	\bibitem [{\citenamefont {Bennett}\ \emph {et~al.}(1993)\citenamefont
		{Bennett}, \citenamefont {Brassard}, \citenamefont {Cr\'epeau}, \citenamefont
		{Jozsa}, \citenamefont {Peres},\ and\ \citenamefont
		{Wootters}}]{Bennett1993}%
	\BibitemOpen
	\bibfield  {author} {\bibinfo {author} {\bibfnamefont {C.~H.}\ \bibnamefont
			{Bennett}}, \bibinfo {author} {\bibfnamefont {G.}~\bibnamefont {Brassard}},
		\bibinfo {author} {\bibfnamefont {C.}~\bibnamefont {Cr\'epeau}}, \bibinfo
		{author} {\bibfnamefont {R.}~\bibnamefont {Jozsa}}, \bibinfo {author}
		{\bibfnamefont {A.}~\bibnamefont {Peres}},\ and\ \bibinfo {author}
		{\bibfnamefont {W.~K.}\ \bibnamefont {Wootters}},\ }\bibfield  {title}
	{\bibinfo {title} {Teleporting an unknown quantum state via dual classical
			and einstein-podolsky-rosen channels},\ }\href
	{https://doi.org/10.1103/PhysRevLett.70.1895} {\bibfield  {journal} {\bibinfo
			{journal} {Phys. Rev. Lett.}\ }\textbf {\bibinfo {volume} {70}},\ \bibinfo
		{pages} {1895} (\bibinfo {year} {1993})}\BibitemShut {NoStop}%
	\bibitem [{\citenamefont {Deutsch}(1985)}]{Deutsch1985}%
	\BibitemOpen
	\bibfield  {author} {\bibinfo {author} {\bibfnamefont {D.}~\bibnamefont
			{Deutsch}},\ }\bibfield  {title} {\bibinfo {title} {Quantum theory, the
			church{\textendash}turing principle and the universal quantum computer},\
	}\href {https://doi.org/10.1098/rspa.1985.0070} {\bibfield  {journal}
		{\bibinfo  {journal} {Proc. Math. Phys. Eng. Sci.}\ }\textbf {\bibinfo
			{volume} {400}},\ \bibinfo {pages} {97} (\bibinfo {year} {1985})}\BibitemShut
	{NoStop}%
	\bibitem [{\citenamefont {Yurke}(1986)}]{Yurke1986}%
	\BibitemOpen
	\bibfield  {author} {\bibinfo {author} {\bibfnamefont {B.}~\bibnamefont
			{Yurke}},\ }\bibfield  {title} {\bibinfo {title} {Input states for
			enhancement of fermion interferometer sensitivity},\ }\href
	{https://doi.org/10.1103/PhysRevLett.56.1515} {\bibfield  {journal} {\bibinfo
			{journal} {Phys. Rev. Lett.}\ }\textbf {\bibinfo {volume} {56}},\ \bibinfo
		{pages} {1515} (\bibinfo {year} {1986})}\BibitemShut {NoStop}%
	\bibitem [{\citenamefont {Giovannetti}\ \emph {et~al.}(2004)\citenamefont
		{Giovannetti}, \citenamefont {Lloyd},\ and\ \citenamefont
		{Maccone}}]{Giovannetti2004}%
	\BibitemOpen
	\bibfield  {author} {\bibinfo {author} {\bibfnamefont {V.}~\bibnamefont
			{Giovannetti}}, \bibinfo {author} {\bibfnamefont {S.}~\bibnamefont {Lloyd}},\
		and\ \bibinfo {author} {\bibfnamefont {L.}~\bibnamefont {Maccone}},\
	}\bibfield  {title} {\bibinfo {title} {Quantum-enhanced measurements: Beating
			the standard quantum limit},\ }\href
	{https://doi.org/10.1126/science.1104149} {\bibfield  {journal} {\bibinfo
			{journal} {Science}\ }\textbf {\bibinfo {volume} {306}},\ \bibinfo {pages}
		{1330} (\bibinfo {year} {2004})}\BibitemShut {NoStop}%
	\bibitem [{\citenamefont {Xia}\ \emph {et~al.}(2023)\citenamefont {Xia},
		\citenamefont {Agrawal}, \citenamefont {Pluchar}, \citenamefont {Brady},
		\citenamefont {Liu}, \citenamefont {Zhuang}, \citenamefont {Wilson},\ and\
		\citenamefont {Zhang}}]{Xia2023}%
	\BibitemOpen
	\bibfield  {author} {\bibinfo {author} {\bibfnamefont {Y.}~\bibnamefont
			{Xia}}, \bibinfo {author} {\bibfnamefont {A.~R.}\ \bibnamefont {Agrawal}},
		\bibinfo {author} {\bibfnamefont {C.~M.}\ \bibnamefont {Pluchar}}, \bibinfo
		{author} {\bibfnamefont {A.~J.}\ \bibnamefont {Brady}}, \bibinfo {author}
		{\bibfnamefont {Z.}~\bibnamefont {Liu}}, \bibinfo {author} {\bibfnamefont
			{Q.}~\bibnamefont {Zhuang}}, \bibinfo {author} {\bibfnamefont {D.~J.}\
			\bibnamefont {Wilson}},\ and\ \bibinfo {author} {\bibfnamefont
			{Z.}~\bibnamefont {Zhang}},\ }\bibfield  {title} {\bibinfo {title}
		{Entanglement-enhanced optomechanical sensing},\ }\bibfield  {journal}
	{\bibinfo  {journal} {Nature Photonics}\ }\href
	{https://doi.org/10.1038/s41566-023-01178-0} {10.1038/s41566-023-01178-0}
	(\bibinfo {year} {2023})\BibitemShut {NoStop}%
	\bibitem [{\citenamefont {Horodecki}\ \emph {et~al.}(2009)\citenamefont
		{Horodecki}, \citenamefont {Horodecki}, \citenamefont {Horodecki},\ and\
		\citenamefont {Horodecki}}]{Horodecki2009}%
	\BibitemOpen
	\bibfield  {author} {\bibinfo {author} {\bibfnamefont {R.}~\bibnamefont
			{Horodecki}}, \bibinfo {author} {\bibfnamefont {P.}~\bibnamefont
			{Horodecki}}, \bibinfo {author} {\bibfnamefont {M.}~\bibnamefont
			{Horodecki}},\ and\ \bibinfo {author} {\bibfnamefont {K.}~\bibnamefont
			{Horodecki}},\ }\bibfield  {title} {\bibinfo {title} {Quantum entanglement},\
	}\href {https://doi.org/10.1103/RevModPhys.81.865} {\bibfield  {journal}
		{\bibinfo  {journal} {Rev. Mod. Phys.}\ }\textbf {\bibinfo {volume} {81}},\
		\bibinfo {pages} {865} (\bibinfo {year} {2009})}\BibitemShut {NoStop}%
	\bibitem [{\citenamefont {D\"ur}\ \emph {et~al.}(2000)\citenamefont {D\"ur},
		\citenamefont {Vidal},\ and\ \citenamefont {Cirac}}]{Dur2000}%
	\BibitemOpen
	\bibfield  {author} {\bibinfo {author} {\bibfnamefont {W.}~\bibnamefont
			{D\"ur}}, \bibinfo {author} {\bibfnamefont {G.}~\bibnamefont {Vidal}},\ and\
		\bibinfo {author} {\bibfnamefont {J.~I.}\ \bibnamefont {Cirac}},\ }\bibfield
	{title} {\bibinfo {title} {Three qubits can be entangled in two inequivalent
			ways},\ }\href {https://doi.org/10.1103/PhysRevA.62.062314} {\bibfield
		{journal} {\bibinfo  {journal} {Phys. Rev. A}\ }\textbf {\bibinfo {volume}
			{62}},\ \bibinfo {pages} {062314} (\bibinfo {year} {2000})}\BibitemShut
	{NoStop}%
	\bibitem [{\citenamefont {Verstraete}\ \emph {et~al.}(2002)\citenamefont
		{Verstraete}, \citenamefont {Dehaene}, \citenamefont {De~Moor},\ and\
		\citenamefont {Verschelde}}]{Verstraete2002}%
	\BibitemOpen
	\bibfield  {author} {\bibinfo {author} {\bibfnamefont {F.}~\bibnamefont
			{Verstraete}}, \bibinfo {author} {\bibfnamefont {J.}~\bibnamefont {Dehaene}},
		\bibinfo {author} {\bibfnamefont {B.}~\bibnamefont {De~Moor}},\ and\ \bibinfo
		{author} {\bibfnamefont {H.}~\bibnamefont {Verschelde}},\ }\bibfield  {title}
	{\bibinfo {title} {Four qubits can be entangled in nine different ways},\
	}\href {https://doi.org/10.1103/PhysRevA.65.052112} {\bibfield  {journal}
		{\bibinfo  {journal} {Phys. Rev. A}\ }\textbf {\bibinfo {volume} {65}},\
		\bibinfo {pages} {052112} (\bibinfo {year} {2002})}\BibitemShut {NoStop}%
	\bibitem [{\citenamefont {Cirac}\ \emph {et~al.}(1997)\citenamefont {Cirac},
		\citenamefont {Zoller}, \citenamefont {Kimble},\ and\ \citenamefont
		{Mabuchi}}]{Cirac1997}%
	\BibitemOpen
	\bibfield  {author} {\bibinfo {author} {\bibfnamefont {J.~I.}\ \bibnamefont
			{Cirac}}, \bibinfo {author} {\bibfnamefont {P.}~\bibnamefont {Zoller}},
		\bibinfo {author} {\bibfnamefont {H.~J.}\ \bibnamefont {Kimble}},\ and\
		\bibinfo {author} {\bibfnamefont {H.}~\bibnamefont {Mabuchi}},\ }\bibfield
	{title} {\bibinfo {title} {Quantum state transfer and entanglement
			distribution among distant nodes in a quantum network},\ }\href
	{https://doi.org/10.1103/PhysRevLett.78.3221} {\bibfield  {journal} {\bibinfo
			{journal} {Phys. Rev. Lett.}\ }\textbf {\bibinfo {volume} {78}},\ \bibinfo
		{pages} {3221} (\bibinfo {year} {1997})}\BibitemShut {NoStop}%
	\bibitem [{\citenamefont {Raussendorf}\ and\ \citenamefont
		{Briegel}(2001)}]{Raussendorf2001}%
	\BibitemOpen
	\bibfield  {author} {\bibinfo {author} {\bibfnamefont {R.}~\bibnamefont
			{Raussendorf}}\ and\ \bibinfo {author} {\bibfnamefont {H.~J.}\ \bibnamefont
			{Briegel}},\ }\bibfield  {title} {\bibinfo {title} {A one-way quantum
			computer},\ }\href {https://doi.org/10.1103/PhysRevLett.86.5188} {\bibfield
		{journal} {\bibinfo  {journal} {Phys. Rev. Lett.}\ }\textbf {\bibinfo
			{volume} {86}},\ \bibinfo {pages} {5188} (\bibinfo {year}
		{2001})}\BibitemShut {NoStop}%
	\bibitem [{\citenamefont {Arrighi}\ and\ \citenamefont
		{Salvail}(2006)}]{Arrighi2006}%
	\BibitemOpen
	\bibfield  {author} {\bibinfo {author} {\bibfnamefont {P.}~\bibnamefont
			{Arrighi}}\ and\ \bibinfo {author} {\bibfnamefont {L.}~\bibnamefont
			{Salvail}},\ }\bibfield  {title} {\bibinfo {title} {Blind quantum
			computation},\ }\href {https://doi.org/10.1142/s0219749906002171} {\bibfield
		{journal} {\bibinfo  {journal} {Int. J. Quantum Inf.}\ }\textbf {\bibinfo
			{volume} {04}},\ \bibinfo {pages} {883} (\bibinfo {year} {2006})}\BibitemShut
	{NoStop}%
	\bibitem [{\citenamefont {Kimble}(2008)}]{Kimble2008}%
	\BibitemOpen
	\bibfield  {author} {\bibinfo {author} {\bibfnamefont {H.~J.}\ \bibnamefont
			{Kimble}},\ }\bibfield  {title} {\bibinfo {title} {The quantum internet},\
	}\href {https://doi.org/10.1038/nature07127} {\bibfield  {journal} {\bibinfo
			{journal} {Nature}\ }\textbf {\bibinfo {volume} {453}},\ \bibinfo {pages}
		{1023} (\bibinfo {year} {2008})}\BibitemShut {NoStop}%
	\bibitem [{\citenamefont {Wehner}\ \emph {et~al.}(2018)\citenamefont {Wehner},
		\citenamefont {Elkouss},\ and\ \citenamefont {Hanson}}]{Wehner2018}%
	\BibitemOpen
	\bibfield  {author} {\bibinfo {author} {\bibfnamefont {S.}~\bibnamefont
			{Wehner}}, \bibinfo {author} {\bibfnamefont {D.}~\bibnamefont {Elkouss}},\
		and\ \bibinfo {author} {\bibfnamefont {R.}~\bibnamefont {Hanson}},\
	}\bibfield  {title} {\bibinfo {title} {Quantum internet: A vision for the
			road ahead},\ }\href {https://doi.org/10.1126/science.aam9288} {\bibfield
		{journal} {\bibinfo  {journal} {Science}\ }\textbf {\bibinfo {volume}
			{362}},\ \bibinfo {pages} {eaam9288} (\bibinfo {year} {2018})}\BibitemShut
	{NoStop}%
	\bibitem [{\citenamefont {Agrawal}\ \emph {et~al.}(2019)\citenamefont
		{Agrawal}, \citenamefont {Halder},\ and\ \citenamefont
		{Banik}}]{Agrawal2019}%
	\BibitemOpen
	\bibfield  {author} {\bibinfo {author} {\bibfnamefont {S.}~\bibnamefont
			{Agrawal}}, \bibinfo {author} {\bibfnamefont {S.}~\bibnamefont {Halder}},\
		and\ \bibinfo {author} {\bibfnamefont {M.}~\bibnamefont {Banik}},\ }\bibfield
	{title} {\bibinfo {title} {Genuinely entangled subspace with
			all-encompassing distillable entanglement across every bipartition},\ }\href
	{https://doi.org/10.1103/PhysRevA.99.032335} {\bibfield  {journal} {\bibinfo
			{journal} {Phys. Rev. A}\ }\textbf {\bibinfo {volume} {99}},\ \bibinfo
		{pages} {032335} (\bibinfo {year} {2019})}\BibitemShut {NoStop}%
	\bibitem [{\citenamefont {Saha}\ \emph {et~al.}(2020)\citenamefont {Saha},
		\citenamefont {Bhattacharya}, \citenamefont {Guha}, \citenamefont {Halder},\
		and\ \citenamefont {Banik}}]{Saha2020}%
	\BibitemOpen
	\bibfield  {author} {\bibinfo {author} {\bibfnamefont {S.}~\bibnamefont
			{Saha}}, \bibinfo {author} {\bibfnamefont {S.~S.}\ \bibnamefont
			{Bhattacharya}}, \bibinfo {author} {\bibfnamefont {T.}~\bibnamefont {Guha}},
		\bibinfo {author} {\bibfnamefont {S.}~\bibnamefont {Halder}},\ and\ \bibinfo
		{author} {\bibfnamefont {M.}~\bibnamefont {Banik}},\ }\bibfield  {title}
	{\bibinfo {title} {Advantage of quantum theory over nonclassical models of
			communication},\ }\href {https://doi.org/10.1002/andp.202000334} {\bibfield
		{journal} {\bibinfo  {journal} {Annalen der Physik}\ }\textbf {\bibinfo
			{volume} {532}},\ \bibinfo {pages} {2000334} (\bibinfo {year}
		{2020})}\BibitemShut {NoStop}%
	\bibitem [{\citenamefont {Rout}\ \emph {et~al.}(2021)\citenamefont {Rout},
		\citenamefont {Maity}, \citenamefont {Mukherjee}, \citenamefont {Halder},\
		and\ \citenamefont {Banik}}]{Rout2021}%
	\BibitemOpen
	\bibfield  {author} {\bibinfo {author} {\bibfnamefont {S.}~\bibnamefont
			{Rout}}, \bibinfo {author} {\bibfnamefont {A.~G.}\ \bibnamefont {Maity}},
		\bibinfo {author} {\bibfnamefont {A.}~\bibnamefont {Mukherjee}}, \bibinfo
		{author} {\bibfnamefont {S.}~\bibnamefont {Halder}},\ and\ \bibinfo {author}
		{\bibfnamefont {M.}~\bibnamefont {Banik}},\ }\bibfield  {title} {\bibinfo
		{title} {Multiparty orthogonal product states with minimal genuine
			nonlocality},\ }\href {https://doi.org/10.1103/PhysRevA.104.052433}
	{\bibfield  {journal} {\bibinfo  {journal} {Phys. Rev. A}\ }\textbf {\bibinfo
			{volume} {104}},\ \bibinfo {pages} {052433} (\bibinfo {year}
		{2021})}\BibitemShut {NoStop}%
	\bibitem [{\citenamefont {Banik}\ \emph {et~al.}(2021)\citenamefont {Banik},
		\citenamefont {Guha}, \citenamefont {Alimuddin}, \citenamefont {Kar},
		\citenamefont {Halder},\ and\ \citenamefont {Bhattacharya}}]{Banik2021}%
	\BibitemOpen
	\bibfield  {author} {\bibinfo {author} {\bibfnamefont {M.}~\bibnamefont
			{Banik}}, \bibinfo {author} {\bibfnamefont {T.}~\bibnamefont {Guha}},
		\bibinfo {author} {\bibfnamefont {M.}~\bibnamefont {Alimuddin}}, \bibinfo
		{author} {\bibfnamefont {G.}~\bibnamefont {Kar}}, \bibinfo {author}
		{\bibfnamefont {S.}~\bibnamefont {Halder}},\ and\ \bibinfo {author}
		{\bibfnamefont {S.~S.}\ \bibnamefont {Bhattacharya}},\ }\bibfield  {title}
	{\bibinfo {title} {Multicopy adaptive local discrimination: Strongest
			possible two-qubit nonlocal bases},\ }\href
	{https://doi.org/10.1103/PhysRevLett.126.210505} {\bibfield  {journal}
		{\bibinfo  {journal} {Phys. Rev. Lett.}\ }\textbf {\bibinfo {volume} {126}},\
		\bibinfo {pages} {210505} (\bibinfo {year} {2021})}\BibitemShut {NoStop}%
	\bibitem [{\citenamefont {Bhattacharya}\ \emph {et~al.}(2021)\citenamefont
		{Bhattacharya}, \citenamefont {Maity}, \citenamefont {Guha}, \citenamefont
		{Chiribella},\ and\ \citenamefont {Banik}}]{Bhattacharya2021}%
	\BibitemOpen
	\bibfield  {author} {\bibinfo {author} {\bibfnamefont {S.~S.}\ \bibnamefont
			{Bhattacharya}}, \bibinfo {author} {\bibfnamefont {A.~G.}\ \bibnamefont
			{Maity}}, \bibinfo {author} {\bibfnamefont {T.}~\bibnamefont {Guha}},
		\bibinfo {author} {\bibfnamefont {G.}~\bibnamefont {Chiribella}},\ and\
		\bibinfo {author} {\bibfnamefont {M.}~\bibnamefont {Banik}},\ }\bibfield
	{title} {\bibinfo {title} {Random-receiver quantum communication},\ }\href
	{https://doi.org/10.1103/PRXQuantum.2.020350} {\bibfield  {journal} {\bibinfo
			{journal} {PRX Quantum}\ }\textbf {\bibinfo {volume} {2}},\ \bibinfo {pages}
		{020350} (\bibinfo {year} {2021})}\BibitemShut {NoStop}%
	\bibitem [{\citenamefont {Gurvits}(2003)}]{Gurvits2003}%
	\BibitemOpen
	\bibfield  {author} {\bibinfo {author} {\bibfnamefont {L.}~\bibnamefont
			{Gurvits}},\ }\bibfield  {title} {\bibinfo {title} {Classical deterministic
			complexity of edmonds{\textquotesingle} problem and quantum entanglement},\
	}in\ \href {https://doi.org/10.1145/780542.780545} {\emph {\bibinfo
			{booktitle} {Proceedings of the thirty-fifth annual {ACM} symposium on Theory
				of computing}}}\ (\bibinfo  {publisher} {{ACM}},\ \bibinfo {year}
	{2003})\BibitemShut {NoStop}%
	\bibitem [{\citenamefont {Peres}(1996)}]{Peres1996}%
	\BibitemOpen
	\bibfield  {author} {\bibinfo {author} {\bibfnamefont {A.}~\bibnamefont
			{Peres}},\ }\bibfield  {title} {\bibinfo {title} {Separability criterion for
			density matrices},\ }\href {https://doi.org/10.1103/PhysRevLett.77.1413}
	{\bibfield  {journal} {\bibinfo  {journal} {Phys. Rev. Lett.}\ }\textbf
		{\bibinfo {volume} {77}},\ \bibinfo {pages} {1413} (\bibinfo {year}
		{1996})}\BibitemShut {NoStop}%
	\bibitem [{\citenamefont {Horodecki}\ \emph {et~al.}(1996)\citenamefont
		{Horodecki}, \citenamefont {Horodecki},\ and\ \citenamefont
		{Horodecki}}]{Horodecki1996}%
	\BibitemOpen
	\bibfield  {author} {\bibinfo {author} {\bibfnamefont {M.}~\bibnamefont
			{Horodecki}}, \bibinfo {author} {\bibfnamefont {P.}~\bibnamefont
			{Horodecki}},\ and\ \bibinfo {author} {\bibfnamefont {R.}~\bibnamefont
			{Horodecki}},\ }\bibfield  {title} {\bibinfo {title} {Separability of mixed
			states: necessary and sufficient conditions},\ }\href
	{https://doi.org/10.1016/s0375-9601(96)00706-2} {\bibfield  {journal}
		{\bibinfo  {journal} {Physics Letters A}\ }\textbf {\bibinfo {volume}
			{223}},\ \bibinfo {pages} {1} (\bibinfo {year} {1996})}\BibitemShut {NoStop}%
	\bibitem [{\citenamefont {Dirkse}\ \emph {et~al.}(2020)\citenamefont {Dirkse},
		\citenamefont {Pompili}, \citenamefont {Hanson}, \citenamefont {Walter},\
		and\ \citenamefont {Wehner}}]{Dirkse2020}%
	\BibitemOpen
	\bibfield  {author} {\bibinfo {author} {\bibfnamefont {B.}~\bibnamefont
			{Dirkse}}, \bibinfo {author} {\bibfnamefont {M.}~\bibnamefont {Pompili}},
		\bibinfo {author} {\bibfnamefont {R.}~\bibnamefont {Hanson}}, \bibinfo
		{author} {\bibfnamefont {M.}~\bibnamefont {Walter}},\ and\ \bibinfo {author}
		{\bibfnamefont {S.}~\bibnamefont {Wehner}},\ }\bibfield  {title} {\bibinfo
		{title} {Witnessing entanglement in experiments with correlated noise},\
	}\href {https://doi.org/10.1088/2058-9565/ab8d88} {\bibfield  {journal}
		{\bibinfo  {journal} {Quantum Sci. Technol.}\ }\textbf {\bibinfo {volume}
			{5}},\ \bibinfo {pages} {035007} (\bibinfo {year} {2020})}\BibitemShut
	{NoStop}%
	\bibitem [{\citenamefont {Bennett}\ \emph {et~al.}(1996)\citenamefont
		{Bennett}, \citenamefont {DiVincenzo}, \citenamefont {Smolin},\ and\
		\citenamefont {Wootters}}]{Bennett1996}%
	\BibitemOpen
	\bibfield  {author} {\bibinfo {author} {\bibfnamefont {C.~H.}\ \bibnamefont
			{Bennett}}, \bibinfo {author} {\bibfnamefont {D.~P.}\ \bibnamefont
			{DiVincenzo}}, \bibinfo {author} {\bibfnamefont {J.~A.}\ \bibnamefont
			{Smolin}},\ and\ \bibinfo {author} {\bibfnamefont {W.~K.}\ \bibnamefont
			{Wootters}},\ }\bibfield  {title} {\bibinfo {title} {Mixed-state entanglement
			and quantum error correction},\ }\href
	{https://doi.org/10.1103/PhysRevA.54.3824} {\bibfield  {journal} {\bibinfo
			{journal} {Phys. Rev. A}\ }\textbf {\bibinfo {volume} {54}},\ \bibinfo
		{pages} {3824} (\bibinfo {year} {1996})}\BibitemShut {NoStop}%
	\bibitem [{\citenamefont {Devetak}\ and\ \citenamefont
		{Winter}(2005)}]{Devetak2005}%
	\BibitemOpen
	\bibfield  {author} {\bibinfo {author} {\bibfnamefont {I.}~\bibnamefont
			{Devetak}}\ and\ \bibinfo {author} {\bibfnamefont {A.}~\bibnamefont
			{Winter}},\ }\bibfield  {title} {\bibinfo {title} {Distillation of secret key
			and entanglement from quantum states},\ }\href
	{https://doi.org/10.1098/rspa.2004.1372} {\bibfield  {journal} {\bibinfo
			{journal} {Proc. R. Soc. A: Math. Phys. Eng. Sc.}\ }\textbf {\bibinfo
			{volume} {461}},\ \bibinfo {pages} {207} (\bibinfo {year}
		{2005})}\BibitemShut {NoStop}%
	\bibitem [{\citenamefont {Perarnau-Llobet}\ \emph {et~al.}(2015)\citenamefont
		{Perarnau-Llobet}, \citenamefont {Hovhannisyan}, \citenamefont {Huber},
		\citenamefont {Skrzypczyk}, \citenamefont {Brunner},\ and\ \citenamefont
		{Ac\'{\i}n}}]{Perarnau2015}%
	\BibitemOpen
	\bibfield  {author} {\bibinfo {author} {\bibfnamefont {M.}~\bibnamefont
			{Perarnau-Llobet}}, \bibinfo {author} {\bibfnamefont {K.~V.}\ \bibnamefont
			{Hovhannisyan}}, \bibinfo {author} {\bibfnamefont {M.}~\bibnamefont {Huber}},
		\bibinfo {author} {\bibfnamefont {P.}~\bibnamefont {Skrzypczyk}}, \bibinfo
		{author} {\bibfnamefont {N.}~\bibnamefont {Brunner}},\ and\ \bibinfo {author}
		{\bibfnamefont {A.}~\bibnamefont {Ac\'{\i}n}},\ }\bibfield  {title} {\bibinfo
		{title} {Extractable work from correlations},\ }\href
	{https://doi.org/10.1103/PhysRevX.5.041011} {\bibfield  {journal} {\bibinfo
			{journal} {Phys. Rev. X}\ }\textbf {\bibinfo {volume} {5}},\ \bibinfo {pages}
		{041011} (\bibinfo {year} {2015})}\BibitemShut {NoStop}%
	\bibitem [{\citenamefont {Mukherjee}\ \emph {et~al.}(2016)\citenamefont
		{Mukherjee}, \citenamefont {Roy}, \citenamefont {Bhattacharya},\ and\
		\citenamefont {Banik}}]{Mukherjee16}%
	\BibitemOpen
	\bibfield  {author} {\bibinfo {author} {\bibfnamefont {A.}~\bibnamefont
			{Mukherjee}}, \bibinfo {author} {\bibfnamefont {A.}~\bibnamefont {Roy}},
		\bibinfo {author} {\bibfnamefont {S.~S.}\ \bibnamefont {Bhattacharya}},\ and\
		\bibinfo {author} {\bibfnamefont {M.}~\bibnamefont {Banik}},\ }\bibfield
	{title} {\bibinfo {title} {Presence of quantum correlations results in a
			nonvanishing ergotropic gap},\ }\href
	{https://doi.org/10.1103/PhysRevE.93.052140} {\bibfield  {journal} {\bibinfo
			{journal} {Phys. Rev. E}\ }\textbf {\bibinfo {volume} {93}},\ \bibinfo
		{pages} {052140} (\bibinfo {year} {2016})}\BibitemShut {NoStop}%
	\bibitem [{\citenamefont {Alimuddin}\ \emph {et~al.}(2019)\citenamefont
		{Alimuddin}, \citenamefont {Guha},\ and\ \citenamefont
		{Parashar}}]{Alimuddin2019}%
	\BibitemOpen
	\bibfield  {author} {\bibinfo {author} {\bibfnamefont {M.}~\bibnamefont
			{Alimuddin}}, \bibinfo {author} {\bibfnamefont {T.}~\bibnamefont {Guha}},\
		and\ \bibinfo {author} {\bibfnamefont {P.}~\bibnamefont {Parashar}},\
	}\bibfield  {title} {\bibinfo {title} {Bound on ergotropic gap for bipartite
			separable states},\ }\href {https://doi.org/10.1103/PhysRevA.99.052320}
	{\bibfield  {journal} {\bibinfo  {journal} {Phys. Rev. A}\ }\textbf {\bibinfo
			{volume} {99}},\ \bibinfo {pages} {052320} (\bibinfo {year}
		{2019})}\BibitemShut {NoStop}%
	\bibitem [{\citenamefont {Alimuddin}\ \emph
		{et~al.}(2020{\natexlab{a}})\citenamefont {Alimuddin}, \citenamefont {Guha},\
		and\ \citenamefont {Parashar}}]{Alimuddin2020(1)}%
	\BibitemOpen
	\bibfield  {author} {\bibinfo {author} {\bibfnamefont {M.}~\bibnamefont
			{Alimuddin}}, \bibinfo {author} {\bibfnamefont {T.}~\bibnamefont {Guha}},\
		and\ \bibinfo {author} {\bibfnamefont {P.}~\bibnamefont {Parashar}},\
	}\bibfield  {title} {\bibinfo {title} {Independence of work and entropy for
			equal-energetic finite quantum systems: Passive-state energy as an
			entanglement quantifier},\ }\href
	{https://doi.org/10.1103/PhysRevE.102.012145} {\bibfield  {journal} {\bibinfo
			{journal} {Phys. Rev. E}\ }\textbf {\bibinfo {volume} {102}},\ \bibinfo
		{pages} {012145} (\bibinfo {year} {2020}{\natexlab{a}})}\BibitemShut
	{NoStop}%
	\bibitem [{\citenamefont {Alimuddin}\ \emph
		{et~al.}(2020{\natexlab{b}})\citenamefont {Alimuddin}, \citenamefont {Guha},\
		and\ \citenamefont {Parashar}}]{Alimuddin2020(2)}%
	\BibitemOpen
	\bibfield  {author} {\bibinfo {author} {\bibfnamefont {M.}~\bibnamefont
			{Alimuddin}}, \bibinfo {author} {\bibfnamefont {T.}~\bibnamefont {Guha}},\
		and\ \bibinfo {author} {\bibfnamefont {P.}~\bibnamefont {Parashar}},\
	}\bibfield  {title} {\bibinfo {title} {Structure of passive states and its
			implication in charging quantum batteries},\ }\href
	{https://doi.org/10.1103/PhysRevE.102.022106} {\bibfield  {journal} {\bibinfo
			{journal} {Phys. Rev. E}\ }\textbf {\bibinfo {volume} {102}},\ \bibinfo
		{pages} {022106} (\bibinfo {year} {2020}{\natexlab{b}})}\BibitemShut
	{NoStop}%
	\bibitem [{\citenamefont {Puliyil}\ \emph {et~al.}(2022)\citenamefont
		{Puliyil}, \citenamefont {Banik},\ and\ \citenamefont
		{Alimuddin}}]{Puliyil2022}%
	\BibitemOpen
	\bibfield  {author} {\bibinfo {author} {\bibfnamefont {S.}~\bibnamefont
			{Puliyil}}, \bibinfo {author} {\bibfnamefont {M.}~\bibnamefont {Banik}},\
		and\ \bibinfo {author} {\bibfnamefont {M.}~\bibnamefont {Alimuddin}},\
	}\bibfield  {title} {\bibinfo {title} {Thermodynamic signatures of genuinely
			multipartite entanglement},\ }\href
	{https://doi.org/10.1103/PhysRevLett.129.070601} {\bibfield  {journal}
		{\bibinfo  {journal} {Phys. Rev. Lett.}\ }\textbf {\bibinfo {volume} {129}},\
		\bibinfo {pages} {070601} (\bibinfo {year} {2022})}\BibitemShut {NoStop}%
	\bibitem [{\citenamefont {Yang}\ \emph {et~al.}(2023)\citenamefont {Yang},
		\citenamefont {Yang}, \citenamefont {Alimuddin}, \citenamefont {Salvia},
		\citenamefont {Fei}, \citenamefont {Zhao}, \citenamefont {Nimmrichter},\ and\
		\citenamefont {Luo}}]{Yang'23}%
	\BibitemOpen
	\bibfield  {author} {\bibinfo {author} {\bibfnamefont {X.}~\bibnamefont
			{Yang}}, \bibinfo {author} {\bibfnamefont {Y.~H.}\ \bibnamefont {Yang}},
		\bibinfo {author} {\bibfnamefont {M.}~\bibnamefont {Alimuddin}}, \bibinfo
		{author} {\bibfnamefont {R.}~\bibnamefont {Salvia}}, \bibinfo {author}
		{\bibfnamefont {S.~M.}\ \bibnamefont {Fei}}, \bibinfo {author} {\bibfnamefont
			{L.~M.}\ \bibnamefont {Zhao}}, \bibinfo {author} {\bibfnamefont
			{S.}~\bibnamefont {Nimmrichter}},\ and\ \bibinfo {author} {\bibfnamefont
			{M.~X.}\ \bibnamefont {Luo}},\ }\bibfield  {title} {\bibinfo {title} {The
			battery capacity of energy-storing quantum systems},\ }\href
	{https://arxiv.org/abs/2302.09905} {\bibfield  {journal} {\bibinfo  {journal}
			{arXiv:2302.09905}\ } (\bibinfo {year} {2023})}\BibitemShut {NoStop}%
	\bibitem [{\citenamefont {Allahverdyan}\ \emph {et~al.}(2004)\citenamefont
		{Allahverdyan}, \citenamefont {Balian},\ and\ \citenamefont
		{Nieuwenhuizen}}]{Allahverdyan2004}%
	\BibitemOpen
	\bibfield  {author} {\bibinfo {author} {\bibfnamefont {A.~E.}\ \bibnamefont
			{Allahverdyan}}, \bibinfo {author} {\bibfnamefont {R.}~\bibnamefont
			{Balian}},\ and\ \bibinfo {author} {\bibfnamefont {T.~M.}\ \bibnamefont
			{Nieuwenhuizen}},\ }\bibfield  {title} {\bibinfo {title} {Maximal work
			extraction from finite quantum systems},\ }\href
	{https://doi.org/10.1209/epl/i2004-10101-2} {\bibfield  {journal} {\bibinfo
			{journal} {Europhysics Letters ({EPL})}\ }\textbf {\bibinfo {volume} {67}},\
		\bibinfo {pages} {565} (\bibinfo {year} {2004})}\BibitemShut {NoStop}%
	\bibitem [{\citenamefont {Jones}\ \emph {et~al.}(2009)\citenamefont {Jones},
		\citenamefont {Karlen}, \citenamefont {Fitzsimons}, \citenamefont {Ardavan},
		\citenamefont {Benjamin}, \citenamefont {Briggs},\ and\ \citenamefont
		{Morton}}]{Jones2009}%
	\BibitemOpen
	\bibfield  {author} {\bibinfo {author} {\bibfnamefont {J.~A.}\ \bibnamefont
			{Jones}}, \bibinfo {author} {\bibfnamefont {S.~D.}\ \bibnamefont {Karlen}},
		\bibinfo {author} {\bibfnamefont {J.}~\bibnamefont {Fitzsimons}}, \bibinfo
		{author} {\bibfnamefont {A.}~\bibnamefont {Ardavan}}, \bibinfo {author}
		{\bibfnamefont {S.~C.}\ \bibnamefont {Benjamin}}, \bibinfo {author}
		{\bibfnamefont {G.~A.~D.}\ \bibnamefont {Briggs}},\ and\ \bibinfo {author}
		{\bibfnamefont {J.~J.~L.}\ \bibnamefont {Morton}},\ }\bibfield  {title}
	{\bibinfo {title} {Magnetic field sensing beyond the standard quantum limit
			using 10-spin {NOON} states},\ }\href
	{https://doi.org/10.1126/science.1170730} {\bibfield  {journal} {\bibinfo
			{journal} {Science}\ }\textbf {\bibinfo {volume} {324}},\ \bibinfo {pages}
		{1166} (\bibinfo {year} {2009})}\BibitemShut {NoStop}%
	\bibitem [{\citenamefont {Mahesh}\ \emph {et~al.}(2021)\citenamefont {Mahesh},
		\citenamefont {Khurana}, \citenamefont {Krithika}, \citenamefont {Sreejith},\
		and\ \citenamefont {Kumar}}]{mahesh2021star}%
	\BibitemOpen
	\bibfield  {author} {\bibinfo {author} {\bibfnamefont {T.~S.}\ \bibnamefont
			{Mahesh}}, \bibinfo {author} {\bibfnamefont {D.}~\bibnamefont {Khurana}},
		\bibinfo {author} {\bibfnamefont {V.~R.}\ \bibnamefont {Krithika}}, \bibinfo
		{author} {\bibfnamefont {G.~J.}\ \bibnamefont {Sreejith}},\ and\ \bibinfo
		{author} {\bibfnamefont {C.~S.~S.}\ \bibnamefont {Kumar}},\ }\bibfield
	{title} {\bibinfo {title} {Star-topology registers: {NMR} and quantum
			information perspectives},\ }\href {https://doi.org/10.1088/1361-648x/ac0dd3}
	{\bibfield  {journal} {\bibinfo  {journal} {J. Phys. Condens. Matter}\
		}\textbf {\bibinfo {volume} {33}},\ \bibinfo {pages} {383002} (\bibinfo
		{year} {2021})}\BibitemShut {NoStop}%
	\bibitem [{\citenamefont {Shukla}\ \emph {et~al.}(2014)\citenamefont {Shukla},
		\citenamefont {Sharma},\ and\ \citenamefont {Mahesh}}]{shukla2014noon}%
	\BibitemOpen
	\bibfield  {author} {\bibinfo {author} {\bibfnamefont {A.}~\bibnamefont
			{Shukla}}, \bibinfo {author} {\bibfnamefont {M.}~\bibnamefont {Sharma}},\
		and\ \bibinfo {author} {\bibfnamefont {T.}~\bibnamefont {Mahesh}},\
	}\bibfield  {title} {\bibinfo {title} {{NOON} states in star-topology
			spin-systems: Applications in diffusion studies and {RF} inhomogeneity
			mapping},\ }\href {https://doi.org/10.1016/j.cplett.2013.11.065} {\bibfield
		{journal} {\bibinfo  {journal} {Chem. Phys. Lett.}\ }\textbf {\bibinfo
			{volume} {592}},\ \bibinfo {pages} {227} (\bibinfo {year}
		{2014})}\BibitemShut {NoStop}%
	\bibitem [{\citenamefont {Pusz}\ and\ \citenamefont
		{Woronowicz}(1978)}]{Pusz1978}%
	\BibitemOpen
	\bibfield  {author} {\bibinfo {author} {\bibfnamefont {W.}~\bibnamefont
			{Pusz}}\ and\ \bibinfo {author} {\bibfnamefont {S.~L.}\ \bibnamefont
			{Woronowicz}},\ }\bibfield  {title} {\bibinfo {title} {Passive states and
			{KMS} states for general quantum systems},\ }\href
	{https://doi.org/10.1007/bf01614224} {\bibfield  {journal} {\bibinfo
			{journal} {Commun. Math. Phys.}\ }\textbf {\bibinfo {volume} {58}},\ \bibinfo
		{pages} {273} (\bibinfo {year} {1978})}\BibitemShut {NoStop}%
	\bibitem [{\citenamefont {Lenard}(1978)}]{Lenard1978}%
	\BibitemOpen
	\bibfield  {author} {\bibinfo {author} {\bibfnamefont {A.}~\bibnamefont
			{Lenard}},\ }\bibfield  {title} {\bibinfo {title} {Thermodynamical proof of
			the gibbs formula for elementary quantum systems},\ }\href
	{https://doi.org/10.1007/bf01011769} {\bibfield  {journal} {\bibinfo
			{journal} {J. Stat. Phys.}\ }\textbf {\bibinfo {volume} {19}},\ \bibinfo
		{pages} {575} (\bibinfo {year} {1978})}\BibitemShut {NoStop}%
	\bibitem [{\citenamefont {{\AA}berg}(2013)}]{Aberg2013}%
	\BibitemOpen
	\bibfield  {author} {\bibinfo {author} {\bibfnamefont {J.}~\bibnamefont
			{{\AA}berg}},\ }\bibfield  {title} {\bibinfo {title} {Truly work-like work
			extraction via a single-shot analysis},\ }\href
	{https://doi.org/10.1038/ncomms2712} {\bibfield  {journal} {\bibinfo
			{journal} {Nat. Commun.}\ }\textbf {\bibinfo {volume} {4}},\ \bibinfo {pages}
		{1925} (\bibinfo {year} {2013})}\BibitemShut {NoStop}%
	\bibitem [{\citenamefont {Skrzypczyk}\ \emph {et~al.}(2014)\citenamefont
		{Skrzypczyk}, \citenamefont {Short},\ and\ \citenamefont
		{Popescu}}]{Skrzypczyk2014}%
	\BibitemOpen
	\bibfield  {author} {\bibinfo {author} {\bibfnamefont {P.}~\bibnamefont
			{Skrzypczyk}}, \bibinfo {author} {\bibfnamefont {A.~J.}\ \bibnamefont
			{Short}},\ and\ \bibinfo {author} {\bibfnamefont {S.}~\bibnamefont
			{Popescu}},\ }\bibfield  {title} {\bibinfo {title} {Work extraction and
			thermodynamics for individual quantum systems},\ }\bibfield  {journal}
	{\bibinfo  {journal} {Nat. Commun.}\ }\textbf {\bibinfo {volume} {5}},\ \href
	{https://doi.org/10.1038/ncomms5185} {10.1038/ncomms5185} (\bibinfo {year}
	{2014})\BibitemShut {NoStop}%
	\bibitem [{\citenamefont {Skrzypczyk}\ \emph {et~al.}(2015)\citenamefont
		{Skrzypczyk}, \citenamefont {Silva},\ and\ \citenamefont
		{Brunner}}]{Skrzypczyk2015}%
	\BibitemOpen
	\bibfield  {author} {\bibinfo {author} {\bibfnamefont {P.}~\bibnamefont
			{Skrzypczyk}}, \bibinfo {author} {\bibfnamefont {R.}~\bibnamefont {Silva}},\
		and\ \bibinfo {author} {\bibfnamefont {N.}~\bibnamefont {Brunner}},\
	}\bibfield  {title} {\bibinfo {title} {Passivity, complete passivity, and
			virtual temperatures},\ }\href {https://doi.org/10.1103/PhysRevE.91.052133}
	{\bibfield  {journal} {\bibinfo  {journal} {Phys. Rev. E}\ }\textbf {\bibinfo
			{volume} {91}},\ \bibinfo {pages} {052133} (\bibinfo {year}
		{2015})}\BibitemShut {NoStop}%
	\bibitem [{\citenamefont {Joshi}\ and\ \citenamefont
		{Mahesh}(2022)}]{PhysRevA.106.042601}%
	\BibitemOpen
	\bibfield  {author} {\bibinfo {author} {\bibfnamefont {J.}~\bibnamefont
			{Joshi}}\ and\ \bibinfo {author} {\bibfnamefont {T.~S.}\ \bibnamefont
			{Mahesh}},\ }\bibfield  {title} {\bibinfo {title} {Experimental investigation
			of a quantum battery using star-topology nmr spin systems},\ }\href
	{https://doi.org/10.1103/PhysRevA.106.042601} {\bibfield  {journal} {\bibinfo
			{journal} {Phys. Rev. A}\ }\textbf {\bibinfo {volume} {106}},\ \bibinfo
		{pages} {042601} (\bibinfo {year} {2022})}\BibitemShut {NoStop}%
	\bibitem [{Sup()}]{Supple}%
	\BibitemOpen
	\bibfield  {title} {\bibinfo {title} {See supplemental at https://-----,
			which also includes the references \cite{Marshall1979,Nielsen1999,
				Horodecki2003,
				Winter2016,Horodecki2013,Bennett1999,Dur00,Guhne2010,cavanagh1996protein,riedel2021bell,roy2010density,carravetta2004long,krithika2021observation}},\
	}\href@noop {} {\ }\BibitemShut {NoStop}%
	\bibitem [{\citenamefont {Nielsen}\ and\ \citenamefont
		{Kempe}(2001)}]{Nielsen2001}%
	\BibitemOpen
	\bibfield  {author} {\bibinfo {author} {\bibfnamefont {M.~A.}\ \bibnamefont
			{Nielsen}}\ and\ \bibinfo {author} {\bibfnamefont {J.}~\bibnamefont
			{Kempe}},\ }\bibfield  {title} {\bibinfo {title} {Separable states are more
			disordered globally than locally},\ }\href
	{https://doi.org/10.1103/PhysRevLett.86.5184} {\bibfield  {journal} {\bibinfo
			{journal} {Phys. Rev. Lett.}\ }\textbf {\bibinfo {volume} {86}},\ \bibinfo
		{pages} {5184} (\bibinfo {year} {2001})}\BibitemShut {NoStop}%
	\bibitem [{\citenamefont {Cory}\ \emph {et~al.}(1997)\citenamefont {Cory},
		\citenamefont {Fahmy},\ and\ \citenamefont {Havel}}]{cory1997ensemble}%
	\BibitemOpen
	\bibfield  {author} {\bibinfo {author} {\bibfnamefont {D.~G.}\ \bibnamefont
			{Cory}}, \bibinfo {author} {\bibfnamefont {A.~F.}\ \bibnamefont {Fahmy}},\
		and\ \bibinfo {author} {\bibfnamefont {T.~F.}\ \bibnamefont {Havel}},\
	}\bibfield  {title} {\bibinfo {title} {Ensemble quantum computing by
			{NMR}{\hspace{0.167em}}spectroscopy},\ }\href
	{https://doi.org/10.1073/pnas.94.5.1634} {\bibfield  {journal} {\bibinfo
			{journal} {Proc. Natl. Acad. Sci.}\ }\textbf {\bibinfo {volume} {94}},\
		\bibinfo {pages} {1634} (\bibinfo {year} {1997})}\BibitemShut {NoStop}%
	\bibitem [{\citenamefont {Raimond}\ \emph {et~al.}(2001)\citenamefont
		{Raimond}, \citenamefont {Brune},\ and\ \citenamefont
		{Haroche}}]{Raimond2001}%
	\BibitemOpen
	\bibfield  {author} {\bibinfo {author} {\bibfnamefont {J.~M.}\ \bibnamefont
			{Raimond}}, \bibinfo {author} {\bibfnamefont {M.}~\bibnamefont {Brune}},\
		and\ \bibinfo {author} {\bibfnamefont {S.}~\bibnamefont {Haroche}},\
	}\bibfield  {title} {\bibinfo {title} {Manipulating quantum entanglement with
			atoms and photons in a cavity},\ }\href
	{https://doi.org/10.1103/RevModPhys.73.565} {\bibfield  {journal} {\bibinfo
			{journal} {Rev. Mod. Phys.}\ }\textbf {\bibinfo {volume} {73}},\ \bibinfo
		{pages} {565} (\bibinfo {year} {2001})}\BibitemShut {NoStop}%
	\bibitem [{\citenamefont {Kok}\ \emph {et~al.}(2007)\citenamefont {Kok},
		\citenamefont {Munro}, \citenamefont {Nemoto}, \citenamefont {Ralph},
		\citenamefont {Dowling},\ and\ \citenamefont {Milburn}}]{Kok2007}%
	\BibitemOpen
	\bibfield  {author} {\bibinfo {author} {\bibfnamefont {P.}~\bibnamefont
			{Kok}}, \bibinfo {author} {\bibfnamefont {W.~J.}\ \bibnamefont {Munro}},
		\bibinfo {author} {\bibfnamefont {K.}~\bibnamefont {Nemoto}}, \bibinfo
		{author} {\bibfnamefont {T.~C.}\ \bibnamefont {Ralph}}, \bibinfo {author}
		{\bibfnamefont {J.~P.}\ \bibnamefont {Dowling}},\ and\ \bibinfo {author}
		{\bibfnamefont {G.~J.}\ \bibnamefont {Milburn}},\ }\bibfield  {title}
	{\bibinfo {title} {Linear optical quantum computing with photonic qubits},\
	}\href {https://doi.org/10.1103/RevModPhys.79.135} {\bibfield  {journal}
		{\bibinfo  {journal} {Rev. Mod. Phys.}\ }\textbf {\bibinfo {volume} {79}},\
		\bibinfo {pages} {135} (\bibinfo {year} {2007})}\BibitemShut {NoStop}%
	\bibitem [{\citenamefont {Wendin}(2017)}]{Wendin2017}%
	\BibitemOpen
	\bibfield  {author} {\bibinfo {author} {\bibfnamefont {G.}~\bibnamefont
			{Wendin}},\ }\bibfield  {title} {\bibinfo {title} {Quantum information
			processing with superconducting circuits: a review},\ }\href
	{https://doi.org/10.1088/1361-6633/aa7e1a} {\bibfield  {journal} {\bibinfo
			{journal} {Rep. Prog. Phys.}\ }\textbf {\bibinfo {volume} {80}},\ \bibinfo
		{pages} {106001} (\bibinfo {year} {2017})}\BibitemShut {NoStop}%
	\bibitem [{\citenamefont {Lu}\ \emph {et~al.}(2015)\citenamefont {Lu},
		\citenamefont {Brodutch}, \citenamefont {Park}, \citenamefont {Katiyar},
		\citenamefont {Jochym-O'Connor},\ and\ \citenamefont {Laflamme}}]{Lu2015}%
	\BibitemOpen
	\bibfield  {author} {\bibinfo {author} {\bibfnamefont {D.}~\bibnamefont
			{Lu}}, \bibinfo {author} {\bibfnamefont {A.}~\bibnamefont {Brodutch}},
		\bibinfo {author} {\bibfnamefont {J.}~\bibnamefont {Park}}, \bibinfo {author}
		{\bibfnamefont {H.}~\bibnamefont {Katiyar}}, \bibinfo {author} {\bibfnamefont
			{T.}~\bibnamefont {Jochym-O'Connor}},\ and\ \bibinfo {author} {\bibfnamefont
			{R.}~\bibnamefont {Laflamme}},\ }\href@noop {} {\bibinfo {title} {Nmr quantum
			information processing}} (\bibinfo {year} {2015}),\ \Eprint
	{https://arxiv.org/abs/1501.01353} {arXiv:1501.01353 [quant-ph]} \BibitemShut
	{NoStop}%
	\bibitem [{\citenamefont {Xiang}\ \emph {et~al.}(2013)\citenamefont {Xiang},
		\citenamefont {Ashhab}, \citenamefont {You},\ and\ \citenamefont
		{Nori}}]{Xiang2013}%
	\BibitemOpen
	\bibfield  {author} {\bibinfo {author} {\bibfnamefont {Z.-L.}\ \bibnamefont
			{Xiang}}, \bibinfo {author} {\bibfnamefont {S.}~\bibnamefont {Ashhab}},
		\bibinfo {author} {\bibfnamefont {J.~Q.}\ \bibnamefont {You}},\ and\ \bibinfo
		{author} {\bibfnamefont {F.}~\bibnamefont {Nori}},\ }\bibfield  {title}
	{\bibinfo {title} {Hybrid quantum circuits: Superconducting circuits
			interacting with other quantum systems},\ }\href
	{https://doi.org/10.1103/RevModPhys.85.623} {\bibfield  {journal} {\bibinfo
			{journal} {Rev. Mod. Phys.}\ }\textbf {\bibinfo {volume} {85}},\ \bibinfo
		{pages} {623} (\bibinfo {year} {2013})}\BibitemShut {NoStop}%
	\bibitem [{\citenamefont {Aspect}\ \emph {et~al.}(1981)\citenamefont {Aspect},
		\citenamefont {Grangier},\ and\ \citenamefont {Roger}}]{Aspect1981}%
	\BibitemOpen
	\bibfield  {author} {\bibinfo {author} {\bibfnamefont {A.}~\bibnamefont
			{Aspect}}, \bibinfo {author} {\bibfnamefont {P.}~\bibnamefont {Grangier}},\
		and\ \bibinfo {author} {\bibfnamefont {G.}~\bibnamefont {Roger}},\ }\bibfield
	{title} {\bibinfo {title} {Experimental tests of realistic local theories
			via bell's theorem},\ }\href {https://doi.org/10.1103/PhysRevLett.47.460}
	{\bibfield  {journal} {\bibinfo  {journal} {Phys. Rev. Lett.}\ }\textbf
		{\bibinfo {volume} {47}},\ \bibinfo {pages} {460} (\bibinfo {year}
		{1981})}\BibitemShut {NoStop}%
	\bibitem [{\citenamefont {Aspect}\ \emph
		{et~al.}(1982{\natexlab{a}})\citenamefont {Aspect}, \citenamefont
		{Dalibard},\ and\ \citenamefont {Roger}}]{Aspect1982(1)}%
	\BibitemOpen
	\bibfield  {author} {\bibinfo {author} {\bibfnamefont {A.}~\bibnamefont
			{Aspect}}, \bibinfo {author} {\bibfnamefont {J.}~\bibnamefont {Dalibard}},\
		and\ \bibinfo {author} {\bibfnamefont {G.}~\bibnamefont {Roger}},\ }\bibfield
	{title} {\bibinfo {title} {Experimental test of bell's inequalities using
			time-varying analyzers},\ }\href
	{https://doi.org/10.1103/PhysRevLett.49.1804} {\bibfield  {journal} {\bibinfo
			{journal} {Phys. Rev. Lett.}\ }\textbf {\bibinfo {volume} {49}},\ \bibinfo
		{pages} {1804} (\bibinfo {year} {1982}{\natexlab{a}})}\BibitemShut {NoStop}%
	\bibitem [{\citenamefont {Aspect}\ \emph
		{et~al.}(1982{\natexlab{b}})\citenamefont {Aspect}, \citenamefont
		{Grangier},\ and\ \citenamefont {Roger}}]{Aspect1982(2)}%
	\BibitemOpen
	\bibfield  {author} {\bibinfo {author} {\bibfnamefont {A.}~\bibnamefont
			{Aspect}}, \bibinfo {author} {\bibfnamefont {P.}~\bibnamefont {Grangier}},\
		and\ \bibinfo {author} {\bibfnamefont {G.}~\bibnamefont {Roger}},\ }\bibfield
	{title} {\bibinfo {title} {Experimental realization of
			einstein-podolsky-rosen-bohm gedankenexperiment: A new violation of bell's
			inequalities},\ }\href {https://doi.org/10.1103/PhysRevLett.49.91} {\bibfield
		{journal} {\bibinfo  {journal} {Phys. Rev. Lett.}\ }\textbf {\bibinfo
			{volume} {49}},\ \bibinfo {pages} {91} (\bibinfo {year}
		{1982}{\natexlab{b}})}\BibitemShut {NoStop}%
	\bibitem [{\citenamefont {Bouwmeester}\ \emph {et~al.}(1999)\citenamefont
		{Bouwmeester}, \citenamefont {Pan}, \citenamefont {Daniell}, \citenamefont
		{Weinfurter},\ and\ \citenamefont {Zeilinger}}]{Bouwmeester1999}%
	\BibitemOpen
	\bibfield  {author} {\bibinfo {author} {\bibfnamefont {D.}~\bibnamefont
			{Bouwmeester}}, \bibinfo {author} {\bibfnamefont {J.-W.}\ \bibnamefont
			{Pan}}, \bibinfo {author} {\bibfnamefont {M.}~\bibnamefont {Daniell}},
		\bibinfo {author} {\bibfnamefont {H.}~\bibnamefont {Weinfurter}},\ and\
		\bibinfo {author} {\bibfnamefont {A.}~\bibnamefont {Zeilinger}},\ }\bibfield
	{title} {\bibinfo {title} {Observation of three-photon
			greenberger-horne-zeilinger entanglement},\ }\href
	{https://doi.org/10.1103/PhysRevLett.82.1345} {\bibfield  {journal} {\bibinfo
			{journal} {Phys. Rev. Lett.}\ }\textbf {\bibinfo {volume} {82}},\ \bibinfo
		{pages} {1345} (\bibinfo {year} {1999})}\BibitemShut {NoStop}%
	\bibitem [{\citenamefont {Pan}\ \emph {et~al.}(2000)\citenamefont {Pan},
		\citenamefont {Bouwmeester}, \citenamefont {M.Daniell}, \citenamefont
		{Weinfurter},\ and\ \citenamefont {Zeilinger}}]{Pan2000}%
	\BibitemOpen
	\bibfield  {author} {\bibinfo {author} {\bibfnamefont {J.-W.}\ \bibnamefont
			{Pan}}, \bibinfo {author} {\bibfnamefont {D.}~\bibnamefont {Bouwmeester}},
		\bibinfo {author} {\bibnamefont {M.Daniell}}, \bibinfo {author}
		{\bibfnamefont {H.}~\bibnamefont {Weinfurter}},\ and\ \bibinfo {author}
		{\bibfnamefont {A.}~\bibnamefont {Zeilinger}},\ }\bibfield  {title} {\bibinfo
		{title} {Experimental test of quantum nonlocality in three-photon
			greenberger{\textendash}horne{\textendash}zeilinger entanglement},\ }\href
	{https://doi.org/10.1038/35000514} {\bibfield  {journal} {\bibinfo  {journal}
			{Nature}\ }\textbf {\bibinfo {volume} {403}},\ \bibinfo {pages} {515}
		(\bibinfo {year} {2000})}\BibitemShut {NoStop}%
	\bibitem [{\citenamefont {Zhao}\ \emph {et~al.}(2003)\citenamefont {Zhao},
		\citenamefont {Yang}, \citenamefont {Chen}, \citenamefont {Zhang},
		\citenamefont {\ifmmode~\dot{Z}\else \.{Z}\fi{}ukowski},\ and\ \citenamefont
		{Pan}}]{Zhao2003}%
	\BibitemOpen
	\bibfield  {author} {\bibinfo {author} {\bibfnamefont {Z.}~\bibnamefont
			{Zhao}}, \bibinfo {author} {\bibfnamefont {T.}~\bibnamefont {Yang}}, \bibinfo
		{author} {\bibfnamefont {Y.-A.}\ \bibnamefont {Chen}}, \bibinfo {author}
		{\bibfnamefont {A.-N.}\ \bibnamefont {Zhang}}, \bibinfo {author}
		{\bibfnamefont {M.}~\bibnamefont {\ifmmode~\dot{Z}\else \.{Z}\fi{}ukowski}},\
		and\ \bibinfo {author} {\bibfnamefont {J.-W.}\ \bibnamefont {Pan}},\
	}\bibfield  {title} {\bibinfo {title} {Experimental violation of local
			realism by four-photon greenberger-horne-zeilinger entanglement},\ }\href
	{https://doi.org/10.1103/PhysRevLett.91.180401} {\bibfield  {journal}
		{\bibinfo  {journal} {Phys. Rev. Lett.}\ }\textbf {\bibinfo {volume} {91}},\
		\bibinfo {pages} {180401} (\bibinfo {year} {2003})}\BibitemShut {NoStop}%
	\bibitem [{\citenamefont {Chen}\ \emph {et~al.}(2016)\citenamefont {Chen},
		\citenamefont {Wu}, \citenamefont {Su}, \citenamefont {Cai}, \citenamefont
		{Wang}, \citenamefont {Yang}, \citenamefont {Li}, \citenamefont {Liu},
		\citenamefont {Lu},\ and\ \citenamefont {Pan}}]{Chen2016}%
	\BibitemOpen
	\bibfield  {author} {\bibinfo {author} {\bibfnamefont {M.-C.}\ \bibnamefont
			{Chen}}, \bibinfo {author} {\bibfnamefont {D.}~\bibnamefont {Wu}}, \bibinfo
		{author} {\bibfnamefont {Z.-E.}\ \bibnamefont {Su}}, \bibinfo {author}
		{\bibfnamefont {X.-D.}\ \bibnamefont {Cai}}, \bibinfo {author} {\bibfnamefont
			{X.-L.}\ \bibnamefont {Wang}}, \bibinfo {author} {\bibfnamefont
			{T.}~\bibnamefont {Yang}}, \bibinfo {author} {\bibfnamefont {L.}~\bibnamefont
			{Li}}, \bibinfo {author} {\bibfnamefont {N.-L.}\ \bibnamefont {Liu}},
		\bibinfo {author} {\bibfnamefont {C.-Y.}\ \bibnamefont {Lu}},\ and\ \bibinfo
		{author} {\bibfnamefont {J.-W.}\ \bibnamefont {Pan}},\ }\bibfield  {title}
	{\bibinfo {title} {Efficient measurement of multiparticle entanglement with
			embedding quantum simulator},\ }\href
	{https://doi.org/10.1103/PhysRevLett.116.070502} {\bibfield  {journal}
		{\bibinfo  {journal} {Phys. Rev. Lett.}\ }\textbf {\bibinfo {volume} {116}},\
		\bibinfo {pages} {070502} (\bibinfo {year} {2016})}\BibitemShut {NoStop}%
	\bibitem [{\citenamefont {Loredo}\ \emph {et~al.}(2016)\citenamefont {Loredo},
		\citenamefont {Almeida}, \citenamefont {Di~Candia}, \citenamefont
		{Pedernales}, \citenamefont {Casanova}, \citenamefont {Solano},\ and\
		\citenamefont {White}}]{Loredo2016}%
	\BibitemOpen
	\bibfield  {author} {\bibinfo {author} {\bibfnamefont {J.~C.}\ \bibnamefont
			{Loredo}}, \bibinfo {author} {\bibfnamefont {M.~P.}\ \bibnamefont {Almeida}},
		\bibinfo {author} {\bibfnamefont {R.}~\bibnamefont {Di~Candia}}, \bibinfo
		{author} {\bibfnamefont {J.~S.}\ \bibnamefont {Pedernales}}, \bibinfo
		{author} {\bibfnamefont {J.}~\bibnamefont {Casanova}}, \bibinfo {author}
		{\bibfnamefont {E.}~\bibnamefont {Solano}},\ and\ \bibinfo {author}
		{\bibfnamefont {A.~G.}\ \bibnamefont {White}},\ }\bibfield  {title} {\bibinfo
		{title} {Measuring entanglement in a photonic embedding quantum simulator},\
	}\href {https://doi.org/10.1103/PhysRevLett.116.070503} {\bibfield  {journal}
		{\bibinfo  {journal} {Phys. Rev. Lett.}\ }\textbf {\bibinfo {volume} {116}},\
		\bibinfo {pages} {070503} (\bibinfo {year} {2016})}\BibitemShut {NoStop}%
	\bibitem [{\citenamefont {Singh}\ \emph {et~al.}(2018)\citenamefont {Singh},
		\citenamefont {Dorai},\ and\ \citenamefont {Arvind}}]{Singh_2018}%
	\BibitemOpen
	\bibfield  {author} {\bibinfo {author} {\bibfnamefont {A.}~\bibnamefont
			{Singh}}, \bibinfo {author} {\bibfnamefont {K.}~\bibnamefont {Dorai}},\ and\
		\bibinfo {author} {\bibnamefont {Arvind}},\ }\bibfield  {title} {\bibinfo
		{title} {Experimentally identifying the entanglement class of pure tripartite
			states},\ }\bibfield  {journal} {\bibinfo  {journal} {Quan. Inf. Process.}\
	}\textbf {\bibinfo {volume} {17}},\ \href
	{https://doi.org/10.1007/s11128-018-2105-5} {10.1007/s11128-018-2105-5}
	(\bibinfo {year} {2018})\BibitemShut {NoStop}%
	\bibitem [{\citenamefont {Xue}\ \emph {et~al.}(2022)\citenamefont {Xue},
		\citenamefont {Huang}, \citenamefont {Zhao}, \citenamefont {Wei},
		\citenamefont {Li}, \citenamefont {Dong}, \citenamefont {Gao}, \citenamefont
		{Lu}, \citenamefont {Xin},\ and\ \citenamefont {Long}}]{xue2022experimental}%
	\BibitemOpen
	\bibfield  {author} {\bibinfo {author} {\bibfnamefont {S.}~\bibnamefont
			{Xue}}, \bibinfo {author} {\bibfnamefont {Y.}~\bibnamefont {Huang}}, \bibinfo
		{author} {\bibfnamefont {D.}~\bibnamefont {Zhao}}, \bibinfo {author}
		{\bibfnamefont {C.}~\bibnamefont {Wei}}, \bibinfo {author} {\bibfnamefont
			{J.}~\bibnamefont {Li}}, \bibinfo {author} {\bibfnamefont {Y.}~\bibnamefont
			{Dong}}, \bibinfo {author} {\bibfnamefont {J.}~\bibnamefont {Gao}}, \bibinfo
		{author} {\bibfnamefont {D.}~\bibnamefont {Lu}}, \bibinfo {author}
		{\bibfnamefont {T.}~\bibnamefont {Xin}},\ and\ \bibinfo {author}
		{\bibfnamefont {G.-L.}\ \bibnamefont {Long}},\ }\bibfield  {title} {\bibinfo
		{title} {Experimental measurement of bipartite entanglement using
			parameterized quantum circuits},\ }\href
	{https://doi.org/10.1007/s11433-022-1904-3} {\bibfield  {journal} {\bibinfo
			{journal} {Sci. China Phys. Mech. Astron.}\ }\textbf {\bibinfo {volume}
			{65}},\ \bibinfo {pages} {280312} (\bibinfo {year} {2022})}\BibitemShut
	{NoStop}%
	\bibitem [{\citenamefont {Xie}\ \emph {et~al.}(2021)\citenamefont {Xie},
		\citenamefont {Zhao}, \citenamefont {Kong}, \citenamefont {Ma}, \citenamefont
		{Wang}, \citenamefont {Ye}, \citenamefont {Yu}, \citenamefont {Yang},
		\citenamefont {Xu}, \citenamefont {Wang}, \citenamefont {Wang}, \citenamefont
		{Shi},\ and\ \citenamefont {Du}}]{Xie2021}%
	\BibitemOpen
	\bibfield  {author} {\bibinfo {author} {\bibfnamefont {T.}~\bibnamefont
			{Xie}}, \bibinfo {author} {\bibfnamefont {Z.}~\bibnamefont {Zhao}}, \bibinfo
		{author} {\bibfnamefont {X.}~\bibnamefont {Kong}}, \bibinfo {author}
		{\bibfnamefont {W.}~\bibnamefont {Ma}}, \bibinfo {author} {\bibfnamefont
			{M.}~\bibnamefont {Wang}}, \bibinfo {author} {\bibfnamefont {X.}~\bibnamefont
			{Ye}}, \bibinfo {author} {\bibfnamefont {P.}~\bibnamefont {Yu}}, \bibinfo
		{author} {\bibfnamefont {Z.}~\bibnamefont {Yang}}, \bibinfo {author}
		{\bibfnamefont {S.}~\bibnamefont {Xu}}, \bibinfo {author} {\bibfnamefont
			{P.}~\bibnamefont {Wang}}, \bibinfo {author} {\bibfnamefont {Y.}~\bibnamefont
			{Wang}}, \bibinfo {author} {\bibfnamefont {F.}~\bibnamefont {Shi}},\ and\
		\bibinfo {author} {\bibfnamefont {J.}~\bibnamefont {Du}},\ }\bibfield
	{title} {\bibinfo {title} {Beating the standard quantum limit under ambient
			conditions with solid-state spins},\ }\href
	{https://doi.org/10.1126/sciadv.abg9204} {\bibfield  {journal} {\bibinfo
			{journal} {Sci. Adv.}\ }\textbf {\bibinfo {volume} {7}},\ \bibinfo {pages}
		{eabg9204} (\bibinfo {year} {2021})}\BibitemShut {NoStop}%
	\bibitem [{\citenamefont {Andolina}\ \emph {et~al.}(2019)\citenamefont
		{Andolina}, \citenamefont {Keck}, \citenamefont {Mari}, \citenamefont
		{Campisi}, \citenamefont {Giovannetti},\ and\ \citenamefont
		{Polini}}]{Andolina2019}%
	\BibitemOpen
	\bibfield  {author} {\bibinfo {author} {\bibfnamefont {G.~M.}\ \bibnamefont
			{Andolina}}, \bibinfo {author} {\bibfnamefont {M.}~\bibnamefont {Keck}},
		\bibinfo {author} {\bibfnamefont {A.}~\bibnamefont {Mari}}, \bibinfo {author}
		{\bibfnamefont {M.}~\bibnamefont {Campisi}}, \bibinfo {author} {\bibfnamefont
			{V.}~\bibnamefont {Giovannetti}},\ and\ \bibinfo {author} {\bibfnamefont
			{M.}~\bibnamefont {Polini}},\ }\bibfield  {title} {\bibinfo {title}
		{Extractable work, the role of correlations, and asymptotic freedom in
			quantum batteries},\ }\href {https://doi.org/10.1103/PhysRevLett.122.047702}
	{\bibfield  {journal} {\bibinfo  {journal} {Phys. Rev. Lett.}\ }\textbf
		{\bibinfo {volume} {122}},\ \bibinfo {pages} {047702} (\bibinfo {year}
		{2019})}\BibitemShut {NoStop}%
	\bibitem [{\citenamefont {Francica}\ \emph {et~al.}(2020)\citenamefont
		{Francica}, \citenamefont {Binder}, \citenamefont {Guarnieri}, \citenamefont
		{Mitchison}, \citenamefont {Goold},\ and\ \citenamefont
		{Plastina}}]{Francica2020}%
	\BibitemOpen
	\bibfield  {author} {\bibinfo {author} {\bibfnamefont {G.}~\bibnamefont
			{Francica}}, \bibinfo {author} {\bibfnamefont {F.~C.}\ \bibnamefont
			{Binder}}, \bibinfo {author} {\bibfnamefont {G.}~\bibnamefont {Guarnieri}},
		\bibinfo {author} {\bibfnamefont {M.~T.}\ \bibnamefont {Mitchison}}, \bibinfo
		{author} {\bibfnamefont {J.}~\bibnamefont {Goold}},\ and\ \bibinfo {author}
		{\bibfnamefont {F.}~\bibnamefont {Plastina}},\ }\bibfield  {title} {\bibinfo
		{title} {Quantum coherence and ergotropy},\ }\href
	{https://doi.org/10.1103/PhysRevLett.125.180603} {\bibfield  {journal}
		{\bibinfo  {journal} {Phys. Rev. Lett.}\ }\textbf {\bibinfo {volume} {125}},\
		\bibinfo {pages} {180603} (\bibinfo {year} {2020})}\BibitemShut {NoStop}%
	\bibitem [{\citenamefont {Tirone}\ \emph {et~al.}(2021)\citenamefont {Tirone},
		\citenamefont {Salvia},\ and\ \citenamefont {Giovannetti}}]{Tirone2021}%
	\BibitemOpen
	\bibfield  {author} {\bibinfo {author} {\bibfnamefont {S.}~\bibnamefont
			{Tirone}}, \bibinfo {author} {\bibfnamefont {R.}~\bibnamefont {Salvia}},\
		and\ \bibinfo {author} {\bibfnamefont {V.}~\bibnamefont {Giovannetti}},\
	}\bibfield  {title} {\bibinfo {title} {Quantum energy lines and the optimal
			output ergotropy problem},\ }\href
	{https://doi.org/10.1103/PhysRevLett.127.210601} {\bibfield  {journal}
		{\bibinfo  {journal} {Phys. Rev. Lett.}\ }\textbf {\bibinfo {volume} {127}},\
		\bibinfo {pages} {210601} (\bibinfo {year} {2021})}\BibitemShut {NoStop}%
	\bibitem [{\citenamefont {\ifmmode~\check{S}\else \v{S}\fi{}afr\'anek}\ \emph
		{et~al.}(2023)\citenamefont {\ifmmode~\check{S}\else \v{S}\fi{}afr\'anek},
		\citenamefont {Rosa},\ and\ \citenamefont {Binder}}]{Safranek2023}%
	\BibitemOpen
	\bibfield  {author} {\bibinfo {author} {\bibfnamefont {D.}~\bibnamefont
			{\ifmmode~\check{S}\else \v{S}\fi{}afr\'anek}}, \bibinfo {author}
		{\bibfnamefont {D.}~\bibnamefont {Rosa}},\ and\ \bibinfo {author}
		{\bibfnamefont {F.~C.}\ \bibnamefont {Binder}},\ }\bibfield  {title}
	{\bibinfo {title} {Work extraction from unknown quantum sources},\ }\href
	{https://doi.org/10.1103/PhysRevLett.130.210401} {\bibfield  {journal}
		{\bibinfo  {journal} {Phys. Rev. Lett.}\ }\textbf {\bibinfo {volume} {130}},\
		\bibinfo {pages} {210401} (\bibinfo {year} {2023})}\BibitemShut {NoStop}%
	\bibitem [{\citenamefont {Frey}\ \emph {et~al.}(2014)\citenamefont {Frey},
		\citenamefont {Funo},\ and\ \citenamefont {Hotta}}]{Frey2015}%
	\BibitemOpen
	\bibfield  {author} {\bibinfo {author} {\bibfnamefont {M.}~\bibnamefont
			{Frey}}, \bibinfo {author} {\bibfnamefont {K.}~\bibnamefont {Funo}},\ and\
		\bibinfo {author} {\bibfnamefont {M.}~\bibnamefont {Hotta}},\ }\bibfield
	{title} {\bibinfo {title} {Strong local passivity in finite quantum
			systems},\ }\href {https://doi.org/10.1103/PhysRevE.90.012127} {\bibfield
		{journal} {\bibinfo  {journal} {Phys. Rev. E}\ }\textbf {\bibinfo {volume}
			{90}},\ \bibinfo {pages} {012127} (\bibinfo {year} {2014})}\BibitemShut
	{NoStop}%
	\bibitem [{\citenamefont {Alhambra}\ \emph {et~al.}(2019)\citenamefont
		{Alhambra}, \citenamefont {Styliaris}, \citenamefont
		{Rodr\'{\i}guez-Briones}, \citenamefont {Sikora},\ and\ \citenamefont
		{Mart\'{\i}n-Mart\'{\i}nez}}]{Alhambra2019}%
	\BibitemOpen
	\bibfield  {author} {\bibinfo {author} {\bibfnamefont {A.~M.}\ \bibnamefont
			{Alhambra}}, \bibinfo {author} {\bibfnamefont {G.}~\bibnamefont {Styliaris}},
		\bibinfo {author} {\bibfnamefont {N.~A.}\ \bibnamefont
			{Rodr\'{\i}guez-Briones}}, \bibinfo {author} {\bibfnamefont {J.}~\bibnamefont
			{Sikora}},\ and\ \bibinfo {author} {\bibfnamefont {E.}~\bibnamefont
			{Mart\'{\i}n-Mart\'{\i}nez}},\ }\bibfield  {title} {\bibinfo {title}
		{Fundamental limitations to local energy extraction in quantum systems},\
	}\href {https://doi.org/10.1103/PhysRevLett.123.190601} {\bibfield  {journal}
		{\bibinfo  {journal} {Phys. Rev. Lett.}\ }\textbf {\bibinfo {volume} {123}},\
		\bibinfo {pages} {190601} (\bibinfo {year} {2019})}\BibitemShut {NoStop}%
	\bibitem [{\citenamefont {Rodr\'{\i}guez-Briones}\ \emph
		{et~al.}(2023)\citenamefont {Rodr\'{\i}guez-Briones}, \citenamefont
		{Katiyar}, \citenamefont {Mart\'{\i}n-Mart\'{\i}nez},\ and\ \citenamefont
		{Laflamme}}]{Katiyar2023}%
	\BibitemOpen
	\bibfield  {author} {\bibinfo {author} {\bibfnamefont {N.~A.}\ \bibnamefont
			{Rodr\'{\i}guez-Briones}}, \bibinfo {author} {\bibfnamefont {H.}~\bibnamefont
			{Katiyar}}, \bibinfo {author} {\bibfnamefont {E.}~\bibnamefont
			{Mart\'{\i}n-Mart\'{\i}nez}},\ and\ \bibinfo {author} {\bibfnamefont
			{R.}~\bibnamefont {Laflamme}},\ }\bibfield  {title} {\bibinfo {title}
		{Experimental activation of strong local passive states with quantum
			information},\ }\href {https://doi.org/10.1103/PhysRevLett.130.110801}
	{\bibfield  {journal} {\bibinfo  {journal} {Phys. Rev. Lett.}\ }\textbf
		{\bibinfo {volume} {130}},\ \bibinfo {pages} {110801} (\bibinfo {year}
		{2023})}\BibitemShut {NoStop}%
	\bibitem [{\citenamefont {Marshall}\ \emph {et~al.}(2011)\citenamefont
		{Marshall}, \citenamefont {Olkin},\ and\ \citenamefont
		{Arnold}}]{Marshall1979}%
	\BibitemOpen
	\bibfield  {author} {\bibinfo {author} {\bibfnamefont {A.~W.}\ \bibnamefont
			{Marshall}}, \bibinfo {author} {\bibfnamefont {I.}~\bibnamefont {Olkin}},\
		and\ \bibinfo {author} {\bibfnamefont {B.~C.}\ \bibnamefont {Arnold}},\
	}\href {https://doi.org/10.1007/978-0-387-68276-1} {\emph {\bibinfo {title}
			{Inequalities: Theory of Majorization and its Applications}}},\ \bibinfo
	{edition} {2nd}\ ed.,\ Vol.\ \bibinfo {volume} {143}\ (\bibinfo  {publisher}
	{Springer},\ \bibinfo {year} {2011})\BibitemShut {NoStop}%
	\bibitem [{\citenamefont {Nielsen}(1999)}]{Nielsen1999}%
	\BibitemOpen
	\bibfield  {author} {\bibinfo {author} {\bibfnamefont {M.~A.}\ \bibnamefont
			{Nielsen}},\ }\bibfield  {title} {\bibinfo {title} {Conditions for a class of
			entanglement transformations},\ }\href
	{https://doi.org/10.1103/PhysRevLett.83.436} {\bibfield  {journal} {\bibinfo
			{journal} {Phys. Rev. Lett.}\ }\textbf {\bibinfo {volume} {83}},\ \bibinfo
		{pages} {436} (\bibinfo {year} {1999})}\BibitemShut {NoStop}%
	\bibitem [{\citenamefont {Horodecki}\ \emph {et~al.}(2003)\citenamefont
		{Horodecki}, \citenamefont {Horodecki},\ and\ \citenamefont
		{Oppenheim}}]{Horodecki2003}%
	\BibitemOpen
	\bibfield  {author} {\bibinfo {author} {\bibfnamefont {M.}~\bibnamefont
			{Horodecki}}, \bibinfo {author} {\bibfnamefont {P.}~\bibnamefont
			{Horodecki}},\ and\ \bibinfo {author} {\bibfnamefont {J.}~\bibnamefont
			{Oppenheim}},\ }\bibfield  {title} {\bibinfo {title} {Reversible
			transformations from pure to mixed states and the unique measure of
			information},\ }\href {https://doi.org/10.1103/PhysRevA.67.062104} {\bibfield
		{journal} {\bibinfo  {journal} {Phys. Rev. A}\ }\textbf {\bibinfo {volume}
			{67}},\ \bibinfo {pages} {062104} (\bibinfo {year} {2003})}\BibitemShut
	{NoStop}%
	\bibitem [{\citenamefont {Winter}\ and\ \citenamefont
		{Yang}(2016)}]{Winter2016}%
	\BibitemOpen
	\bibfield  {author} {\bibinfo {author} {\bibfnamefont {A.}~\bibnamefont
			{Winter}}\ and\ \bibinfo {author} {\bibfnamefont {D.}~\bibnamefont {Yang}},\
	}\bibfield  {title} {\bibinfo {title} {Operational resource theory of
			coherence},\ }\href {https://doi.org/10.1103/PhysRevLett.116.120404}
	{\bibfield  {journal} {\bibinfo  {journal} {Phys. Rev. Lett.}\ }\textbf
		{\bibinfo {volume} {116}},\ \bibinfo {pages} {120404} (\bibinfo {year}
		{2016})}\BibitemShut {NoStop}%
	\bibitem [{\citenamefont {Horodecki}\ and\ \citenamefont
		{Oppenheim}(2013)}]{Horodecki2013}%
	\BibitemOpen
	\bibfield  {author} {\bibinfo {author} {\bibfnamefont {M.}~\bibnamefont
			{Horodecki}}\ and\ \bibinfo {author} {\bibfnamefont {J.}~\bibnamefont
			{Oppenheim}},\ }\bibfield  {title} {\bibinfo {title} {Fundamental limitations
			for quantum and nanoscale thermodynamics},\ }\href
	{https://doi.org/10.1038/ncomms3059} {\bibfield  {journal} {\bibinfo
			{journal} {Nature communications}\ }\textbf {\bibinfo {volume} {4}},\
		\bibinfo {pages} {2059} (\bibinfo {year} {2013})}\BibitemShut {NoStop}%
	\bibitem [{\citenamefont {Bennett}\ \emph {et~al.}(1999)\citenamefont
		{Bennett}, \citenamefont {DiVincenzo}, \citenamefont {Mor}, \citenamefont
		{Shor}, \citenamefont {Smolin},\ and\ \citenamefont {Terhal}}]{Bennett1999}%
	\BibitemOpen
	\bibfield  {author} {\bibinfo {author} {\bibfnamefont {C.~H.}\ \bibnamefont
			{Bennett}}, \bibinfo {author} {\bibfnamefont {D.~P.}\ \bibnamefont
			{DiVincenzo}}, \bibinfo {author} {\bibfnamefont {T.}~\bibnamefont {Mor}},
		\bibinfo {author} {\bibfnamefont {P.~W.}\ \bibnamefont {Shor}}, \bibinfo
		{author} {\bibfnamefont {J.~A.}\ \bibnamefont {Smolin}},\ and\ \bibinfo
		{author} {\bibfnamefont {B.~M.}\ \bibnamefont {Terhal}},\ }\bibfield  {title}
	{\bibinfo {title} {Unextendible product bases and bound entanglement},\
	}\href {https://doi.org/10.1103/PhysRevLett.82.5385} {\bibfield  {journal}
		{\bibinfo  {journal} {Phys. Rev. Lett.}\ }\textbf {\bibinfo {volume} {82}},\
		\bibinfo {pages} {5385} (\bibinfo {year} {1999})}\BibitemShut {NoStop}%
	\bibitem [{\citenamefont {D\"ur}\ and\ \citenamefont {Cirac}(2000)}]{Dur00}%
	\BibitemOpen
	\bibfield  {author} {\bibinfo {author} {\bibfnamefont {W.}~\bibnamefont
			{D\"ur}}\ and\ \bibinfo {author} {\bibfnamefont {J.~I.}\ \bibnamefont
			{Cirac}},\ }\bibfield  {title} {\bibinfo {title} {Classification of
			multiqubit mixed states: Separability and distillability properties},\ }\href
	{https://doi.org/10.1103/PhysRevA.61.042314} {\bibfield  {journal} {\bibinfo
			{journal} {Phys. Rev. A}\ }\textbf {\bibinfo {volume} {61}},\ \bibinfo
		{pages} {042314} (\bibinfo {year} {2000})}\BibitemShut {NoStop}%
	\bibitem [{\citenamefont {Gühne}\ and\ \citenamefont
		{Seevinck}(2010)}]{Guhne2010}%
	\BibitemOpen
	\bibfield  {author} {\bibinfo {author} {\bibfnamefont {O.}~\bibnamefont
			{Gühne}}\ and\ \bibinfo {author} {\bibfnamefont {M.}~\bibnamefont
			{Seevinck}},\ }\bibfield  {title} {\bibinfo {title} {Separability criteria
			for genuine multiparticle entanglement},\ }\href
	{https://doi.org/10.1088/1367-2630/12/5/053002} {\bibfield  {journal}
		{\bibinfo  {journal} {New J. Phys.}\ }\textbf {\bibinfo {volume} {12}},\
		\bibinfo {pages} {053002} (\bibinfo {year} {2010})}\BibitemShut {NoStop}%
	\bibitem [{\citenamefont {Cavanagh}\ \emph {et~al.}(1996)\citenamefont
		{Cavanagh}, \citenamefont {Fairbrother}, \citenamefont {Palmer~III},\ and\
		\citenamefont {Skelton}}]{cavanagh1996protein}%
	\BibitemOpen
	\bibfield  {author} {\bibinfo {author} {\bibfnamefont {J.}~\bibnamefont
			{Cavanagh}}, \bibinfo {author} {\bibfnamefont {W.~J.}\ \bibnamefont
			{Fairbrother}}, \bibinfo {author} {\bibfnamefont {A.~G.}\ \bibnamefont
			{Palmer~III}},\ and\ \bibinfo {author} {\bibfnamefont {N.~J.}\ \bibnamefont
			{Skelton}},\ }\href {https://doi.org/10.1016/B978-0-12-164491-8.X5000-3}
	{\emph {\bibinfo {title} {Protein NMR spectroscopy: principles and
				practice}}}\ (\bibinfo  {publisher} {Academic press},\ \bibinfo {year}
	{1996})\BibitemShut {NoStop}%
	\bibitem [{\citenamefont {Riedel~G{\aa}rding}\ \emph
		{et~al.}(2021)\citenamefont {Riedel~G{\aa}rding}, \citenamefont {Schwaller},
		\citenamefont {Chan}, \citenamefont {Chang}, \citenamefont {Bosch},
		\citenamefont {Gessler}, \citenamefont {Laborde}, \citenamefont {Hernandez},
		\citenamefont {Si}, \citenamefont {Dupertuis} \emph
		{et~al.}}]{riedel2021bell}%
	\BibitemOpen
	\bibfield  {author} {\bibinfo {author} {\bibfnamefont {E.}~\bibnamefont
			{Riedel~G{\aa}rding}}, \bibinfo {author} {\bibfnamefont {N.}~\bibnamefont
			{Schwaller}}, \bibinfo {author} {\bibfnamefont {C.~L.}\ \bibnamefont {Chan}},
		\bibinfo {author} {\bibfnamefont {S.~Y.}\ \bibnamefont {Chang}}, \bibinfo
		{author} {\bibfnamefont {S.}~\bibnamefont {Bosch}}, \bibinfo {author}
		{\bibfnamefont {F.}~\bibnamefont {Gessler}}, \bibinfo {author} {\bibfnamefont
			{W.~R.}\ \bibnamefont {Laborde}}, \bibinfo {author} {\bibfnamefont {J.~N.}\
			\bibnamefont {Hernandez}}, \bibinfo {author} {\bibfnamefont {X.}~\bibnamefont
			{Si}}, \bibinfo {author} {\bibfnamefont {M.}~\bibnamefont {Dupertuis}}, \emph
		{et~al.},\ }\bibfield  {title} {\bibinfo {title} {Bell diagonal and werner
			state generation: Entanglement, non-locality, steering and discord on the ibm
			quantum computer},\ }\href {https://doi.org/10.3390/e23070797} {\bibfield
		{journal} {\bibinfo  {journal} {Entropy}\ }\textbf {\bibinfo {volume} {23}},\
		\bibinfo {pages} {797} (\bibinfo {year} {2021})}\BibitemShut {NoStop}%
	\bibitem [{\citenamefont {Roy}\ and\ \citenamefont
		{Mahesh}(2010)}]{roy2010density}%
	\BibitemOpen
	\bibfield  {author} {\bibinfo {author} {\bibfnamefont {S.~S.}\ \bibnamefont
			{Roy}}\ and\ \bibinfo {author} {\bibfnamefont {T.}~\bibnamefont {Mahesh}},\
	}\bibfield  {title} {\bibinfo {title} {Density matrix tomography of singlet
			states},\ }\href {https://doi.org/10.1016/j.jmr.2010.06.014} {\bibfield
		{journal} {\bibinfo  {journal} {J. Magn. Reson.}\ }\textbf {\bibinfo {volume}
			{206}},\ \bibinfo {pages} {127} (\bibinfo {year} {2010})}\BibitemShut
	{NoStop}%
	\bibitem [{\citenamefont {Carravetta}\ and\ \citenamefont
		{Levitt}(2004)}]{carravetta2004long}%
	\BibitemOpen
	\bibfield  {author} {\bibinfo {author} {\bibfnamefont {M.}~\bibnamefont
			{Carravetta}}\ and\ \bibinfo {author} {\bibfnamefont {M.~H.}\ \bibnamefont
			{Levitt}},\ }\bibfield  {title} {\bibinfo {title} {Long-lived nuclear spin
			states in high-field solution {NMR}},\ }\href
	{https://doi.org/10.1021/ja0490931} {\bibfield  {journal} {\bibinfo
			{journal} {J. Am. Chem. Soc.}\ }\textbf {\bibinfo {volume} {126}},\ \bibinfo
		{pages} {6228} (\bibinfo {year} {2004})}\BibitemShut {NoStop}%
	\bibitem [{\citenamefont {Krithika}\ \emph {et~al.}(2021)\citenamefont
		{Krithika}, \citenamefont {Pal}, \citenamefont {Nath},\ and\ \citenamefont
		{Mahesh}}]{krithika2021observation}%
	\BibitemOpen
	\bibfield  {author} {\bibinfo {author} {\bibfnamefont {V.}~\bibnamefont
			{Krithika}}, \bibinfo {author} {\bibfnamefont {S.}~\bibnamefont {Pal}},
		\bibinfo {author} {\bibfnamefont {R.}~\bibnamefont {Nath}},\ and\ \bibinfo
		{author} {\bibfnamefont {T.}~\bibnamefont {Mahesh}},\ }\bibfield  {title}
	{\bibinfo {title} {Observation of interaction induced blockade and local spin
			freezing in a nmr quantum simulator},\ }\href
	{https://doi.org/10.1103/PhysRevResearch.3.033035} {\bibfield  {journal}
		{\bibinfo  {journal} {Phys. Rev. Research}\ }\textbf {\bibinfo {volume}
			{3}},\ \bibinfo {pages} {033035} (\bibinfo {year} {2021})}\BibitemShut
	{NoStop}%
\end{thebibliography}
\end{document}